\documentclass[11pt]{article}
\usepackage[margin=1in]{geometry}
\usepackage{amsmath,bm}%
\usepackage{amsfonts}%
\usepackage{amssymb}%
\usepackage{amsthm}
\usepackage{graphicx}
\usepackage{hyperref}
\usepackage{url}
\usepackage{enumitem}
\usepackage{color}
\usepackage{colortbl}
\usepackage{xcolor}
\usepackage{algpseudocode}
\usepackage{algorithmicx}
\usepackage{algorithm}

\usepackage[%
    font={small,sf},
    labelfont=bf,
    format=hang,    
    format=plain,
    margin=0pt,
    width=0.8\textwidth,
]{caption}
\usepackage[list=true]{subcaption}

\newcommand{\hemant}[1]{\textcolor{black}{#1}}
\newcommand{\steph}[1]{\textcolor{black}{#1}}

\newtheorem{theorem}{Theorem}

\newtheorem{claim}{Claim}

\newtheorem{corollary}{Corollary}

\newtheorem{definition}{Definition}

\newtheorem{lemma}{Lemma}

\newtheorem{proposition}{Proposition}
\newtheorem{remark}{Remark}

\numberwithin{equation}{section}

\newcommand{\norm}[1]{\|{#1}\|}
\newcommand{\abs}[1]{|{#1}|}

\newcommand{\set}[1]{\left\{{#1}\right\}}

\newcommand{\est}[1]{\widehat{#1}}

\newcommand{\matR}{\ensuremath{\mathbb{R}}}
\newcommand{\matZ}{\ensuremath{\mathbb{Z}}}

\newcommand{\ftil}{\ensuremath{\widetilde{f}}}

\newcommand{\noiseint}{\ensuremath{\eta'}}

\newcommand{\Ahat}{\ensuremath{\est{A}}}
\newcommand{\Bhat}{\ensuremath{\est{B}}}
\newcommand{\Uhat}{\ensuremath{\est{U}}}
\newcommand{\noisemax}{\ensuremath{\eta_{\mathrm{max}}}}
\newcommand{\umax}{\ensuremath{u_{\max}}}
\newcommand{\umin}{\ensuremath{u_{\min}}}
\newcommand{\urel}{\ensuremath{u_{\mathrm{rel}}}}

\newcommand{\condnum}{\ensuremath{\kappa}}
\newcommand{\sep}{\ensuremath{\triangle}}
\newcommand{\sigmamin}{\ensuremath{\sigma_{\text{min}}}}
\newcommand{\sigmamax}{\ensuremath{\sigma_{\text{max}}}}
\newcommand{\chord}{\ensuremath{\chi}}
\newcommand{\mdist}{\ensuremath{\text{md}_{\chord}}}
\newcommand{\kerfn}{\ensuremath{g}}
\newcommand{\kerfnvar}{\ensuremath{\mu}}
\newcommand{\kerfour}{\ensuremath{\bar{g}}}
\newcommand{\perm}{\ensuremath{\phi}}
\newcommand{\diag}{\text{diag}}

\newcommand{\mup}[1]{\ensuremath{M_{{#1},\mathrm{up}}}}
\newcommand{\fnoise}{\ensuremath{\ftil}}

\newcommand\convloin{\hspace{3mm}\mathrel{\mathop{\kern 0pt \put(-14,3){\vector(1,0){28}}}
		\limits_{n\to +\infty}^{loi}}\hspace{3mm}}

\newcommand\convpst{\hspace{3mm}\mathrel{\mathop{\kern 0pt \put(-14,3){\vector(1,0){28}}}
		\limits_{t\to +\infty}^{p.s.}}\hspace{3mm}}

\title{Multi-kernel unmixing and super-resolution using the Modified Matrix Pencil method}

\author{St\'ephane Chr\'etien\thanks{National Physical Laboratory, Teddington, UK and the Alan Turing Institure, London, UK.}\\ \texttt{stephane.chretien@npl.co.uk}
\and Hemant Tyagi\thanks{\hemant{INRIA Lille-Nord Europe, France. This work was done by the author while affiliated to the 
Alan Turing Institute, London, UK, and School of Mathematics, University of Edinburgh, Edinburgh, UK. 
This author's work was supported by EPSRC grant EP/N510129/1.}}\\ \texttt{hemant.tyagi@inria.fr}} 

\begin{document}
\maketitle

\begin{abstract}
Consider $L$ groups of point sources or spike trains, with the $l^{\text{th}}$ group represented by $x_l(t)$.
For a function $g:\matR \rightarrow \matR$, let $g_l(t) = g(t/\kerfnvar_l)$ denote a point spread function with 
scale $\kerfnvar_l > 0$, and with $\kerfnvar_1 < \cdots < \kerfnvar_L$. 
With $y(t) = \sum_{l=1}^{L} (g_l \star x_l)(t)$, our goal is to recover the 
source parameters given samples of $y$, or given the Fourier samples of $y$. This problem is a generalization of 
the usual super-resolution setup wherein $L = 1$; we call this the multi-kernel unmixing super-resolution problem. 
Assuming access to Fourier samples of $y$, we derive an algorithm for this problem for estimating the source parameters 
of each group, along with precise non-asymptotic guarantees. Our approach involves estimating the group parameters sequentially 
in the order of increasing scale parameters, i.e., from group $1$ to $L$. In particular, the estimation process at stage 
$1 \leq l \leq L$ involves (i) carefully sampling the tail of the Fourier transform of $y$, 
(ii) a \emph{deflation} step wherein we subtract the contribution of the groups processed thus far from the obtained Fourier samples, 
and (iii) applying Moitra's modified Matrix Pencil method on a deconvolved version of the samples in (ii). 
\end{abstract}

{\noindent \bf Key words:} Matrix Pencil, super-resolution, unmixing kernels, mixture models, sampling, approximation, signal recovery. \\\\
{\noindent \bf Mathematics Subject Classifications (2010):} 15B05, 42A82, 65T99, 65F15, 94A20. \\

\section{Introduction} \label{sec:intro}
\subsection{Background on super-resolution}
Super-resolution consists of estimating a signal $x$, given blurred observations 
obtained after convolution with a point spread function $g$ which is assumed to represent the 
impulse response of the measurement system, such as for e.g., a microscope in high density 
single molecule imaging. Mathematically, $x$ is typically modeled as a superposition of $K$ Dirac's, 
i.e., a \emph{sparse} atomic measure of the form
\begin{equation*} 
x(t) = \sum_{i=1}^K u_i \delta(t-t_i); \quad u_i \in \mathbb{C}, t_i \in [0,1),
\end{equation*}
while $g$ is a low pass filter. Denoting 
\begin{equation} \label{eq:y_conv_exp}
y(t) = \sum_{i=1}^K u_i g(t-t_i)
\end{equation} 
to be the convolution of $x$ and $g$, one is usually given information about $x$ either as samples of $y$, 
or the Fourier samples of $y$.
This problem has a number of important applications arising for instance in geophysics \cite{khaidukov04}, 
astronomy \cite{puschmann05}, medical imaging \cite{greenspan09} etc. The reader is referred to \cite{candes2014towards} 
for a more comprehensive list of applications.
Super-resolution can be seen as a highly non-trivial ``off the grid'' extension of the finite dimensional sparse estimation problem in compressed sensing \cite{foucart2013mathematical}, \cite{eldar2012compressed} and statistics \cite{buhlmann2011statistics}. In the new setting, instead of estimating a sparse vector in a finite dimensional space, the goal is to estimate a sparse measure over the real line $\mathbb R$ endowed with its Borel $\sigma$-algebra.

Recently, the problem of super-resolution has received a great deal of interest in the signal processing 
community, triggered by the seminal work of Cand\`es and Fernandez-Granda \cite{candes2014towards,candes2013super}. 
They considered the setup where one is given the first few Fourier coefficients of $x$, 
i.e., for a cut-off frequency $m \in \matZ^{+}$, we observe $f(s) \in \mathbb{C}$ where 
\begin{equation} \label{eq:four_meas_model}
f(s) = \int_{0}^{1} e^{\iota 2\pi st} x(dt) = \sum_{i=1}^K u_i e^{\iota 2\pi s t_i}; \quad s \in \set{-m,-m+1,\dots,m}.
\end{equation}
Note that this corresponds to taking $g(t) = 2m\ \text{sinc}(2mt)$ in \eqref{eq:y_conv_exp}, 
and sampling the Fourier transform of $y$ on the regular grid $\set{-m,-m+1,\dots,m}$.

\subsubsection{Total variation and atomic norm-based approaches}

\hemant{In \cite{candes2014towards}, the authors consider the noiseless setting and propose solving an optimization problem over the space of measures which involves 
minimizing the total variation (TV) norm amongst all measures which are consistent with the observations.  
The resulting minimization problem is an infinite dimensional convex program over Borel measures on $\matR$.
It was shown that the dual of this problem can be formulated as a semi-definite program (SDP) 
in finitely many variables, and thus can be solved in polynomial computational time. Since then, there have been numerous 
follow-up works such as by Schiebinger et al. \cite{schiebingersuperresolution}, 
Duval and Peyre \cite{duval2015exact}, Denoyelle et al. \cite{denoyelle2015asymptotic}, 
Bendory et al. \cite{bendory2016robust}, Aza\"is et al. \cite{castro15} and many others. 
For instance, \cite{schiebingersuperresolution} considers the noiseless setting by taking real-valued 
samples of $y$ with a more general choice of $g$ (such as a Gaussian) and also assumes $x$ to be non-negative. 
Their proposed approach again involves TV norm minimization with linear constraints. 
Bendory et al. \cite{bendory2016robust} consider $g$ to be Gaussian or Cauchy, do not place sign assumptions on 
$x$, and also analyze TV norm minimization with linear fidelity constraints for estimating $x$ from noiseless samples of $y$. 
The approach adopted in \cite{duval2015exact,denoyelle2015asymptotic} is to solve a 
least-squares-type minimization procedure with a TV norm based penalty term (also referred to as the Beurling 
LASSO (see for e.g., \cite{castro15})) for recovering $x$ from samples of $y$.  
The approach in \cite{duval2017sparse} considers a natural finite approximation on the grid to the 
continuous problem, and studies the limiting behaviour as the grid becomes finer; see also \cite{duval2017sparse2}.}

\hemant{From a statistical view point, Cand\`es and Fernand\`ez-Granda \cite{candes2014towards} showed that their 
approach exactly recovers $x$ in the noiseless case provided $m \geq 2/\sep$, where $\sep$ denotes the minimum separation 
between the spike locations. Similar results for other choices of $g$ were shown by Schiebinger 
et al. \cite{schiebingersuperresolution} (for positive measures and without 
any minimum separation condition), and by Bendory et al. \cite{bendory2016robust} (for 
signed measures and with a separation condition).   
In the noisy setting, the state of affairs is radically different since it 
is known (see for e.g., \cite[Section 3.2]{candes2014towards}, \cite[Corollary 1.4]{moitra15}) that some separation between the spike locations is indeed necessary 
for stable recovery. When sufficient separation is assumed and provided the noise level is small enough,  
then stable recovery of $x$ is possible (see for e.g., \cite{grandanoisy13,duval2015exact,denoyelle2015asymptotic,castro15}).}

\hemant{Recently also, Tang et al. \cite{tang2013compressed,tang_minimax15} studied approaches based on atomic norm minimization, 
which can be formulated as a SDP. In \cite[Theorem 1.1]{tang2013compressed}, the authors considered the signs of the 
amplitudes of $u_j$ to be generated randomly, with noiseless samples. 
It was shown that if $m \geq 2/\sep$, and if the Fourier samples are obtained at $n = \Omega(K \log K \log m)$ 
indices selected uniformly at random from $\set{-m,\dots,m}$, then $x$ is recovered exactly with high probability.  
In \cite[Thorem 1]{tang_minimax15}, the authors considered Gaussian noise in the samples and showed that if $m \geq 2/\sep$, 
then the solution returned by the SDP estimates the vector of original Fourier samples (i.e., $(f(-m),\dots,f(m))^T$) at a 
mean square rate $O(\sigma^2 K \frac{\log m}{m})$, with $\sigma^2$ denoting variance of noise. Moreover, they also 
show \cite[Theorem 2]{tang_minimax15} that the spike locations and amplitude terms corresponding to the SDP 
solution are localized around the true values.} 

\hemant{From a computational perspective, the aforementioned approaches all admit a finite dimensional 
dual problem with an infinite number of linear constraints; this is a semi infinite program (SIP) for which 
there is an extensive literature \cite{SIP93}.
For the particular case of non-negative $x$, Boyd et al. \cite{boyd2017alternating} proposed an improved 
Frank-Wolfe algorithm in the primal. In certain instances, for e.g., with Fourier 
samples (such as in \cite{candes2014towards,candes2013super}), this SIP can also be reformulated as a SDP.
From a practical point of view, SDP is notoriously slow for moderately large number of variables.  
The algorithm of \cite{boyd2017alternating} is a first order scheme with potential local 
correction steps, and is practically more viable. }

\subsubsection{Prony, ESPRIT, MUSIC and extensions} \label{subsubsec:prony_intro_disc}
\hemant{When one is given the first few Fourier samples of the spike train $x$ (i.e., \eqref{eq:four_meas_model}), 
then there exist other approaches that can be employed. Prony's method \cite{prony} is a classical method that 
involves finding the roots of a polynomial, whose coefficients form a null vector of a certain Hankel matrix.
Prony's method and its variants have also been recently studied by Potts and Tasche (for e.g., \cite{potts2010parameter,potts13}), 
Plonka and Tasche \cite{plonka2014prony}, and others. 
The matrix pencil (MP) method \cite{mpmethod} is another classical approach that involves computing 
the generalized eigenvalues of a suitable matrix pencil. Both these 
methods recover $x$ exactly in the absence of noise provided $m \geq K$ (so $2K+1$ samples), 
and are also computationally feasible. Recently, Moitra \cite[Theorem 2.8]{moitra15} showed that the 
MP method is stable in the presence of noise provided the noise level is not too large, and $m > \frac{1}{\sep} + 1$. 
Moitra also showed \cite[Corollary 1.4]{moitra15} that such a dependence is necessary, in the 
sense that if $m < (1-\epsilon)/\sep$, then the noise level would have to be $O(2^{-\epsilon K})$ to be able to stably recover $x$.
Very similar in spirit to the MP method for sum of exponential estimation are the Prony-like methods ESPRIT and MUSIC which can also be  used for spike train estimation using the same Fourier domain measurement trick as in \cite{moitra15}. These methods can be interpreted as model reduction methods based on low-rank approximation of Hankel matrices \cite{markovsky2008structured}. The ESPRIT method  was studied in \cite{potts2017error} based on previous results\footnote{The estimation error on the $t_i$'s, $i=1,\ldots,K$, can be deduced from \cite[Theorem 5.1]{potts2017error}, while the error in the coefficients $u_i$, $i=1,\ldots,K$ comes as a result of a perturbation analysis based on the condition number of the Vandermonde matrix associated with the frequencies.} from \cite{bazan2003error}. Another line of research is about relationships with \hemant{AAK theory of Adamjan, Arov, and Krein} as developed in  \cite{plonka2014prony} (\hemant{see also \cite{plonka2016application}}). The MUSIC method was also studied in great detail in \cite[Theorem 4]{liao2016music}. The modified Matrix Pencil method, on the one hand, is often considered as less computationally expensive than MUSIC and, on the other hand, is amenable to an error analysis quite similar to the one in \cite{potts2017error}.} 

\subsection{Super-resolution with multiple kernels}
In this paper, we consider a generalization of the standard super-resolution problem by assuming that the measurement process 
now involves a \textit{superposition} of convolutions of several spike trains with \textit{different} point spread functions. 
This problem, which we call the ``\textit{multi-kernel unmixing super-resolution problem}'' 
appears naturally in many practical applications such as single molecule 
spectroscopy \cite{huang20153d}, spike sorting in neuron activity analysis \cite{lewicki1998review,brown2004multiple}, 
DNA sequencing \cite{li2002dna,li2004deconvolution}, spike hunting in 
galaxy spectra \cite{brutti2005spike,liu2007estimating} etc. 

A problem of particular interest at the National Physical Laboratory, the UK's national metrology institute, 
is isotope identification \cite{lu2014photopeak,stinnett2016automated} which is of paramount 
importance in nuclear security. Hand-held radio-isotope identifiers are known to suffer from 
poor performance \cite{stinnett2016automated} and new and precise algorithms have to be devised 
for this problem. While it is legitimately expected for Radio Isotope Identifier Devices to be 
accurate across all radioactive isotopes, the US Department of Homeland Security requires all 
future identification systems to be able to meet a minimum identification standard set forth in ANSI N42.34. Effective identification from low resolution information is known to be reliably achievable by trained spectroscopists whereas automated identification using standard algorithms sometimes fails up to three fourth of the time \cite{stinnett2014bayesian}. For instance, the spectrum of $^{232}Th$ is plotted in Figure \ref{fig:spectrumTh232}, which is extracted from \cite[p.9]{stinnett2014bayesian}. 
\begin{figure}
    \centering
    \includegraphics[scale = 0.4]{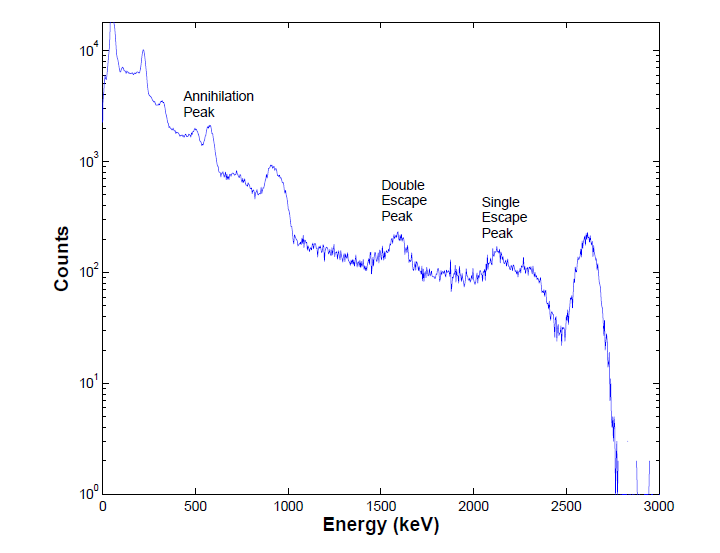}
    \caption{$^{232}Th$ spectrum showing escape peaks and annihilation peaks \cite[p.9]{stinnett2014bayesian}}
    \label{fig:spectrumTh232}
\end{figure}
Isotope identification involves the estimation of a set of peak locations in the gamma spectrum where the signal is blurred by convolution with kernels of different window sizes. Mixtures of different spectra can be particularly difficult to analyse using traditional methods and a need for precise unmixing algorithms in such generalised super-resolution problems may play an important role in future applications such as reliable isotope identification.

Another application of interest is DNA sequencing in the vein of \cite{li2002dna}. Sequencing is usually performed using some of the variants of the enzymatic method invented by Frederick Sanger \cite{adams2012automated}. Sequencing is based on enzymatic reactions, electrophoresis, and some detection technique.
Electrophoresis is a technique used to separate the DNA sub-fragments produced as the output of four specific chemical reactions, as described in more detail in \cite{li2002dna}. DNA fragments are negatively charged in solution. An 
efficient color-coding strategy has been developed to permit sizing of all four kinds of DNA sub-fragments by electrophoresis in a single lane of a slab gel or in a capillary. In each of the four chemical reactions, the primers are labeled by one of four different fluorescent dyes. The dyes are then excited by a laser based con-focal fluorescent detection system, in a region within the slab gel or capillary. In that process, fluorescence intensities are emitted in four wavelength bands as shown with different color codes in Figure \ref{fig:DNA} below. However, quoting \cite{berno1996graph}, ``\emph{because the electrophoretic process often fails to
separate peaks adequately, some form of deconvolution filter must be applied to the data to resolve
overlapping events. This process is complicated by the variability of peak shapes, meaning
that conventional deconvolution often fails.}'' 
The methods developed in the present paper aim at achieving accurate deconvolution with different 
peak shapes and might therefore be useful for practical deconvolution problems in DNA sequencing. 
\begin{figure}
    \centering
    \includegraphics[width=15cm]{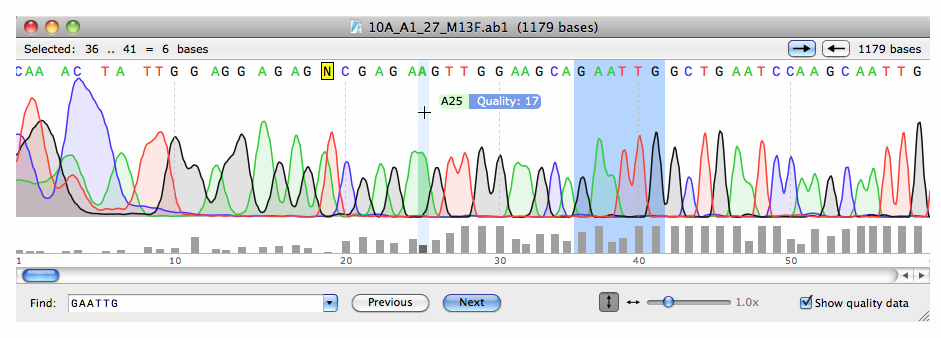}
    \caption{Capture of the result of a sequencing using SnapGene\_Viewer \cite{snapgene}.}
    \label{fig:DNA}
\end{figure}

We will now provide the mathematical formulation of the problem and describe the 
main idea of our approach along with our contributions, and discuss related work for this problem.

\subsection{Problem formulation} \label{subsec:prob_form}
Say we have $L$ groups of point sources where $\set{t_{l,i}}_{i=1}^{K} \subset [0,1)$ and 
$(u_{l,i})_{i=1}^{K}$ (with $u_{l,i} \in \mathbb{C}$) denote the locations and (complex-valued) amplitudes 
respectively of the sources in the $l^{th}$ group. Our signal of interest is defined as
\begin{align*}
x(t) = \sum_{l=1}^L x_l(t) = \sum_{l=1}^L \left(\sum_{j=1}^K \ u_{l,j} \delta(t-t_{l,j})\right).
\end{align*}	
Let $\kerfn \in L_1(\matR)$ be a positive definite function\footnote{Recall from \cite[Theorem 6.11]{wendland04} 
that a continuous function $g \in L_1(\matR)$ is positive definite if and only if $g$ is bounded and 
its Fourier transform is nonnegative and non vanishing.} with its Fourier transform 
$\kerfour(s) = \int_{\mathbb R} \kerfn(t) \exp \left( \iota 2\pi st\right) dt$ for $s \in \matR$. 
Consider $L$  distinct kernels $\kerfn_l(t) = \kerfn(t/\kerfnvar_l)$, $l=1,\ldots,L$
where $0 < \kerfnvar_1 < \cdots < \kerfnvar_L$.
Let $y(t) = \sum_{l=1}^L (\kerfn_l \star x_l)(t) = \sum_{l=1}^L \sum_{j=1}^K u_{l,j} \ \kerfn_l(t-t_{l,j})$ 
where $\star$ denotes the convolution operator.
Let $f$ denote the Fourier transform of $y$, i.e., 
$f(s) = \int_{\mathbb R} y(t) \exp \left( \iota 2\pi st\right) dt$. Denoting the Fourier transform 
of $\kerfn_l$ by $\kerfour_l$, we get
%
%
\begin{align} \label{eq:fourier_gauss_mix} 
f(s) & = \sum_{l=1}^L \underbrace{\kerfour_l(s) \left( \sum_{j=1}^K u_{l,j} \exp \left(\iota 2\pi st_{l,j}\right)\right)}_{f_{l}(s)}. 
\end{align}	
%
%
Assuming black box access to the complex valued function $f$, our aim is to recover estimates 
of $\set{t_{l,i}}_{i=1}^{K}$ and $(u_{l,i})_{i=1}^{K}$ for each $l = 1,\dots,L$ from 
few \textit{possibly noisy} samples $\tilde f(s) = f(s) + w(s)$ of $f$. Here, $w(s) \in \mathbb{C}$ 
denotes measurement noise at location $s$. \hemant{We remark that the choice of having $K$ summands for each 
group is only for ease of exposition, one can more generally consider $K_l \leq K$ summands for each $l = 1,\dots,L$.}

\paragraph{Gaussian kernels.} 
For ease of exposition, we will from now on consider $g$ to be a Gaussian, i.e., $g(t) = \exp(-t^2/2)$ so that 
$\kerfn_l(t) = \exp(-t^2/(2 \kerfnvar_l^2)$, $l=1,\ldots,L$. 
It is well known that $g$ is positive definite \cite[Theorem 6.10]{wendland04}, 
moreover, $\kerfour_l(s) = \sqrt{2\pi} \kerfnvar_l \exp(-2\pi^2 s^2 \kerfnvar_l^2)$.
We emphasize that our restriction to Gaussian kernels is only 
to minimize the amount of tedious computations in the proof. However, our proof technique 
can \hemant{likely} be extended to handle more general positive definite $g$ \hemant{possessing a strictly positive Fourier transform. Examples of such functions are: (a)  (Laplace kernel) $g(t) = \exp(-\abs{t})$, and (b) (Cauchy kernel) $g(t) = \frac{1}{1+t^2}$.}
%
%
\subsection{Main idea of our work: Fourier tail sampling} \label{subsec:main_idea}
To explain the main idea, let us consider the noiseless setting $w(s) = 0$. 
Our main algorithmic idea stems from observing \eqref{eq:fourier_gauss_mix}, 
wherein we notice that for $s$ sufficiently large, $f(s)$ is equal to $f_1(s)$ plus a perturbation term 
arising from the tails of $f_2(s),\dots,f_{L}(s)$. 
Thus, $f(s)/\kerfour_1(s)$ is equal to \hemant{$\sum_{j=1}^K u_{1,j} \exp \left(\iota 2\pi st_{1,j}\right)$} (which is 
a weighted sum of complex exponentials) plus a perturbation term. 
We control this perturbation by choosing $s$ to be sufficiently large, and recover estimates 
$\widehat t_{1,j}, \widehat u_{1,j}$ (up to a permutation $\perm_1$) 
via the Modified Matrix Pencil (MMP) method of Moitra \cite{moitra15} (outlined as Algorithm \ref{algo:mmp_method}).
Given these, we form the estimate $\widehat f_1(s)$ to $f_1(s)$ where
\begin{align*}
\widehat f_1(s) =\kerfour_1(s) \sum_{j=1}^K \widehat u_{1,\perm_1(j)} \exp(\iota 2 \pi s \widehat t_{1,\perm_1(j)}).
\end{align*} 
Provided the estimates are good enough, we will have $f(s) - \widehat f_1(s) \approx \sum_{l=2}^{L} f_l(s)$. Therefore, by  
applying the above procedure to $f(s) - \widehat f_1(s)$, we can hope to recover estimates of  
$t_{2,j}$'s and $u_{2,j}$'s as well. By proceeding recursively, it is clear that we can perform the above procedure to 
recover estimates to each $\set{t_{l,i}}_{i=1}^{K} \subset [0,1)$ and $(u_{l,i})_{i=1}^{K}$ for all $l=1,\dots,L$. 
A delicate issue that needs to be addressed for each intermediate group $1 < l < L$ is the following. 
While estimating the parameters for group $1 < l < L$, the samples $\frac{f(s) - \sum_{i=1}^{l-1} \widehat f_i(s)}{\kerfour_l(s)}$ 
that we obtain will have perturbation arising due to (a) the tails of $f_{l+1}(s), \dots, f_{L}(s)$ and, 
(b) the estimates $\widehat f_{1}(s),\dots,\widehat f_{l-1}(s)$ computed thus far.  
In particular, going ``too deep'' in to the tail of $f_l(s)$ would blow up 
the perturbation term in (b), while not going sufficiently deep would blow up the perturbation term in (a). 
Therefore, in order to obtain stable estimates to the parameters $t_{l,j}, u_{l,j}$, we will need to control 
these perturbation terms by carefully choosing the locations, as well as the number of sampling points in the tail.

\paragraph{Further remarks.} Our choice of using the MMP method for estimating the spike locations and amplitudes at 
each stage is primarily dictated by two reasons. Firstly, it is extremely simple to implement in 
practice. Moreover, it comes with precise quantitative error bounds (see Theorem \ref{thm:mmp_main_thm}) 
for estimating the source parameters -- especially for the setting 
of adversarial bounded noise -- which fit seamlessly in the theoretical analysis of our method. 
\hemant{Of course, in practice, other methods such as ESPRIT, MUSIC   
could also be used, and as discussed in Section \ref{subsubsec:prony_intro_disc}, there also exist some error bounds in the literature for these methods. To our knowledge, error bounds of the form in 
Theorem \ref{thm:mmp_main_thm} do not currently exist for these other methods. In any case, a full discussion in this regard is outside the scope of the paper, and analyzing the performance of our algorithm with methods other than the MMP 
method is a direction for future work.}


\subsection{Main results}  \label{subsec:main_res}
Our algorithm based on the aforementioned idea is outlined as Algorithm \ref{algo:mmp_gaussian_unmix} which we name 
KrUMMP (Kernel Unmixing via Modified Matrix Pencil). \hemant{At each stage $l = 1,\dots,L$, we choose a ``sampling offset'' $s_l$ and obtain the (potentially noisy) Fourier samples at $2 m_l$ locations $s_l +i$ for $i=-m_l\dots,m_l - 1$.} Our main result 
for the noiseless setting ($w \equiv 0$) is stated as Theorem \ref{thm:gen_case_main} in its full generality. 
We state its following informal version assuming the spike amplitudes in each group to 
be\footnote{\hemant{The symbols $\lesssim$  and $\asymp$ are used to hide positive constants, see Section \ref{sec:prob_setup_prelim}.}} \hemant{$\asymp 1$}. Note that $d_w:[0,1]^2 \rightarrow [0,1/2]$ is the usual wrap around distance 
on $[0,1]$ (see \eqref{eq:wrap_arnd_dist}).
%
\begin{theorem}[Noiseless case] \label{thm:main_noiseless_informal}
Denote $\sep_l := \min_{i,j} d_w(t_{l,i},t_{l,j}) > 0$ for each $1 \leq l \leq L$. 
Let $0 < \varepsilon_1 \leq \varepsilon_2 \leq \cdots \leq \varepsilon_L$ satisfy 
$\varepsilon_l \lesssim \sep_l/2$ for each $l$. Moreover, let $\varepsilon_{L-1} \lesssim \alpha \varepsilon_L^2$ 
and $\varepsilon_l \lesssim \beta_l(\varepsilon_{l+1})^{2(1+\gamma_l)}$ hold for $1 \leq l \leq L-1$ with 
$\alpha,\beta_l,\gamma_l > 0$ depending on the problem parameters (see \eqref{eq:eps_main_alpha}, \eqref{eq:eps_main_beta}). 
Finally, in Algorithm \ref{algo:mmp_gaussian_unmix}, let $m_L \asymp 1/\sep_L$, $s_L = 0$, and 
\begin{equation*}
m_l \asymp 1/\sep_l, \quad 
s_l \asymp m_l + \frac{1}{(\kerfnvar_{l+1}^2 - \kerfnvar_l^2)^{1/2}} 
\log^{1/2}\left(\frac{K^{3/2}(L-l)\kerfnvar_L}{\varepsilon_l\kerfnvar_l}\right) ; \quad 1 \leq l \leq L-1.
\end{equation*}
Then, for each $l=1,\dots,L$, there exists a permutation $\perm_l: [K] \rightarrow [K]$ such that 
\begin{align*}
d_w(\widehat t_{l,\perm_l(j)}, t_{l,j}) &\leq \varepsilon_l, \quad 
\abs{\widehat u_{l,\perm_l(j)} - u_{l,j}} \lesssim E_l(\varepsilon_l) \varepsilon_l; \quad j=1,\dots,K.
\end{align*}
Here, $E_L(\varepsilon_L) \lesssim \frac{K}{\sep_L}$,  
$E_l(\varepsilon_l) \lesssim \left(\frac{K}{\sep_l} + 
\frac{1}{(\kerfnvar_{l+1}^2 - \kerfnvar_l^2)^{1/2}} \log^{1/2}\left(\frac{K^{3/2}(L-l)\kerfnvar_L}{\varepsilon_l\kerfnvar_l}\right)\right)$ for $1 \leq l \leq L-1$.
\end{theorem}
The conditions on $\varepsilon_l \in (0,1)$ imply that the estimation errors corresponding to group $l$ 
should be sufficiently smaller than that of group $l+1$. This is because the estimation errors arising 
in stage $l$ percolate to stage $l+1$, and hence need to be controlled for stable recovery of source parameters for 
the $(l+1)^{th}$ group. Note that the conditions on $s_l$ (sampling offset at stage $l$) involve upper and lower bounds for 
reasons stated in the previous section. If $\kerfnvar_l$ is close to $\kerfnvar_{l+1}$ then $s_l$ will have to 
be suitably large in order to distinguish between $\kerfour_l$ and $\kerfour_{l+1}$, as one would expect intuitively.
An interesting feature of the result is that it only depends on separation within a group (specified by $\sep_l$), and so
spikes belonging to different groups are allowed to overlap.

Our second result is for the noisy setting. 
Say at stage $1 \leq p \leq L$ of Algorithm \ref{algo:mmp_gaussian_unmix}, we observe $\fnoise(s) = f(s) + w_p(s)$
where $w_p(s) \in \mathbb{C}$ denotes noise at location $s$. 
Our main result for this setting is Theorem \ref{thm:gen_case_main_noisy}.
Denoting $w_p = (w_p(s_p-m_p),w_p(s_p-m_p+1),\dots,w_p(s_p+m_p-1))^{T} \in \mathbb{C}^{2m_p}$ to be 
the noise vector at stage $p$, we state its informal version below assuming the spike amplitudes in each group 
to be \hemant{$\asymp 1$}.
%
\begin{theorem}[Noisy case] \label{thm:main_noisy_informal}
Say at stage $1 \leq p \leq L$ of Algorithm \ref{algo:mmp_gaussian_unmix}, we observe $\fnoise(s) = f(s) + w_p(s)$
where $w_p(s) \in \mathbb{C}$ denotes noise at location $s$. For each $1 \leq l \leq L$, let 
$\varepsilon_l,m_l,s_l$ be chosen as specified in Theorem \ref{thm:main_noiseless_informal}. 
Say $\norm{w_L}_{\infty} \lesssim \kerfnvar_L e^{-\kerfnvar_L^2/\sep_L^2} \frac{\varepsilon_L}{\sqrt{K}}$, 
and also
\begin{equation*} 
\norm{w_l}_{\infty} 
\lesssim \frac{(\varepsilon_l\kerfnvar_l)^{1+C(\kerfnvar_l,\kerfnvar_{l+1},\sep_l)}}{\sqrt{K} (K^{3/2} L \kerfnvar_L)^{C(\kerfnvar_l,\kerfnvar_{l+1},\sep_l)}}; \quad 2 \leq l \leq L-1 
\end{equation*}
where $C(\kerfnvar_l,\kerfnvar_{l+1},\sep_l) > 0$ depends only on $\kerfnvar_l,\kerfnvar_{l+1},\sep_l$. 
Then, for each $l=1,\dots,L$, there exists a permutation $\perm_l: [K] \rightarrow [K]$ such that 
\begin{align*}
d_w(\widehat t_{l,\perm_l(j)}, t_{l,j}) \leq \varepsilon_l, \quad 
\abs{\widehat u_{l,\perm_l(j)} - u_{l,j}} \lesssim E_l(\varepsilon_l) \varepsilon_l; \quad j=1,\dots,K, 
\end{align*}
where $E_l(\cdot)$ is as in Theorem \ref{thm:main_noiseless_informal}.
\end{theorem}
\hemant{As remarked earlier in Section \ref{subsec:prob_form}, one can more generally consider $K_l \leq K$ summands for the $l^{th}$ group -- our algorithm 
and results remain unchanged.}
%
\section{Notation and Preliminaries} \label{sec:prob_setup_prelim}
\paragraph{Notation.} Vectors and matrices are denoted by lower and upper case letters respectively. 
For $n \in \mathbb{N}$, we denote $[n] = \set{1,\dots,n}$. The imaginary unit is denoted by $\iota = \sqrt{-1}$.  
\hemant{The notation $\log^{1/2}(\cdot)$ is used to denote $\abs{\log(\abs{\cdot})}^{1/2}$.}
The $\ell_p$ ($1 \leq p \leq \infty$) norm of a vector $x \in \matR^n$ is denoted by 
$\norm{x}_p$ (defined as $(\sum_i \abs{x_i}^p)^{1/p}$). In particular, $\norm{x}_{\infty} := \max_i \abs{x_i}$. 
For a matrix $A \in \matR^{m \times n}$, we will denote its spectral norm (i.e., largest singular value) by $\norm{A}$ 
and its Frobenius norm by $\norm{A}_F$ (defined as $(\sum_{i,j} A_{i,j}^2)^{1/2}$). For positive numbers $a,b$, 
we denote $a \lesssim b$ to mean that there exists an absolute constant $C > 0$ such that $a \leq C b$. 
\hemant{If $a \lesssim b$ and $b \lesssim a$ then we denote $a \asymp b$.} 
The wrap around distance on $[0,1]$ is denoted by $d_w: [0,1]^2 \rightarrow [0,1/2]$ where we recall that 
\begin{equation} \label{eq:wrap_arnd_dist}
 d_w(t_1,t_2) = \min\set{\abs{t_1 - t_2}, 1-\abs{t_1 - t_2}}.
\end{equation}
%

\subsection{Matrix Pencil (MP) method}
We now review the classical Matrix Pencil (MP) method for estimating positions of 
point sources from Fourier samples. Consider the signal $x(t) := \sum_{j=1}^{K} u_{j} \delta(t - t_{j})$ where 
$u_{j} \in \mathbb{C}$, $t_j \in [0,1)$ are unknown. Let $f: \matR \rightarrow \mathbb{C}$ be the Fourier 
transform of $x$ so that $f(s) = \sum_{j=1}^K u_j \exp \left(\iota 2\pi s t_j \right)$. For any given offset 
$s_0 \in \matZ^{+}$, let $s = s_0 + i$ where $i \in \matZ$; clearly 
\begin{align*}
	f(s_0 + i) = \sum_{j=1}^K u_j \exp \left(\iota 2\pi(s_0+i)t_j\right) = \sum_{j=1}^K u'_j \exp \left(\iota 2\pi i t_j \right) 
\end{align*}
where $u'_j = u_j \exp \left(\iota 2\pi s_0 t_j \right), \quad j=1,\ldots,K$. 
Choose $i \in \set{-m,-m+1,\ldots,m-1}$ to form the $m \times m$ matrices
\begin{align} \label{eq:mp_matrix_0}
	H_0  =  
	\begin{bmatrix}
	f(s_0) & f(s_0+1) & \cdots & f(s_0+m-1) \\
	f(s_0-1) & f(s_0) & \cdots & f(s_0+m-2) \\
	\vdots & \vdots & & \vdots \\\
	f(s_0-m+1) & f(s_0-m+2) & \cdots & f(s_0)
\end{bmatrix}
\end{align}
and 
\begin{align} \label{eq:mp_matrix_1}
	H_1 = 
	\begin{bmatrix}
	f(s_0-1) & f(s_0) & \cdots & f(s_0+m-2) \\
	f(s_0-2) & f(s_0-1) & \cdots & f(s_0+m-3) \\
	\vdots & \vdots & & \vdots \\
	f(s_0-m) & f(s_0-m+1) & \cdots & f(s_0-1)  
	\end{bmatrix}.
\end{align}
Denoting $\alpha_j = \exp \left(-\iota 2\pi t_j \right)$ for $j= 1, \ldots, K$, and the Vandermonde matrix 
\begin{align} 
	V = \begin{bmatrix}
	1 & 1 & \cdots & 1 \\
	\alpha_1 & \alpha_2 & \cdots & \alpha_K \\
	\vdots & \vdots &  & \vdots \\ \alpha_1^{m-1} & \alpha_2^{m-1} & \cdots & \alpha_K^{m-1}
	\end{bmatrix},
	\label{vander}
\end{align}
clearly $H_0=V D_{u'} V^H$ and $H_1=VD_{u'} D_\alpha V^H$. Here, 
$D_{u'} = \text{diag}(u'_1,\dots,u'_K)$ and  $D_\alpha = \text{diag}(\alpha_1,\dots,\alpha_K)$ are 
diagonal matrices. One can readily 
verify\footnote{\hemant{Recall (see \cite[Definition 2.1]{Sun83_GenEigDef}) that 
$\lambda = \beta/\gamma$ (where $(\beta,\gamma) \neq (0,0)$) is 
a generalized eigenvalue of $(H_1,H_0)$ if it satisfies
$\text{rank } (\gamma H_1 - \beta H_0) 
< \max_{(\zeta_1,\zeta_0) \in \mathbb{C}^2 \setminus \set{0,0}}\text{ rank}(\zeta_1 H_1 - \zeta_0 H_0).$ 
Clearly, this is only satisfied if $\lambda = \alpha_j$ (see also \cite[Theorem 2.1]{mpmethod}).}}
that the $K$ non zero generalized eigenvalues of ($H_1,H_0$) are 
equal to the $\alpha_j$'s. Hence by forming the matrices $H_0, H_1$, we can recover the unknown 
$t_j$'s exactly from $2K$ samples of $f$. Once the $t_j$'s are recovered, we can recover the 
$u'_j$'s exactly, as the solution of the linear system
\begin{align*}
\begin{bmatrix}
f(0) \\ f(1) \\ \vdots \\  f(m-1)
\end{bmatrix} & 
= \begin{bmatrix}
1 & 1 & \cdots & 1 \\
\alpha_1 & \alpha_2 & \cdots & \alpha_K \\
\vdots & \vdots & & \vdots \\
\alpha_1^{m-1} & \alpha_2^{m-1} & \cdots & \alpha_K^{m-1}
\end{bmatrix}
\begin{bmatrix}
u'_1 \\
\vdots \\
u'_K
\end{bmatrix}.
\end{align*}
Thereafter, the $u_j$'s are found as $u_j = u'_j/\exp(\iota 2\pi s_0 t_j)$.

\steph{In \cite{mpmethod}, the authors elaborate further on this approach and in the noiseless case, propose an equivalent formulation of the Matrix Pencil approach as a standard eigenvalue problem. Using continuity of the eigenvalues with respect to perturbation, they also suggest that this eigenvalue problem provides a good estimator in the noisy case. In the next section, we discuss an alternative approach proposed by Moitra, which involves solving a generalised eigenproblem, and for which Moitra provided a precise quantitative perturbation result.  
} 
\subsection{The Modified Matrix Pencil (MMP) method} 
We now consider the noisy version of the setup defined in the previous section. For a collection of $K$ 
point sources with parameters $u_j \in \mathbb{C}$, $t_j \in [0,1]$, we are given noisy samples 
\begin{equation} \label{eq:noisy_four_samps}
\fnoise(s) = f(s) + \eta_s,  
\end{equation}
for $s \in \mathbb{Z}$ and where $\eta_s \in \mathbb{C}$ denotes noise. 

Let us choose $s = s_0 + i$ for a given offset $s_0 \in \matZ^+$, and $i \in \set{-m,-m+1,\ldots,m-1}$ for a positive integer $m$. 
Using $(\fnoise(s_0 + i))_{i=-m}^{m-1}$, let us form the matrices $\widetilde H_0, \widetilde H_1 \in \mathbb{C}^{m \times m}$ 
as in \eqref{eq:mp_matrix_0},\eqref{eq:mp_matrix_1}. We now have $\widetilde H_0 = H_0 + E$, $\widetilde H_1 = H_1 + F$ where 
$H_0, H_1$ are as defined in \eqref{eq:mp_matrix_0}, \eqref{eq:mp_matrix_1}, and 
\begin{align*}
E = 
\begin{bmatrix}
	\eta_0 & \eta_1 & \ldots & \eta_{m-1} \\
	\eta_{-1} & \eta_0 & \ldots & \eta_{m-2} \\
	\vdots & \vdots & & \vdots \\
	\eta_{-(m-1)} & \eta_{-(m-2)} & \ldots & \eta_0
\end{bmatrix}, \quad 
	F = 
	\begin{bmatrix}
	\eta_{-1} & \eta_0 & \ldots & \eta_{m-2} \\
	\eta_{-2} & \eta_{-1} & \ldots & \eta_{m-3} \\
	\vdots & \vdots & & \vdots \\
	\eta_{-m} & \eta_{-(m-1)} & \ldots & \eta_{-1}
	\end{bmatrix}.
\end{align*}
represent the perturbation matrices. Algorithm \ref{algo:mmp_method} namely the Modifed Matrix Pencil (MMP) method \cite{moitra15} 
outlines how we can recover $\widehat t_j$, $\widehat u_j$ for $j=1,\dots,K$.

Before proceeding we need to make some definitions. Let $\umax = \max_j \abs{u_j}$ and $\umin = \min_j \abs{u_j}$. 
We denote the largest and smallest non zero singular values of $V$ by $\sigmamax,\sigmamin$ respectively, and 
the condition number of $V$ by $\condnum$ where $\condnum = \sigmamax/\sigmamin$. Let $\noisemax := \max_{i} \abs{\eta_i}$. 
We will define $\sep$ as the minimum separation between the locations of the point sources where 
$\sep := \min_{j\neq j'} \ d_w(t_j,t_{j'})$. 

\begin{algorithm*}[!ht]
\caption{Modifed Matrix Pencil (MMP) method \cite{moitra15}} \label{algo:mmp_method} 
\begin{algorithmic}[1] 
\State \textbf{Input:} $K,m,s_0, \widetilde H_0, \widetilde H_1$. 
\State \textbf{Output:}  $\widehat u_j, \widehat t_j$; $j=1,\dots,K$.  

\State Let $\Uhat \in \mathbb{C}^{m \times K}$ be the top $K$ singular vector matrix of $\widetilde H_0$. 

\State Let $\Ahat = \Uhat^H \widetilde H_0 \Uhat$ and $\Bhat = \Uhat^H \widetilde H_1 \Uhat$.

\State Find generalized eigenvalues $(\widehat \lambda_j)_{j=1}^{n}$ of $(\Bhat, \Ahat)$.

\State Let $\widehat\alpha_j = \exp(-\iota2\pi \widehat t_j) = \widehat\lambda_j/\abs{\widehat\lambda_j}$ where $\widehat t_j \in [0,1)$.

\State Form Vandermonde matrix $\widehat V \in \mathbb{C}^{m \times K}$ using $\widehat\alpha_j$'s. 

\State Find $\widehat u' = {\widehat V}^{\dagger} v$ where $v = [f(s_0) \ f(s_0 + 1) \ \cdots \ f(s_0 + m-1)]^T$.

\State Find $\widehat u_j = {\widehat u'}_j \exp(-\iota 2\pi s_0 \widehat t_j)$; $j=1,\dots,K$.
\end{algorithmic}
\end{algorithm*}
The following theorem is a more precise version of \cite[Theorem 2.8]{moitra15}, with the constants computed explicitly. 
Moreover, the result in \cite[Theorem 2.8]{moitra15} was specifically for the case $s_0 = 0$, \hemant{and also has inaccuracies in the 
proof.} We outline \hemant{the corrected version of Moitra's proof 
in Appendix \ref{app:subsec_moitraMP_proof}} and fill in additional details (using auxiliary results from 
Appendix \ref{sec:aux_results}) to arrive at the statement of Theorem \ref{thm:mmp_main_thm}.
%
%
%
%
\begin{theorem}[\cite{moitra15}] \label{thm:mmp_main_thm}
For $0 \leq \varepsilon < \sep/2$, say $ m > \frac{1}{\sep - 2\varepsilon} + 1$. Moreover, for $C = 10 + \frac{1}{2\sqrt{2}}$, say 
\begin{equation} \label{eq:moitra_thm_noise_cond}
\noisemax \leq \varepsilon \frac{\umin \sigma_{\min}^2}{2mC\sqrt{K}} \left( 1+16 \ \kappa^2 \frac{\umax}{\umin}\right)^{-1}.
\end{equation}
Then, there exists a permutation $\perm: [K] \mapsto [K]$ such that the output of the 
MMP method satisfies for each $i=1,\dots,K$ 
\begin{align} \label{eq:mp_thm_bd}
d_w(\widehat t_{\perm(i)},t_i) \le \varepsilon, \quad 
\norm{\widehat u_{\perm} - u}_{\infty} \le 
\left(\frac{2\pi m^{3/2} K \umax + \frac{\umin \sigma^2_{\min}}{2 C \sqrt{mK}}\left( 1+16 \ \kappa^2 \frac{\umax}{\umin}\right)^{-1}}{(m - \frac{1}{\sep-2\varepsilon} -1)^{1/2}} 
+ 2\pi \umax s_0 \right)\varepsilon 
\end{align}
where $\widehat u_{\perm}$ is formed by permuting the indices of $\widehat u$ w.r.t $\perm$. 
\end{theorem}
The following Corollary of Theorem \ref{thm:mmp_main_thm} simplifies the expression for the 
bound on $\norm{\widehat u_{\perm} - u}_{\infty}$ in \eqref{eq:mp_thm_bd}, and will be useful for our main results later on. 
The proof is deferred to Appendix \ref{app:subsec_corr_moitra_MP}.
\begin{corollary} \label{corr:moitra_MP}
For $0 \leq \varepsilon < c \sep/2$ where $c \in [0,1)$ is a constant, say $\frac{2}{\sep - 2\varepsilon} + 1 < m \leq \frac{2}{(1-c)\sep} + 1$. 
\hemant{Denoting $\urel = \frac{\umax}{\umin}$, $C = 10 + \frac{1}{2\sqrt{2}}$, and $B(\urel,K) = \frac{1}{5C\sqrt{K}} (1 + 48\urel)^{-1}$, say}
\begin{equation} \label{eq:corr_moitra_MP_noisecond}
\hemant{\noisemax \leq \varepsilon \umin B(\urel,K).} 
\end{equation}
Then, there exists a permutation $\perm: [K] \mapsto [K]$ such that the output of the 
MMP method satisfies for each $i=1,\dots,K$ 
\begin{align} \label{eq:mmp_corr_est_err}
d_w(\widehat t_{\perm(i)},t_i) \le \varepsilon, 
\quad 
\hemant{\norm{\widehat u_{\perm} - u}_{\infty} < (\tilde{C}(\sep,c,K,\urel) + 2\pi s_0) \umax \varepsilon}
\end{align}
where \hemant{$\tilde{C}(\sep,c,K,\urel) = 4\pi K \left(\frac{2}{\sep(1-c)} + 1\right) + \frac{2}{C\sqrt{K}}\left(\urel 
+ 16\urel^2\right)^{-1}$}, 
and $\widehat u_{\perm}$ is formed by permuting the indices of $\widehat u$ w.r.t $\perm$. 
\end{corollary}
%
%

\section{Unmixing Gaussians in Fourier domain: Noiseless case} \label{sec:multi_gauss_unmix_analysis_noiseless}
We now turn our attention to the main focus of this paper, namely that of unmixing Gaussians in the Fourier domain. We will in general assume the Fourier samples to be noisy as in \eqref{eq:noisy_four_samps}, with $f$ defined in \eqref{eq:fourier_gauss_mix}.
In this section, we will focus on the noiseless setting wherein $\fnoise(s) = f(s)$ for each $s$.
The noisy setting is analyzed in the next section.

Let us begin by noting that when $L = 1$, i.e., in the case of a single kernel, the problem is solved easily. 
Indeed, we have from \eqref{eq:fourier_gauss_mix} that $f(s) = \kerfour_1(s) \sum_{j=1}^K u_{1,j} \exp\left(\iota 2\pi st_{1,j} \right)$.
Clearly, one can exactly recover $(t_{1,j})_{j=1}^K \in [0,1)$ and $(u_{1,j})_{j=1}^K \in \mathbb{C}$  
via the MP method by first obtaining the samples $f(-m)$, $f(-m+1)$, \ldots, $f(m-1)$, and then working 
with $f(s)/\kerfour_1(s)$. 

The situation for the case $L \geq 2$ is however more delicate. Before proceeding, we need to make some definitions and 
assumptions.
\begin{itemize}
\item We will denote $\umax = \max_{l,j} \abs{u_{l,j}}$, $\umin = \min_{l,j} \abs{u_{l,j}}$, \hemant{and $\urel = \frac{\umax}{\umin}$}.

\item The sources in the $l^{th}$ group are assumed to have a minimum separation of 
$$\sep_l:= \min_{i \neq j} d_w(t_{l,i},t_{l,j}) > 0.$$ 

\item Denoting $\alpha_{l,j}= \exp \left(-i2\pi t_{l,j} \right)$, $V_l \in \mathbb{C}^{m_l \times K}$ will denote the Vandermonde matrix 
\begin{align*}
\begin{bmatrix}
1 & 1 & \cdots & 1 \\
\alpha_{l,1} & \alpha_{l,2} & \cdots & \alpha_{l,K} \\
\vdots & \vdots & & \vdots \\
\alpha_{l,1}^{m_l-1} & \alpha_{l,2}^{m_l-1} & \cdots & \alpha_{l,K}^{m_l-1}
\end{bmatrix}
\end{align*}
for each $l=1,\dots, L$ analogous to \eqref{vander}. $\sigma_{\max,l}, \sigma_{\min,l}$ will denote its largest 
and smallest non-zero singular values, and $\condnum_l = \sigma_{\max,l}/\sigma_{\min,l}$ 
its condition number. Recall from Theorem \ref{thm:moitra_vander_condbd} that if $m_l > \frac{1}{\sep_l} + 1$, then 
$\sigma_{\max,l}^2 \leq m_l+\frac{1}{\sep_l}-1$ and 
$\sigma_{\min,l}^2 \geq m_l - \frac{1}{\sep_l}-1$ and thus $\condnum_l^2 \leq \frac{m_l+\frac{1}{\sep_l}-1}{m_l-\frac{1}{\sep_l}-1}$.
\end{itemize}
%
%
%
\begin{algorithm*}[!ht]
\caption{Kernel Unmixing via Modified Matrix Pencil (KrUMMP)} \label{algo:mmp_gaussian_unmix} 
\begin{algorithmic}[1] 
\State \textbf{Input:} $K$, $m_l,s_l,\kerfnvar_l$ ; $l=1,\dots,L$.
\State \textbf{Initialize:} $\est{u}_{l,j}, \est{t}_{l,j} = 0$; $l=1,\dots,L$; $j=1,\dots,K$. Also, $\est{f}_1 \equiv 0$.
\State \textbf{Output:} $\est{u}_{l,j}$, $\est{t}_{l,j}$; $l=1,\dots,L$; $j=1,\dots,K$.  

\For {$l=1, \dots, L$}

\State Obtain samples $\frac{\fnoise(s_l + i) - \sum_{j=1}^{l-1} \est{f}_j(s_l + i)}{\kerfour_l(s_l + i)}$ for $i=-m_l,\dots,m_l-1$.

\State Form $\widetilde {H}_{0}^{(l)}, \widetilde {H}_{1}^{(l)} \in \mathbb{C}^{m_l \times m_l}$ using the above samples as 
in \eqref{eq:mp_matrix_0}, \eqref{eq:mp_matrix_1}.

\State Input $\widetilde {H}_{0}^{(l)}, \widetilde {H}_{1}^{(l)}$ to MMP algorithm and obtain estimates 
$(\est{u}_{l,j})_{j=1}^{K}$, $(\est{t}_{l,j})_{j=1}^{K}$.

\State Define $\est{f}_l: \mathbb{R} \rightarrow \mathbb{C}$ as 
$\est{f}_l(s) :=  \kerfour_l(s) \sum_{j=1}^K \est{u}_{l,j} \exp(\iota 2\pi s \est{t}_{l,j}).$
\EndFor
 
\end{algorithmic}
\end{algorithm*}
%
%
%
%
%
%
\subsection{The case of two kernels} 
We first consider the case of two Gaussian kernels as the analysis here is relatively easier to 
digest compared to the general case. Note that $f$ is now of the form
\begin{align*}
f(s) & = \kerfour_1(s) \left(\sum_{j=1}^K u_{1,j} \exp(\iota 2\pi s t_{1,j})\right)+ 
\kerfour_2(s) \left(\sum_{j=1}^K u_{2,j} \exp (\iota 2\pi s t_{2,j})\right)
\end{align*}
where we recall $\kerfour_l(s) = \sqrt{2\pi} \ \kerfnvar_l \exp \left(-2\pi^2 s^2 \kerfnvar_l^2\right)$; $l=1,2$.
The following theorem provides sufficient conditions on the choice of the sampling parameters for 
approximate recovery of $t_{1,j}, t_{2,j}, u_{1,j}, u_{2,j}$ for each $j=1,\dots,K$. 
%
%
%
%
\begin{theorem} \label{thm:two_win_case}
Let $\hemant{0  <  \varepsilon_2 < c\sep_2/2}$ for a constant $c \in [0,1)$, $s_2 = 0$, 
and $m_2 \in \mathbb{Z}^{+}$ satisfy 
$\frac{2}{\sep_2 - 2\varepsilon_2} + 1 \leq m_2 < \frac{2}{\sep_2 (1-c)} + 1 \ (\ = \mup{2})$. 
For \hemant{$B(\urel,K)$ as in Corollary \ref{corr:moitra_MP}}, let $\hemant{0 <  \varepsilon_1 < c\sep_1/2}$ also satisfy
\begin{align}
e^{(2\pi^2 \mup{2}^2 (\kerfnvar_2^2 - \kerfnvar_1^2))} \left(\hemant{2\pi \mup{2}} + \bar{C}_1 + 
\bar{C}_2 \log^{1/2}\left(\frac{\bar{C}_3}{\varepsilon_1}\right)\right) \hemant{\urel} \varepsilon_1 
\leq 
\hemant{\varepsilon_2 \frac{\kerfnvar_2 B(\urel,K)}{K \kerfnvar_1}}, \label{eq:twowin_eps12_cond}
\end{align}
where $\bar{C}_1, \bar{C}_2, \bar{C}_3 > 0$ are constants depending (see \eqref{eq:twowin_proof_temp9}, \eqref{eq:twowin_proof_temp10}) 
\hemant{on $c,\urel,\kerfnvar_1,\kerfnvar_2,\sep_1,K$ and a constant $\widetilde{c} > 1$}. 
Say $m_1,s_1 \in \mathbb{Z}^{+}$ are chosen to satisfy 
\begin{align*}
\frac{2}{\sep_1 - 2\varepsilon_1} + 1 \leq m_1 < \frac{2}{\sep_1 (1-c)} + 1 \ (\ = \mup{1}), 
\quad 
S_1 \leq s_1 \leq \widetilde{c} S_1, 
\end{align*}
where $S_1 = m_1 + \frac{1}{(2\pi^2(\kerfnvar_2^2 - \kerfnvar_1^2))^{1/2}} 
\log^{1/2}\left(\frac{K\urel\kerfnvar_2}{\kerfnvar_1 B(\urel,K) \varepsilon_1} \right)$. 
Then, there exist permutations $\perm_1,\perm_2: [K] \rightarrow [K]$ such that for $j=1,\dots,K$, 
\begin{align*}
d_w(\widehat t_{1,\phi_1(j)}, t_{1,j}) &\leq \varepsilon_1, \quad 
\abs{\widehat u_{1,\phi_1(j)} - u_{1,j}} < \left(\bar{C}_1 + \bar{C}_2 \log^{1/2}\left(\frac{\bar{C}_3}{\varepsilon_1}\right)\right) \hemant{\umax} \varepsilon_1, \\
d_w(\widehat t_{2,\phi_2(j)}, t_{2,j}) &\leq \varepsilon_2, \quad \abs{\widehat u_{2,\phi_2(j)} - u_{2,j}} < \widetilde{C}_2 \hemant{\umax} \varepsilon_2,
\end{align*}
where 
\begin{align} \label{eq:twogauss_ctil_exp}
\hemant{\widetilde{C}_l = \widetilde{C}(\sep_l,c,K,\urel) = 4\pi K \mup{l} + \frac{2}{C\sqrt{K}}\left(\urel + 16\urel^2\right)^{-1}}; \quad l=1,2, 
\end{align}
\hemant{and $\widetilde{C}(\cdot), C$ are as defined in Corollary \ref{corr:moitra_MP}.}
\end{theorem}
%
%
\paragraph{Interpreting Theorem \ref{thm:two_win_case}.} 
Before proceeding to the proof, we make some useful observations.
\begin{itemize}
\item[(a)] We first choose the sampling parameters $\varepsilon$ (accuracy), $m$ (number of samples), $s$ (offset) for 
the inner kernel $\kerfour_2$ and then the outer kernel, i.e., $\kerfour_1$. 
For group $i$ ($= 1,2$), we first choose $\varepsilon_i$, then $m_i$ (depending on $\varepsilon_i$), and finally the offset $s_i$ 
(depending on $m_i,\varepsilon_i$).

\item[(b)] The choice of $\varepsilon_2$ is free, but the choice of $\varepsilon_1$ is constrained by $\varepsilon_2$ as seen from 
\eqref{eq:twowin_eps12_cond}. In particular, $\varepsilon_1$ needs to be sufficiently small with respect to $\varepsilon_2$ so that 
the perturbation arising due to the estimation errors for group $1$ are controlled when we estimate the 
parameters for group $2$.

\item[(c)] The lower bound on $s_1$ ensures that we are sufficiently deep in the tail of 
$\kerfour_2$, so that its effect is negligible. The upper bound on $s_1$ is to control 
the estimation errors of the source amplitudes for group $1$ (see \eqref{eq:mmp_corr_est_err}). 
Observe that $s_2 = 0$ since $\kerfour_2$ is the innermost kernel, and so there is no perturbation 
arising due to the tail of any other inner kernel.

\item[(d)] In theory, $\varepsilon_1$ can be chosen to be arbitrarily close to zero; however, this would 
result in the offset $s_1$ becoming large. Consequently, this will lead to numerical errors when we divide by 
$\kerfour_1(s_1+i)$; $i=-m_1,\dots,m_1-1$ while estimating the source parameters for group $1$. 
\end{itemize}
%
%
\paragraph{Order wise dependencies.} The Theorem is heavy in notation, so it would help to 
understand the order wise dependencies of the terms involved. Assume $\umax,\umin \asymp 1$. 
We have $\widetilde{C}_1 \asymp K/\sep_1$ and $\widetilde{C}_2 \asymp K/\sep_2$ which leads to
\begin{align*}
\bar{C}_1 \asymp  \frac{K}{\sep_1}, \quad 
\bar{C}_2 \asymp \frac{1}{(\kerfnvar_2^2 - \kerfnvar_1^2)^{1/2}}, \quad \bar{C}_3 \asymp \frac{K^{3/2}\kerfnvar_2}{\kerfnvar_1}.
\end{align*}
\begin{itemize}
\item[(a)] For group $2$, we have $\varepsilon_2 \lesssim \sep_2$, $m_2 \asymp 1/\sep_2$, $s_2 = 0$, and 
\begin{align*}
d_w(\widehat t_{2,\phi_2(j)}, t_{2,j}) \leq \varepsilon_2, \quad 
\abs{\widehat u_{2,\phi_2(j)} - u_{2,j}} \lesssim \frac{K}{\sep_2} \varepsilon_2.
\end{align*}

\item[(b)] For group $1$, $\varepsilon_1 \lesssim \sep_1$ and \eqref{eq:twowin_eps12_cond} translates 
to
\begin{align} \label{eq:twowin_eps12cond_ord}
\left(\frac{1}{\sep_2} + \frac{K}{\sep_1} + \frac{1}{(\kerfnvar_2^2 - \kerfnvar_1^2)^{1/2}} 
\log^{1/2}\left(\frac{K^{3/2}\kerfnvar_2}{\kerfnvar_1\varepsilon_1} \right) \right)\varepsilon_1 
\lesssim \varepsilon_2\left(\frac{\kerfnvar_2}{\kerfnvar_1 K^{3/2}}\right) \exp\left(-\frac{2\pi^2(\kerfnvar_2^2 - \kerfnvar_1^2)}{\sep_2^2}\right).
\end{align}
Moreover, $m_1 \asymp 1/\sep_1$ and 
$s_1 \asymp \frac{1}{\sep_1} + (\mu_2^2 - \mu_1^2)^{-1/2}\log^{1/2}(\frac{K^{3/2} \kerfnvar_2}{\kerfnvar_1 \varepsilon_1})$. 
Finally, 
\begin{align} \label{eq:twowin_error_1_ord}
d_w(\widehat t_{1,\phi_1(j)}, t_{1,j}) \leq \varepsilon_1, \quad 
\abs{\widehat u_{1,\phi_1(j)} - u_{1,j}}\lesssim \left(\frac{K}{\sep_1} + \frac{1}{(\kerfnvar_2^2 - \kerfnvar_1^2)^{1/2}} 
\log^{1/2}\left(\frac{K^{3/2} \kerfnvar_2}{\kerfnvar_1\varepsilon_1}\right)\right) \varepsilon_1.
\end{align}
\end{itemize}
%
\paragraph{Condition on $\varepsilon_1,\varepsilon_2$.}
It is not difficult to verify that a sufficient condition for \eqref{eq:twowin_eps12cond_ord} to 
hold is that for any given $\theta \in (0,1/2)$, it holds that
\begin{equation} \label{eq:eps_12_two_win}
  \varepsilon_1 \lesssim \varepsilon_2^{\frac{1}{1-\theta}} C(\kerfnvar_1,\kerfnvar_2,\sep_1,\sep_2,K,\theta), 
\end{equation}
where $C(\kerfnvar_1,\kerfnvar_2,\sep_1,\sep_2,K,\theta) > 0$ depends only on the indicated parameters. 
This is outlined in Appendix \ref{app:sec:conditions_eps_gen} for Theorem \ref{thm:gen_case_main} for the case of $L$ kernels.
In other words, $\varepsilon_1$ would have to be sufficiently small with respect to $\varepsilon_2$, 
so that the estimation errors carrying forward from the first group to the estimation of the parameters 
for the second group, are controlled. 
%
%
%
\paragraph{Effect of separation between $\kerfnvar_1,\kerfnvar_2$.}
Note that as $\kerfnvar_1 \rightarrow \kerfnvar_2$, then \eqref{eq:twowin_eps12cond_ord} becomes 
more and more difficult to satisfy; in particular, $C(\kerfnvar_1,\kerfnvar_2,\sep_1,\sep_2,K,\theta) \rightarrow 0$ in 
\eqref{eq:eps_12_two_win}. Hence, we would have to sample sufficiently deep in the tail of $\kerfour_1$ in order 
to distinguish $\kerfour_1,\kerfour_2$ as one would intuitively expect. 
Next, for fixed $\kerfnvar_2$ as $\kerfnvar_1 \rightarrow 0$, we see 
that \eqref{eq:twowin_eps12cond_ord} becomes easier to satisfy. 
This is because $\kerfour_1(s)$ is now small for all $s$, and hence the perturbation error arising from 
stage $1$ reduces accordingly. However, notice that $s_1$ now has to increase correspondingly in order to 
distinguish between $\kerfour_1,\kerfour_2$ (since $\kerfour_1(s) \approx 0$ for all $s$).
Therefore, in order to control the estimation error of the amplitudes (see \eqref{eq:twowin_error_1_ord}), 
$\varepsilon_1$ now has to reduce accordingly. For instance, $\varepsilon_1 = o(\kerfnvar_1^{1/3})$ suffices.
On the other hand, for fixed $\kerfnvar_1$, as $\kerfnvar_2 \rightarrow \infty$, 
satisfying \eqref{eq:twowin_eps12cond_ord} becomes more and more difficult. 
This is because the tail of $\kerfour_2$ becomes thinner, and so, the deconvolution step at stage $2$ blows up 
the error arising from stage $1$. 

\begin{proof}[Proof of Theorem \ref{thm:two_win_case}]
The proof is divided into two steps below.
\begin{itemize}
%
%
\item \textbf{Recovering source parameters for first group.}
For offset parameter $s_1\in \mathbb Z^{+}$ (the choice of which will be made clear later), we obtain the samples 
$(f(s_1+i))_{i=-m_1}^{m_1-1}$. Now, for any $i=-m_1,\ldots,m_1-1$, we have that 
\begin{align}
\frac{f(s_1+i)}{\kerfour_1(s_1+i)} & = 
\sum_{j=1}^K u_{1,j}\ \exp\left(\iota 2\pi (s_1+i)t_{1,j}\right)+ 
\frac{\kerfour_2(s_1+i)}{\kerfour_1(s_1+i)} \sum_{j=1}^K \ u_{2,j} \exp \left(\iota 2\pi (s_1+i)t_{2,j}\right) \nonumber \\
& = \sum_{j=1}^K \underbrace{u_{1,j} \exp\left( \iota 2\pi s_1 t_{1,j}\right)}_{u'_{1,j}} \exp\left( \iota 2\pi i t_{1,j}\right)
\nonumber \\
& \hspace{2cm} + 
\underbrace{\frac{\kerfnvar_2}{\kerfnvar_1}\exp\left(-2\pi^2(s_1+i)^2(\kerfnvar_2^2-\kerfnvar_1^2) \right)\sum_{j=1}^K u_{2,j} \exp (\iota 2\pi (s_1+i)t_{2,j})}_{\noiseint_{1,i}}
\nonumber \\
& = \sum_{j=1}^K u'_{1,j} \exp\left(\iota 2\pi i t_{1,j}\right)+\noiseint_{1,i}. \label{eq:twowin_proof_temp1}
\end{align}
Here, $\noiseint_{1,i} \in \mathbb{C}$ corresponds to ``perturbation'' arising from the tail of $f_2(s)$. Since the stated 
choice of $s_1$ implies $s_1 > m_1$, this means $\min_{i} (s_1 + i)^2 = (s_1-m_1)^2$, and hence clearly 
\begin{align*}
\abs{\noiseint_{1,i}} \leq \frac{\kerfnvar_2}{\kerfnvar_1}\exp\left(-2\pi^2(s_1-m_1)^2(\kerfnvar_2^2-\kerfnvar_1^2)\right) K\umax, \quad i=-m_1,\ldots,m_1-1. 
\end{align*}
From \eqref{eq:twowin_proof_temp1}, we can see that $\widetilde {H}_0^{(1)} = V_1 D_{u'_1} V_1^H + E^{(1)}$
and $\widetilde {H}_1^{(1)} = V_1 D_{u'_1} D_{\alpha_1} V_1^H+ F^{(1)}$. Here, $D_{u'_1} = \diag(u'_{1,1},\dots,u'_{1,K})$ and 
$D_{\alpha_1} = \diag(\alpha_{1,1},\dots,\alpha_{1,K})$, while $E^{(1)}, F^{(1)}$ denote the perturbation matrices 
consisting of $(\noiseint_{1,i})_{i=-m_1}^{m_1-1}$ terms, as in \eqref{eq:mp_matrix_0}, \eqref{eq:mp_matrix_1}.

We obtain estimates $\widehat t_{1,j}$, $\est{u}_{1,j}$, $j=1,\ldots,K$ via the MMP method.
Invoking Corollary \ref{corr:moitra_MP}, we have for 
$\varepsilon_1 < c \sep_1/2$, and $\frac{2}{\sep_1 - 2\varepsilon_1} + 1 \leq m_1 < \frac{2}{\sep_1(1-c)} + 1 (= \mup{1})$ that if 
$s_1$ satisfies
\begin{align} \label{eq:twowin_proof_temp11}
\frac{\kerfnvar_2}{\kerfnvar_1}\exp\left(-2\pi^2(s_1-m_1)^2(\kerfnvar_2^2-\kerfnvar_1^2)\right) K\umax 
\leq \hemant{\varepsilon_1 \umin B(\urel,K),} 
\end{align}
then there exists a permutation $\phi_1: [K] \rightarrow [K]$ such that for each $j=1,\dots,K$,  
\begin{align}
d_w(\widehat t_{1,\phi_1(j)}, t_{1,j}) \leq \varepsilon_1, \quad
\abs{\widehat u_{1,\phi_1(j)} - u_{1,j}} 
< \hemant{\left(\widetilde{C}_1 + 2\pi s_1\right) \umax \varepsilon_1}, \label{eq:twowin_proof_temp7}
\end{align}
where \hemant{$\widetilde{C}_1 = \widetilde{C}(\sep_1,c,K,\urel) = 4\pi K \mup{1} + \frac{2}{C\sqrt{K}}\left(\urel + 16\urel^2\right)^{-1}$}.
Clearly, the condition 
\begin{equation*}
s_1 \geq m_1 + \frac{1}{(2\pi^2(\kerfnvar_2^2 - \kerfnvar_1^2))^{1/2}} 
\hemant{\log^{1/2}\left(\frac{K\urel\kerfnvar_2}{\kerfnvar_1 \varepsilon_1 B(\urel,K)} \right)}
\end{equation*}
implies \eqref{eq:twowin_proof_temp11}. Moreover, since $s_1 \leq \widetilde{c} S_1$ and $m_1 < \mup{1}$, we obtain 
\begin{align}
s_1 < \widetilde{c} \left(\mup{1} + \frac{1}{(2\pi^2(\kerfnvar_2^2 - \kerfnvar_1^2))^{1/2}} 
\hemant{\log^{1/2}\left(\frac{K \urel \kerfnvar_2}{\kerfnvar_1 \varepsilon_1 B(\urel,K)} \right)}\right). 
\label{eq:twowin_proof_temp6}
\end{align}
Plugging \eqref{eq:twowin_proof_temp6} into \eqref{eq:twowin_proof_temp7} leads to the bound
\begin{align}
\abs{\widehat u_{1,\phi_1(j)} - u_{1,j}} 
< \left(\bar{C}_1 + \bar{C}_2 \log^{1/2}\left(\frac{\bar{C}_3}{\varepsilon_1}\right)\right) \hemant{\umax} \varepsilon_1; \quad j=1,\dots,K, \label{eq:twowin_proof_temp8}
\end{align}
where $\bar{C}_1, \bar{C}_2, \bar{C}_3 > 0$ are constants defined as follows.
\begin{align}
\bar{C}_1 &= \hemant{\widetilde{C}_1 + 2\pi\widetilde{c}\mup{1}}, \label{eq:twowin_proof_temp9} \\
\bar{C}_2 &= \hemant{\frac{2\pi\widetilde{c}}{(2\pi^2(\kerfnvar_2^2 - \kerfnvar_1^2))^{1/2}}}, \quad 
\bar{C}_3 = \hemant{\frac{K \urel \kerfnvar_2}{\kerfnvar_1 B(\urel,K)}}. \label{eq:twowin_proof_temp10}
\end{align}
%

%
\item \textbf{Recovering source parameters for second group.}
Let $\est{f}_1$ denote the estimate of $f_1$ obtained using the estimates $(\est{u}_{1,j})_{j=1}^K, (\est{t}_{1,j})_{j=1}^K$, defined as
\begin{align*}
\est{f}_1(s) = \kerfour_1(s) \sum_{j=1}^K \est{u}_{1,j} \exp \left(\iota 2\pi \est{t}_{1,j} s\right).
\end{align*}
For suitable $s_2, m_2 \in \mathbb Z_+$ (choice to be made clear later), we now obtain samples  
\begin{align*}
	\frac{f(s_2+i) - \est{f}_1(s_2+i)}{\kerfour_2(s_2+i)}; \quad  i=-m_2,\ldots,m_2-1.
\end{align*}
Let us note that  
\begin{align}
	& \frac{f(s_2+i)-\est{f}_1(s_2+i)}{\kerfour_2(s_2+i)} 
	 = \frac{\kerfour_1(s_2+i)}{\kerfour_2(s_2+i)} \sum_{j=1}^K\left(u_{1,j} \exp\left(\iota 2\pi (s_2+i) t_{1,j} \right)- 
	\widehat u_{1,\perm_1(j)} \exp \left(\iota 2\pi (s_2+i)\widehat t_{1,\perm_1(j)}\right) \right) \nonumber \\
	& \hspace{1cm} +\sum_{j=1}^K u_{2,j} \exp \left(\iota 2\pi (s_2 + i) t_{2,j} \right) \nonumber \\ 
	&= \underbrace{\frac{\kerfnvar_1}{\kerfnvar_2} \exp(2\pi^2 (s_2+i)^2 (\kerfnvar_2^2 - \kerfnvar_1^2)) \sum_{j=1}^K\left(u_{1,j} \exp(\iota 2\pi (s_2+i)t_{1,j}) - 
	\widehat u_{1,\perm_1(j)} \exp \left(\iota 2\pi (s_2+i)\widehat t_{1,\perm_1(j)}\right) \right)}_{\noiseint_{2,i}} \nonumber \\
	& \hspace{1cm} + \sum_{j=1}^K \underbrace{u_{2,j} \exp(\iota 2\pi s_2 t_{2,j})}_{u'_{2,j}} \exp(\iota 2\pi i t_{2,j}) \nonumber \\
	&=  \sum_{j=1}^K u'_{2,j} \exp(\iota 2\pi i t_{2,j}) + \noiseint_{2,i}. \label{eq:twowin_proof_temp3}
\end{align}
Here, $\noiseint_{2,i} \in \mathbb{C}$ corresponds to noise arising from the estimation errors for the parameters in the 
first group of sources. As a direct consequence of Proposition \ref{app:prop_useful_res_2}, we have for each $j=1,\dots,K$ that   
\begin{align}
& \abs{u_{1,j} \exp(\iota 2\pi (s_2+i)t_{1,j}) - \est{u}_{1,\perm_1(j)} \exp (\iota 2\pi (s_2+i)\widehat t_{1,\perm_1(j))}} \nonumber\\
&\leq 2\pi\umax \abs{s_2 + i} d_w(t_{1,j}, \est{t}_{1,\perm_1(j)}) + \abs{u_{1,j} - \est{u}_{1,\perm_1(j)}} \nonumber \\
&< 2\pi \umax \abs{s_2+i} \varepsilon_1 + \left(\bar{C}_1 + 
\bar{C}_2 \log^{1/2}\left(\frac{\bar{C}_3}{\varepsilon_1}\right)\right) \hemant{\umax}\varepsilon_1, \label{eq:twowin_proof_temp4}
\end{align}
where the last inequality above follows from the bounds on $\abs{u_{1,j} - \est{u}_{1,\perm_1(j)}}$, $d_w(t_{1,j}, t_{1,\perm_1(j)})$, 
derived earlier. Now for $s_2 = 0$, and using the fact $\abs{i} \leq m_2 < \frac{2}{\sep_2(1-c)} + 1 \ (\ = \mup{2})$, we obtain from 
\eqref{eq:twowin_proof_temp4} the following uniform bound on $\abs{\noiseint_{2,i}}$.
\begin{align} 
\abs{\noiseint_{2,i}} 
&< \frac{\kerfnvar_1}{\kerfnvar_2} K e^{(2\pi^2 m_2^2 (\kerfnvar_2^2 - \kerfnvar_1^2))} 
\left(\hemant{2\pi m_2} + \bar{C}_1 + \bar{C}_2 \log^{1/2}\left(\frac{\bar{C}_3}{\varepsilon_1}\right)\right) \hemant{\umax} \varepsilon_1 \nonumber \\
&< \frac{\kerfnvar_1}{\kerfnvar_2} K e^{(2\pi^2 \mup{2}^2 (\kerfnvar_2^2 - \kerfnvar_1^2))} 
\left(\hemant{2\pi \mup{2}} + \bar{C}_1 + \bar{C}_2 \log^{1/2}\left(\frac{\bar{C}_3}{\varepsilon_1}\right)\right) \hemant{\umax} \varepsilon_1. \label{eq:twowin_proof_temp5}
\end{align}
From \eqref{eq:twowin_proof_temp3}, we see that $\widetilde {H}_0^{(2)} = V_2 D_{u'_2} V_2^H + E^{(2)}$, 
$\widetilde {H}_1^{(2)} = V_2 D_{u'_2} D_{\alpha_2}V_2^H +F^{(2)}$. Here, $D_{u'_2} = \diag(u'_{2,1},\dots,u'_{2,K})$ and 
$D_{\alpha_2} = \diag(\alpha_{2,1},\dots,\alpha_{2,K})$, while $E^{(2)}, F^{(2)}$ denote the perturbation matrices 
consisting of $(\noiseint_{2,i})_{i=-m_2}^{m_2-1}$ terms, as in \eqref{eq:mp_matrix_0}, \eqref{eq:mp_matrix_1}.

We obtain the estimates $(\widehat t_{2,j})_{j=1}^K$ and $(\est{u}_{2,j})_{j=1}^K$ using the MMP method. 
Invoking Corollary \ref{corr:moitra_MP} and assuming $\varepsilon_2 < c \sep_2/2$, 
$\frac{2}{\sep_2 - 2\varepsilon_2} + 1 \leq m_2 < \mup{2}$ hold, it is sufficient that 
$\varepsilon_1$ satisfies the condition
\begin{align*}
\frac{\kerfnvar_1}{\kerfnvar_2} K e^{(2\pi^2 \mup{2}^2 (\kerfnvar_2^2 - \kerfnvar_1^2))} 
\left(\hemant{2\pi \mup{2}} + \bar{C}_1 + 
\bar{C}_2 \log^{1/2}\left(\frac{\bar{C}_3}{\varepsilon_1}\right)\right) \hemant{\umax} \varepsilon_1 
\leq \varepsilon_2 \hemant{\umin B(\urel,K)}. 
\end{align*}
Indeed, there then exists a permutation $\perm_2: [K] \rightarrow [K]$ such that for each $j=1,\dots,K$, 
\begin{align*}
d_w(\widehat t_{2,\phi_2(j)}, t_{2,j}) \leq \varepsilon_2, \quad \abs{\widehat u_{2,\phi_2(j)} - u_{2,j}} < \widetilde{C}_2 \varepsilon_2, 
\end{align*}
where $\hemant{\widetilde{C}_2 = \widetilde{C}(\sep_2,c,K,\urel) = 4\pi K \mup{2} + \frac{2}{C\sqrt{K}}\left(\urel + 16\urel^2 \right)^{-1}}$.
This completes the proof.
\end{itemize}
\end{proof}
%
%
\subsection{The general case} \label{sec:gen_noiseless_case}
We now move to the general case where $L \geq 1$. The function $f$ is now of the form
\begin{align*}
f(s) & = \sum_{l=1}^{L} \kerfour_l(s) \left(\sum_{j=1}^K u_{l,j} \exp(\iota 2\pi s t_{l,j})\right)
\end{align*}
where we recall that $\kerfour_l(s) = \sqrt{2\pi} \ \kerfnvar_l \exp (-2\pi^2 s^2 \kerfnvar_l^2)$. 
Before stating our result, it will be helpful to define certain terms for ease of notation, later on. 
\begin{enumerate}[label=\textnormal{(\arabic*)}]
\item For $l=1,\dots,L$,
%
\begin{align} \label{eq:gen_proof_temp26}
\mup{l} := \frac{2}{\sep_l(1-c)} + 1, \quad 
\hemant{\widetilde{C}_l := \left[4\pi K \mup{l} + \frac{2}{C\sqrt{K}}\left(\urel + 16\urel^2 \right)^{-1} \right]}.
\end{align}
with constants $c \in (0,1)$ and $C = 10 + \frac{1}{2\sqrt{2}}$ (from \hemant{Corollary \ref{corr:moitra_MP}}). 

\item \label{itm:defs_mainthm_terms_2} For $l=1,\dots,L-1$, and a constant $\widetilde{c} > 1$,  
\begin{align}
\bar{C}_{l,1} &:= \widetilde{C}_l + \hemant{2\pi\widetilde{c}\mup{l}},  \
\bar{C}_{l,2} := \frac{\hemant{2\pi\widetilde{c}}}{(2\pi^2(\kerfnvar_{l+1}^2 - \kerfnvar_l^2))^{1/2}}, \label{eq:gen_proof_temp8} 
\\ 
D_l &:= \hemant{\frac{K \urel (L-l)\kerfnvar_L}{\kerfnvar_l B(\urel,K)}}, \quad
\bar{C}_{l,3} := \left\{
\begin{array}{rl}
D_l \quad ; & l = 1 \\
2 D_l  \quad ; & l > 1.
\end{array} \right. \label{eq:gen_proof_temp9}
\end{align}
\hemant{where $B(\urel,K)$ is as defined in Corollary \ref{corr:moitra_MP}.}

\item For $l = 1,\dots,L$, define
\begin{align} \label{eq:gen_proof_temp27}
E_l(\varepsilon) := 
\left\{
\begin{array}{rl}
\hemant{\left(\bar{C}_{l,1} + \bar{C}_{l,2} \log^{1/2}\left(\frac{\bar{C}_{l,3}}{\varepsilon}\right)\right)}  \quad ; & l < L \\
\hemant{\widetilde{C}_L} \quad ; & l = L.
\end{array} \right.
\end{align}
\hemant{where $\varepsilon \in (0,1)$}.

\item For $l = 2,\dots,L-1$, define
\begin{align}
F_l(\varepsilon) &:= C_{l,1}' + C_{l,2}' \log^{1/2}\left(\frac{2 D_l}{\varepsilon}\right); \label{eq:gen_proof_temp28} \\
\text{where} \quad C_{l,1}' &:= \hemant{2\pi}(\widetilde{c} + 1)\mup{l}, \ C_{l,2}' = \hemant{\bar{C}_{l,2}}. \label{eq:gen_proof_temp29}
\end{align}
Here, $\widetilde{c} > 1$ is the same constant as in \ref{itm:defs_mainthm_terms_2}.
\end{enumerate}
We are now ready to state our main theorem for approximate recovery of the 
source parameters for each group.
%
%
%
%
\begin{theorem} \label{thm:gen_case_main}
For a constant $c \in (0,1)$, let $0 < \varepsilon_L < c\sep_L/2$, $\frac{2}{\sep_L - 2\varepsilon_L} \leq m_L < \mup{L}$ and $s_L = 0$. 
Moreover, for $l = L-1, \dots, 1$, say we choose $\varepsilon_l,m_l,s_l$ as follows.
\begin{enumerate} 
\item $\hemant{0 <  \varepsilon_l < c\sep_l/2}$ additionally satisfies the following conditions.
\begin{enumerate} 

\item\label{itm:main_thm_epscond_2} \hemant{$(2\pi \mup{L} + E_{L-1}(\varepsilon_{L-1})) \varepsilon_{L-1} \urel
\leq \varepsilon_{L} e^{-2\pi^2(\kerfnvar_{L}^2-\kerfnvar_1^2)\mup{L}^2} \frac{\kerfnvar_{L} B(\urel,K) }{K (L-1) \kerfnvar_{L-1}}$}. 

\item\label{itm:main_thm_epscond_3} If $l < L-1$, then $\varepsilon_{l} \leq \varepsilon_{l+1}$, 
$E_l(\varepsilon_{l}) \hemant{\varepsilon_l} \leq E_{l+1}(\varepsilon_{l+1}) \hemant{\varepsilon_{l+1}}$ and 
\begin{equation} \label{eq:app_epsconds1c_3}
\hemant{(F_{l+1}(\varepsilon_{l}) + E_{l}(\varepsilon_{l})) \varepsilon_l \urel
\leq \varepsilon_{l+1} e^{-(\kerfnvar_{l+1}^2-\kerfnvar_1^2)\frac{F_{l+1}^2(\varepsilon_{l+1})}{2}} \frac{\kerfnvar_{l+1} B(\urel,K)}{2K l \kerfnvar_{l}}.}
\end{equation}  
%
\end{enumerate}

\item $\frac{2}{\sep_l - 2\varepsilon_l} \leq m_l < \mup{l}$, and 
$S_l \leq s_l \leq \widetilde{c} S_l$ (for constant $\widetilde{c} > 1$) where
\begin{align*}
S_l = m_l + \frac{1}{(2\pi^2(\kerfnvar_{l+1}^2 - \kerfnvar_l^2))^{1/2}} 
\log^{1/2}\left(\hemant{\frac{b_l K \urel (L-l) \kerfnvar_L}{\varepsilon_l B(\urel,K) \kerfnvar_l}} \right),
\end{align*}
with $b_l = \hemant{1}$ if $l = 1$, and $b_l = \hemant{2}$ otherwise. 
\end{enumerate}
Then, for each $l=1,\dots,L$, there exists a permutation $\perm_l: [K] \rightarrow [K]$ such that 
\begin{align*}
d_w(\widehat t_{l,\perm_l(j)}, t_{l,j}) &\leq \varepsilon_l, \quad 
\abs{\widehat u_{l,\perm_l(j)} - u_{l,j}} < E_l(\varepsilon_l) \hemant{\varepsilon_l \umax}; \quad j=1,\dots,K.
\end{align*}
\end{theorem}
%
%
\paragraph{Interpreting Theorem \ref{thm:gen_case_main}.} Before proceeding to the proof, we make some useful observations. 
\begin{itemize} 
\item [(a)] We first choose the sampling parameters ($\varepsilon,m,s$) for the outermost kernel $\kerfour_L$, then 
for $\kerfour_{L-1}$, and so on. For the $l^{th}$ group ($1 \leq l \leq L$), we first choose $\varepsilon_L$ (accuracy), 
then $m_l$ (number of samples), and finally $s_l$ (sampling offset).

\item [(b)] The choice of $\varepsilon_L \in (0,c\sep_L/2)$, while free, dictates the choice of 
$\varepsilon_1,\dots,\varepsilon_{L-1}$.  To begin with, condition \ref{itm:main_thm_epscond_2} essentially
requires $\varepsilon_{L-1}$ to be sufficiently small with respect to $\varepsilon_L$. Similarly, 
for $l = 1,\dots,L-2$, the conditions in \ref{itm:main_thm_epscond_3} require $\varepsilon_l$ to be sufficiently small with respect to 
$\varepsilon_{l+1}$. It ensures that during the estimation of the parameters for the $(l+1)^{th}$ group, 
the estimation errors carrying forward from the previous groups ($1$ to $l$) are sufficiently small.

\item [(c)] For each $l$ ($< L$), the lower bound on $s_l$ is to ensure that we are sufficiently deep in the tails 
of $\kerfour_{l+1}, \kerfour_{l+2}, \dots, \kerfour_L$. The upper bound on $s_l$ is to control the estimation errors of the 
source amplitudes for group $l$ (see \eqref{eq:mmp_corr_est_err}).
\end{itemize}

\paragraph{Order wise dependencies.} We now discuss the scaling of the terms involved, assuming $\umax,\umin \asymp 1$.
\begin{itemize}
\item[(i)] For $l=1,\dots,L$, we have $\mup{l} \asymp \frac{1}{\sep_l}$, $\widetilde{C}_l  \asymp \frac{K}{\sep_l}$.

\item[(ii)] For $p=1,\dots,L-1$ we have
\begin{align*}
\bar{C}_{p,1} \asymp \frac{K}{\sep_p}, \
\bar{C}_{p,2} \asymp \frac{1}{(\kerfnvar_{p+1}^2 - \kerfnvar_p^2)^{1/2}}, \
D_p \asymp \frac{K^{3/2}(L-p)\kerfnvar_L}{\kerfnvar_p}, \ \bar{C}_{p,3} \asymp D_p.
\end{align*}

\item[(iii)] For $l = 1,\dots,L$, we have
\begin{align} \label{eq:E_L_ord_dep}
E_l(\varepsilon_l) \asymp  
\left\{
\begin{array}{rl}
\hemant{\frac{K}{\sep_l} + \frac{1}{(\kerfnvar_{l+1}^2 - \kerfnvar_l^2)^{1/2}} 
\log^{1/2}\left(\frac{K^{3/2}(L-l)\kerfnvar_L}{\varepsilon_l\kerfnvar_l}\right)}   \quad ; & l < L \\
\\
 \frac{K}{\sep_L} \quad ; & l = L.
\end{array} \right.
\end{align}

\item[(iv)] For $q = 2,\dots,L-1$, we have
\begin{align} \label{eq:orddep_F_gen}
F_q(\varepsilon_q) \asymp \frac{1}{\sep_q} + \frac{1}{(\kerfnvar_{q+1}^2 - \kerfnvar_q^2)^{1/2}} \log^{1/2}\left(\frac{K^{3/2} (L-q) \kerfnvar_L}{\varepsilon_q\kerfnvar_q}\right).
\end{align}
\end{itemize}
%
%
\paragraph{Conditions on $\varepsilon_l$.} 
Theorem \ref{thm:gen_case_main} has several conditions on $\varepsilon_l$, which might be difficult to digest at first glance. 
On a top level, the conditions dictate that the accuracies 
satisfy $\varepsilon_1 \leq \varepsilon_2 \leq \cdots \leq \varepsilon_L$. In fact, they require a stronger condition in the 
sense that for each $1 \leq l \leq L-1$, $\varepsilon_l$ is required to be sufficiently smaller than $\varepsilon_{l+1}$ (the choice 
of $\varepsilon_L \lesssim \sep_L/2$ is free). This places the strongest assumption on $\varepsilon_1$ meaning that the source parameters 
corresponding to the ``outermost kernel'' in the Fourier domain should be estimated with the highest accuracy. Below, we state 
the conditions appearing on $\varepsilon_l$ in the Theorem up to positive constants; the details are deferred to 
Appendix \ref{app:sec:conditions_eps_gen}.
\begin{enumerate}
\item Condition \ref{itm:main_thm_epscond_2} in Theorem \ref{thm:gen_case_main} holds 
if $\varepsilon_{L-1}, \varepsilon_L$ satisfy
\begin{align} \label{eq:eps_main_alpha}
\varepsilon_{L-1} \lesssim \alpha (\varepsilon_L)^{\frac{1}{1-\theta}}
\end{align}
for any given $\theta \in (0,1/2)$. Here, $\alpha > 0$ depends on 
$\theta,K,L,\sep_{L-1},\sep_{L},\kerfnvar_{L},\kerfnvar_{L-1},\kerfnvar_{1}$.

\item For $l < L - 1$, let us look at condition \ref{itm:main_thm_epscond_3} in Theorem \ref{thm:gen_case_main}. 
The requirement $E_l(\varepsilon_l) \hemant{\varepsilon_l} \leq E_{l+1}(\varepsilon_{l+1}) \hemant{\varepsilon_{l+1}}$ 
holds if $\varepsilon_l,\varepsilon_{l+1}$ satisfy
\begin{equation} \label{eq:eps_main_lambda}
\varepsilon_l \lesssim \lambda_l \log^{\frac{1}{2(1-\theta)}} \left(\frac{K^{3/2}\kerfnvar_{L} \hemant{(L-l)}}{\varepsilon_{l+1}\kerfnvar_{l+1}}\right) \varepsilon_{l+1}^{\frac{1}{1-\theta}}
\end{equation}
for any given $\theta \in (0,1/2)$. Here, $\lambda_l > 0$ depends on 
$\sep_l,\sep_{l+1},\kerfnvar_{l},\kerfnvar_{l+1},\kerfnvar_{l+2},L,K,\theta$.
 
Furthermore, the condition in \eqref{eq:app_epsconds1c_3} is satisfied if $\varepsilon_l, \varepsilon_{l+1}$ satisfy
\begin{equation} \label{eq:eps_main_beta}
\varepsilon_l \lesssim \beta_l (\varepsilon_{l+1})^{\frac{1+\gamma_l}{1-\theta}}
\end{equation}
for any given $\theta \in (0,1/2)$. Here, $\beta_l > 0$ depends on 
$L,K,\sep_l$, $\sep_{l+1},\sep_{L-1},\sep_{L}$, $\kerfnvar_{l},\kerfnvar_{l+1}$, $\kerfnvar_{l+2},
\kerfnvar_{L}$, $\kerfnvar_{L-1},\kerfnvar_{1},\theta$, while 
$\gamma_l > 0$ depends on $\kerfnvar_1,\kerfnvar_{l+1},\kerfnvar_{l+2},\sep_{l+1}$. 
Note that the dependence on $\varepsilon_{l+1}$ is stricter in \eqref{eq:eps_main_beta} as compared to 
\eqref{eq:eps_main_lambda}.
\end{enumerate}
%
%
%
\paragraph{Effect of separation between $\kerfnvar_l, \kerfnvar_{l+1}$ for $l = 1,\dots,L-1$.} 
The interaction between $\kerfnvar_{L-1},\kerfnvar_L$ occurs in the same manner as explained for the case of two kernels, 
the reader is invited to verify this. 
We analyze the interaction between $\kerfnvar_{l},\kerfnvar_{l+1}$ below for $l < L-1$.
\begin{enumerate}

\item Consider the scenario where $\kerfnvar_l \rightarrow \kerfnvar_{l+1}$ (with other terms fixed). 
We see that $s_l$ has to be suitably large now in order to be able to distinguish between $\kerfour_l$ 
and $\kerfour_{l+1}$. Moreover, conditions \eqref{eq:eps_main_lambda}, \eqref{eq:eps_main_beta} 
become stricter in the sense that $\lambda_l,\beta_l \rightarrow 0$. 

\item Now say $\kerfnvar_{l+1}$ is fixed, and $\kerfnvar_{l} \rightarrow 0$ 
(and hence $\kerfnvar_1,\dots,\kerfnvar_{l-1} \rightarrow 0$). 
In this case, the conditions $E_1(\varepsilon_1) \hemant{\varepsilon_l} \leq \cdots 
\leq E_{l+1}(\varepsilon_{l+1}) \hemant{\varepsilon_{l+1}}$ become 
vacuous as the estimation error arising from stages $1,\dots,l-1$ themselves approach $0$.
However, $s_1,\dots,s_l$ now increase accordingly in order to distinguish within $\kerfour_1,\dots,\kerfour_l$. 
Hence, to control the estimation error of the amplitudes, i.e, $E_i(\varepsilon_i)$; $1 \leq i \leq l$, 
$\varepsilon_1,\dots,\varepsilon_l$ have to be suitably small. 
\end{enumerate}
%
\begin{proof}[Proof of Theorem \ref{thm:gen_case_main}]
The proof is divided in to three main steps.
\begin{itemize}
\item \textbf{Recovering source parameters for first group.} 
For $i=-m_1,\ldots,m_1-1$, we have  
\begin{align}
\frac{f(s_1+i)}{\kerfour_1(s_1+i)} &= \sum_{j=1}^K u_{1,j} \exp\left(\iota 2\pi (s_1+i)t_{1,j}\right)  
+ \sum_{l=2}^{L}\frac{\kerfour_l(s_1+i)}{\kerfour_1(s_1+i)} \sum_{j=1}^K \ u_{l,j} \exp \left(\iota 2\pi (s_1+i)t_{l,j}\right)
\nonumber \\
& = \sum_{j=1}^K \underbrace{u_{1,j} \exp\left( \iota 2\pi s_1 t_{1,j}\right)}_{u'_{1,j}} \exp\left( \iota 2\pi i t_{1,j}\right)
\nonumber \\
& \hspace{2cm} + \underbrace{\sum_{l=2}^{L} \frac{\kerfnvar_l}{\kerfnvar_1}\exp\left(-2\pi^2(s_1+i)^2(\kerfnvar_l^2-\kerfnvar_1^2) \right)
\sum_{j=1}^K u_{l,j} \exp (\iota 2\pi (s_1+i)t_{l,j})}_{\noiseint_{1,i}} \nonumber \\
& = \sum_{j=1}^K u'_{1,j} \exp\left(\iota 2\pi i t_{1,j}\right) + \noiseint_{1,i}. \label{eq:gen_proof_temp1}
\end{align}
Here, $\noiseint_{1,i}$ is the perturbation due to the tail of $\kerfour_2,\kerfour_3,\dots,\kerfour_L$.
Since the stated choice of $s_1$ implies $s_1 > m_1$, this means $\min_{i} (s_1 + i)^2 = (s_1-m_1)^2$, and hence clearly 
\begin{align*}
\abs{\noiseint_{1,i}} 
&\leq K\umax \sum_{l=2}^{L} \frac{\kerfnvar_l}{\kerfnvar_1}\exp\left(-2\pi^2(s_1-m_1)^2(\kerfnvar_l^2-\kerfnvar_1^2)\right) \\ 
&\leq K\umax (L-1) \frac{\kerfnvar_{L}}{\kerfnvar_1}\exp\left(-2\pi^2(s_1-m_1)^2(\kerfnvar_2^2-\kerfnvar_1^2)\right).
\end{align*}
We obtain estimates $\widehat t_{1,j}$, $\est{u}_{1,j}$, $j=1,\ldots,K$ via the MMP method.
Invoking Corollary \ref{corr:moitra_MP}, we have for $\varepsilon_1 < c \sep_1/2$, and 
$\frac{2}{\sep_1 - 2\varepsilon_1} + 1 \leq m_1 < \frac{2}{\sep_1(1-c)} + 1 (= \mup{1})$ that if 
$s_1$ satisfies 
\begin{align} \label{eq:gen_proof_temp2}
K\umax (L-1) \frac{\kerfnvar_{L}}{\kerfnvar_1}\exp\left(-2\pi^2(s_1-m_1)^2(\kerfnvar_2^2-\kerfnvar_1^2)\right)
\leq 
\varepsilon_1 \hemant{\umin B(\urel,K)}
\end{align}
then there exists a permutation $\phi_1: [K] \rightarrow [K]$ such that for each $j=1,\dots,K$,  
\begin{align}
d_w(\widehat t_{1,\phi_1(j)}, t_{1,j}) \leq \varepsilon_1, \quad
\abs{\widehat u_{1,\phi_1(j)} - u_{1,j}} 
< \left(\widetilde{C}_1 + 2 \pi s_1\right) \hemant{\umax} \varepsilon_1. \label{eq:gen_proof_temp3}
\end{align}
Clearly, the condition 
\begin{equation}
s_1 \geq m_1 + \frac{1}{(2\pi^2(\kerfnvar_2^2 - \kerfnvar_1^2))^{1/2}} 
\log^{1/2}\left(\hemant{\frac{K \urel (L-1) \kerfnvar_L}{\kerfnvar_1 \varepsilon_1 B(\urel,K)}  } \right)
\label{eq:low_s1}
\end{equation}
implies \eqref{eq:gen_proof_temp2}. Moreover, since $s_1 \leq \widetilde{c} S_1$ and  $m_1 < \mup{1}$, we obtain  
\begin{align}
s_1 &< \widetilde{c} \left(\mup{1} + \frac{1}{(2\pi^2(\kerfnvar_2^2 - \kerfnvar_1^2))^{1/2}} 
\log^{1/2}\left(\hemant{\frac{K (L-1) \urel \kerfnvar_L}{\kerfnvar_1 \varepsilon_1 B(\urel,K)}} \right)\right). 
\label{eq:gen_proof_temp4}
\end{align}
Plugging \eqref{eq:gen_proof_temp4} in \eqref{eq:gen_proof_temp3}, we obtain
\begin{align} \label{eq:gen_proof_temp5}
\abs{\widehat u_{1,\phi_1(j)} - u_{1,j}} 
< \left(\bar{C}_{1,1} + \bar{C}_{1,2} \log^{1/2}\left(\frac{\bar{C}_{1,3}}{\varepsilon_1}\right)\right) \varepsilon_1 \hemant{\umax} = E_1(\varepsilon_1) \hemant{\varepsilon_1 \umax}; \quad j=1,\dots,K, 
\end{align}
where $\bar{C}_{1,1}, \bar{C}_{1,2}, \bar{C}_{1,3} > 0$ are constants defined in \eqref{eq:gen_proof_temp8}, \eqref{eq:gen_proof_temp9}, 
and $E_p(\cdot)$ is defined in \eqref{eq:gen_proof_temp27}.
%
%

\item \textbf{Recovering source parameters for $l^{th} (1 < l < L)$ group.}
Say we are at the $l^{th}$ iteration for $1 < l < L$, having estimated the source parameters 
up to the $(l-1)^{th}$ group. 
Say that for each $p = 1,\dots,l-1$ and $j=1,\dots,K$ the following holds.
\begin{align} \label{eq:gen_proof_temp10}
d_w(\est{t}_{p,\phi_p(j)}, t_{p,j}) \leq \varepsilon_p, \quad  \abs{\est{u}_{p,\phi_p(j)} - u_{p,j}} < E_p(\varepsilon_p) \hemant{\umax \varepsilon_p}, 
\end{align}
for some permutations $\phi_p: [K] \rightarrow [K]$, with 
\begin{enumerate} 
\item $\varepsilon_1 \leq \cdots \leq \varepsilon_{l-1}; \quad \quad E_1(\varepsilon_1) \hemant{\varepsilon_1} \leq \cdots \le E_{l-1}(\varepsilon_{l-1}) \hemant{\varepsilon_{l-1}}$; \label{eq:gen_proof_temp20} 

\item $\varepsilon_{p} < c \sep_p/2$ ;  \label{eq:gen_proof_temp21} 

\item $\hemant{(F_{q+1}(\varepsilon_{q}) + E_{q}(\varepsilon_{q})) \urel \varepsilon_q 
\leq \varepsilon_{q+1} e^{-(\kerfnvar_{q+1}^2-\kerfnvar_1^2)\frac{F_{q+1}^2(\varepsilon_{q+1})}{2}} \frac{\kerfnvar_{q+1} B(\urel,K)}{2K q \kerfnvar_{q}}}$,  $1 \leq q \leq l-2$.
\end{enumerate}
For $i=-m_l,\ldots,m_l-1$, we have  
\begin{align*}
&\frac{f(s_l + i) - \sum_{p=1}^{l-1} \est{f}_p(s_l + i)}{\kerfour_l(s_l+i)} \\
&= \underbrace{\sum_{p=1}^{l-1}\frac{\kerfour_p(s_l+i)}{\kerfour_l(s_l + i)} \sum_{j=1}^K[u_{p,j}\exp(\iota2\pi(s_l+i)t_{p,j}) - \est{u}_{p,\perm_p(j)} \exp(\iota2\pi(s_l+i)\est{t}_{p,\perm_p(j)})]}_{\noiseint_{l,i,past}} \\
&+ \underbrace{\sum_{q=l+1}^{L} \frac{\kerfour_q(s_l+i)}{\kerfour_l(s_l + i)} \sum_{j=1}^{K}u_{q,j}\exp(\iota2\pi(s_l+i)t_{q,j})}_{\noiseint_{l,i,fut}} \\
&+ \sum_{j=1}^K \underbrace{u_{l,j} \exp(\iota2\pi s_l t_{l,j})}_{u'_{l,j}} \exp(\iota2\pi i t_{l,j}) 
= \sum_{j=1}^K u'_{l,j} \exp(\iota2\pi i t_{l,j}) + \noiseint_{l,i,past} + \noiseint_{l,i,fut}.
\end{align*}
Here, $\noiseint_{l,i,past}$ denotes perturbation due to the estimation errors of the source parameters in the past. 
Moreover, $\noiseint_{l,i,fut}$ denotes perturbation due to the tails of the kernels that are yet to be processed. 
\begin{enumerate}
\item[(i)] \underline{\textit{Bounding $\noiseint_{l,i,past}$.}} 
To begin with, note 
\begin{align*}
\noiseint_{l,i,past} = 
\sum_{p=1}^{l-1}\frac{\kerfnvar_p}{\kerfnvar_l} \exp(2\pi^2(\kerfnvar_l^2 - \kerfnvar_p^2) (s_l + i)^2) 
&\sum_{j=1}^K[u_{p,j}\exp(\iota2\pi(s_l+i)t_{p,j}) \\ 
&- \est{u}_{p,\perm_p(j)} \exp(\iota2\pi(s_l+i)\est{t}_{p,\perm_p(j)}].
\end{align*}
Using Proposition \ref{app:prop_useful_res_2}, we have for each $p=1,\dots,l-1$ and $j=1,\dots,K$ that   
\begin{align}
& \abs{u_{p,j} \exp(\iota 2\pi (s_l+i)t_{p,j}) - \est{u}_{p,\perm_p(j)} \exp (\iota 2\pi (s_l+i)\widehat t_{p,\perm_p(j)})} \nonumber \\
&\leq 2\pi\umax \abs{s_l + i} d_w(t_{p,j}, \est{t}_{p,\perm_p(j)}) + \abs{u_{p,j} - \est{u}_{p,\perm_p(j)}} \nonumber \\
&< 2\pi \umax \abs{s_l+i} \varepsilon_p + E_p(\varepsilon_p) \hemant{\varepsilon_p \umax}, \label{eq:gen_proof_temp23}
\end{align}
where the last inequality is due to \eqref{eq:gen_proof_temp10}.
Since $s_l > m_l$, hence $(s_l + i)^2 < (s_l + m_l)^2$ for all $i=-m_l,\dots,m_l-1$. 
With the help of \eqref{eq:gen_proof_temp23}, we then readily obtain 
\begin{align}
\abs{\noiseint_{l,i,past}} 
&< \sum_{p=1}^{l-1} \left(\left(\frac{\kerfnvar_p}{\kerfnvar_l} \exp(2\pi^2(\kerfnvar_l^2 - \kerfnvar_p^2) (s_l + m_l)^2)\right) 
(2\pi\umax (s_l+m_l) \varepsilon_p + E_p(\varepsilon_p) \hemant{\varepsilon_p \umax}) K\right) \nonumber \\
&\leq \left(K \frac{\kerfnvar_{l-1}}{\kerfnvar_l} e^{2\pi^2(\kerfnvar_l^2 - \kerfnvar_1^2)(s_l + m_l)^2} \right) 
\left(\sum_{p=1}^{l-1} 2\pi\umax(s_l+m_l)\varepsilon_p + E_p(\varepsilon_p) \hemant{\varepsilon_p \umax} \right) \nonumber \\
&\leq \left(K(l-1) \frac{\kerfnvar_{l-1}}{\kerfnvar_l} e^{2\pi^2(\kerfnvar_l^2 - \kerfnvar_1^2)(s_l + m_l)^2} \right) 
\left(2\pi\umax(s_l+m_l)\varepsilon_{l-1} + E_{l-1}(\varepsilon_{l-1}) \hemant{\varepsilon_{l-1} \umax} \right), \label{eq:gen_proof_temp11}
\end{align}
where in the last inequality, we used \ref{eq:gen_proof_temp20}.

\item[(ii)] \underline{\textit{Bounding $\noiseint_{l,i,fut}$.}} We have
\begin{align*}
\noiseint_{l,i,fut} = \sum_{q=l+1}^{L} \frac{\kerfnvar_q}{\kerfnvar_l} \exp(-2\pi^2(s_l+i)^2(\kerfnvar_q^2 - \kerfnvar_l^2)) 
\left(\sum_{j=1}^K u_{q,j} \exp(\iota2\pi(s_l+i)t_{q,j})\right).
\end{align*}
Since $s_l > m_l$, we have $(s_l+i)^2 \geq (s_l-m_l)^2$ for all $i=-m_l,\dots,m_l-1$. 
This, along with the fact $\frac{\kerfnvar_q}{\kerfnvar_l} \leq \frac{\kerfnvar_L}{\kerfnvar_l}$ gives us
\begin{align*}
\abs{\noiseint_{l,i,fut}} \leq \frac{\kerfnvar_L}{\kerfnvar_l} K\umax (L-l) \exp(-2\pi^2(s_l-m_l)^2(\kerfnvar_{l+1}^2 - \kerfnvar_l^2)).
\end{align*}
It follows that if $s_l \geq S_l$ where
\begin{align*}
S_l = m_l + \frac{1}{(2\pi^2(\kerfnvar_{l+1}^2 - \kerfnvar_l^2))^{1/2}} 
\log^{1/2}\left(\hemant{\frac{2 K \urel (L-l) \kerfnvar_L}{\varepsilon_l \kerfnvar_l B(\urel,K)}} \right)
\end{align*}
then for $i=-m_l,\dots,m_l-1$, we  have
\begin{align} \label{eq:gen_proof_temp12}
\abs{\noiseint_{l,i,fut}} < \varepsilon_l \hemant{\frac{\umin B(\urel,K)}{2}}.
\end{align}

\item[(iii)] \underline{\textit{Back to $\noiseint_{l,i,past}$.}}
We will now find conditions which ensure that the same bound as \eqref{eq:gen_proof_temp12} holds on 
$\abs{\noiseint_{l,i,past}}$, uniformly for all $i$. To this end, since $s_l \leq \widetilde{c} S_l$ and $m_l < \mup{l}$, we obtain 
\begin{align} \label{eq:gen_proof_temp24}
s_l < \widetilde{c} \left(\mup{l} + \frac{1}{(2\pi^2(\kerfnvar_{l+1}^2 - \kerfnvar_l^2))^{1/2}} 
\log^{1/2}\left(\hemant{\frac{2 K \urel (L-l) \kerfnvar_L}{\varepsilon_l \kerfnvar_l B(\urel,K)}}\right)\right). 
\end{align}
This then gives us the bound 
\begin{align*}
\hemant{2\pi}(s_l + m_l) 
&< \hemant{2\pi} (\widetilde{c} + 1) \mup{l} + \frac{\hemant{2\pi} \widetilde{c}}{(2\pi^2(\kerfnvar_{l+1}^2 - \kerfnvar_l^2))^{1/2}}
\log^{1/2}\left(\hemant{\frac{2 K \urel (L-l) \kerfnvar_L}{\varepsilon_l \kerfnvar_l B(\urel,K)}}\right) \\
&= C'_{l,1} + C'_{l,2} \log^{1/2}\left(\frac{2D_{l}}{\varepsilon_l}\right) = F_{l}(\varepsilon_l),
\end{align*}
where we recall the definition of $F_l$, and constants $C_{l,1}', C_{l,2}', D_l > 0$ from \eqref{eq:gen_proof_temp28}, \eqref{eq:gen_proof_temp29}.  
\hemant{Since $\varepsilon_l < 1 < D_l$  and $\varepsilon_l \geq \varepsilon_{l-1}$, 
hence} $\hemant{2\pi} (s_l + m_l) < F_{l}(\varepsilon_l) \leq F_{l}(\varepsilon_{l-1})$. Using this in 
\eqref{eq:gen_proof_temp11}, we obtain
\begin{align}
\abs{\noiseint_{l,i,past}} 
< \left(K(l-1) \frac{\kerfnvar_{l-1}}{\kerfnvar_l} \hemant{e^{(\kerfnvar_l^2 - \kerfnvar_1^2) \frac{F^2_l(\varepsilon_l)}{2}}} \right) 
\hemant{\left( F_l(\varepsilon_{l-1}) + E_{l-1}(\varepsilon_{l-1}) \right) \hemant{\umax\varepsilon_{l-1}}} \label{eq:gen_proof_temp13}
\end{align}
Therefore if $\varepsilon_{l-1}$ satisfies the condition 
\begin{align*}
 \hemant{(F_l(\varepsilon_{l-1}) + E_{l-1}(\varepsilon_{l-1}) ) \umax\varepsilon_{l-1} 
\leq \varepsilon_l e^{-(\kerfnvar_l^2-\kerfnvar_1^2) \frac{F^2_l(\varepsilon_l)}{2}} \frac{\umin\kerfnvar_l B(\urel,K)}{2 K (l-1) \kerfnvar_{l-1}}}
\end{align*}
then it implies 
$\abs{\noiseint_{l,i,past}} < \varepsilon_l \hemant{\frac{\umin B(\urel,K)}{2}}$.
Together with \eqref{eq:gen_proof_temp12}, this gives
\begin{align*}
\abs{\noiseint_{l,i}} \leq \abs{\noiseint_{l,i,fut}} + \abs{\noiseint_{l,i,past}} 
< \varepsilon_l \hemant{\umin B(\urel,K)}.
\end{align*}
We obtain estimates $\widehat t_{l,j}$, $\est{u}_{l,j}$, $j=1,\ldots,K$ via the MMP method. Invoking Corollary \ref{corr:moitra_MP}, if $\varepsilon_l < c \sep_l/2$, then for the stated choice of $m_l$, there exists a permutation $\phi_l: [K] \rightarrow [K]$ such that for each $j=1,\dots,K$,  
\begin{align}
d_w(\widehat t_{l,\perm_l(j)}, t_{l,j}) &\leq \varepsilon_l; \quad
\abs{\widehat u_{l,\perm_l(j)} - u_{l,j}} 
< \left(\widetilde{C}_l + 2\pi s_l \right) \varepsilon_l \hemant{\umax} < E_l(\varepsilon_l) \hemant{\varepsilon_{l} \umax}.   
\label{eq:gen_proof_temp15}
\end{align}
The last inequality follows readily using \eqref{eq:gen_proof_temp24}.
\end{enumerate}
%
%

\item \textbf{Recovering source parameters for last group.}
Say that for each $p = 1,\dots,L-1$ and $j=1,\dots,K$ the following holds.
\begin{align} \label{eq:gen_proof_temp18}
d_w(\est{t}_{p,\phi_p(j)}, t_{p,j}) \leq \varepsilon_p, \quad  
\abs{\widehat u_{p,\phi_p(j)} - u_{p,j}} <  E_p(\varepsilon_p) \hemant{\varepsilon_p \umax},
\end{align}
for some permutations $\phi_p: [K] \rightarrow [K]$, with
\begin{enumerate} 
\item $\varepsilon_1 \leq \cdots \leq \varepsilon_{L-1}$; 
$E_1(\varepsilon_1) \hemant{\varepsilon_1} \leq \cdots \leq E_{L-1}(\varepsilon_{L-1}) \hemant{\varepsilon_{L-1}}$; \label{eq:gen_proof_temp19} 

\item $\varepsilon_{p} < c \sep_p/2$;  \label{eq:gen_proof_temp22}

\item \hemant{$(F_{q+1}(\varepsilon_{q}) + E_{q}(\varepsilon_{q})) \varepsilon_{q} \urel 
\leq \varepsilon_{q+1} e^{-(\kerfnvar_{q+1}^2-\kerfnvar_1^2)\frac{F_{q+1}^2(\varepsilon_{q+1})}{2}} 
\frac{\kerfnvar_{q+1} B(\urel,K)}{2 K q \kerfnvar_{q}}$},  $1 \leq q \leq L-2$.
\end{enumerate}
We proceed by noting that for each $i = -m_L,\dots,m_L-1$
\begin{align*}
&\frac{f(s_L + i) - \sum_{p=1}^{L-1} \est{f}_p(s_L + i)}{\kerfour_L(s_L + i)} \\
&= \underbrace{\sum_{p=1}^{L-1}\frac{\kerfour_p(s_L+i)}{\kerfour_L(s_L + i)} 
\sum_{j=1}^K[u_{p,j}\exp(\iota2\pi(s_L+i)t_{p,j}) - \est{u}_{p,\perm_p(j)} \exp(\iota2\pi(s_L+i)\est{t}_{p,\perm_p(j)}]}_{\noiseint_{L,i}} \\
&+ \sum_{j=1}^K \underbrace{u_{L,j} \exp(\iota2\pi s_L t_{L,j})}_{u'_{L,j}} \exp(\iota2\pi i t_{L,j}) 
= \sum_{j=1}^K u'_{L,j} \exp(\iota2\pi i t_{L,j}) + \noiseint_{L,i}.
\end{align*}
%
%
Using Proposition \ref{app:prop_useful_res_2}, we have for each $p=1,\dots,L-1$ and $j=1,\dots,K$ that   
\begin{align}
& \abs{u_{p,j} \exp(\iota 2\pi (s_L+i)t_{p,j}) - \est{u}_{p,\perm_p(j)} \exp (\iota 2\pi (s_L+i)\widehat t_{p,\perm_p(j)}} \nonumber \\
&\leq 2\pi\umax \abs{s_L + i} d_w(t_{p,j}, \est{t}_{p,\perm_p(j)}) + \abs{u_{p,j} - \est{u}_{p,\perm_p(j)}} \nonumber \\
&< 2\pi \umax \abs{s_L+i} \varepsilon_p + E_p(\varepsilon_p) \hemant{\varepsilon_p \umax}, \label{eq:gen_proof_temp25}
\end{align}
where the last inequality follows from \eqref{eq:gen_proof_temp18}.
Since $s_L = 0$, hence $(s_L + i)^2 < m_L^2 < \mup{L}^2$ for all $i=-m_L,\dots,m_L-1$. 
Using \eqref{eq:gen_proof_temp25}, we then readily obtain 
\begin{align}
\abs{\noiseint_{L,i}} 
&< \sum_{p=1}^{L-1} \left(\left(\frac{\kerfnvar_p}{\kerfnvar_L} e^{2\pi^2(\kerfnvar_L^2 - \kerfnvar_p^2) \mup{L}^2}\right) 
(2\pi\umax\mup{L}\varepsilon_p + E_p(\varepsilon_p) \hemant{\varepsilon_p \umax})K\right) \nonumber\\
&\leq \left(K \frac{\kerfnvar_{L-1}}{\kerfnvar_L} e^{2\pi^2(\kerfnvar_L^2 - \kerfnvar_1^2)\mup{L}^2} \right) 
\hemant{\sum_{p=1}^{L-1}\left(2\pi\umax\mup{L}\varepsilon_p + E_p(\varepsilon_p) \varepsilon_p \umax \right)} \nonumber\\
&\leq \left(K(L-1) \frac{\kerfnvar_{L-1}}{\kerfnvar_L} e^{2\pi^2(\kerfnvar_L^2 - \kerfnvar_1^2)\mup{L}^2} \right) 
\left(2\pi\umax\mup{L}\varepsilon_{L-1} + E_{L-1}(\varepsilon_{L-1}) \hemant{\varepsilon_{L-1} \umax} \right) 
\label{eq:gen_proof_temp16}
\end{align}
where in the last inequality, we used \eqref{eq:gen_proof_temp19}.

We obtain estimates $\widehat t_{L,j}$, $\est{u}_{L,j}$, $j=1,\ldots,K$ via the MMP method.
Invoking Corollary \ref{corr:moitra_MP} and assuming $\varepsilon_L < c \sep_L/2$, 
it follows for the stated conditions on $m_L$, that it suffices if 
$\varepsilon_{L-1}$ satisfies 
\begin{align*}
\hemant{(2\pi\mup{L} + E_{L-1}(\varepsilon_{L-1})) \varepsilon_{L-1} \umax 
\leq \varepsilon_L e^{-2\pi^2(\kerfnvar_L^2 - \kerfnvar_1^2)\mup{L}^2} 
\frac{\umin \kerfnvar_L B(\urel,K)}{K (L-1) \kerfnvar_{L-1}}.}
\end{align*}
Indeed, there then exists a permutation $\phi_L: [K] \rightarrow [K]$ such that for each $j=1,\dots,K$,  
\begin{align}
d_w(\widehat t_{L,\perm_L(j)}, t_{L,j}) 
&\leq \varepsilon_L; \quad \abs{\widehat u_{L,\perm_L(j)} - u_{L,j}} 
< \widetilde{C}_L \varepsilon_L \hemant{\umax} 
= E_L(\varepsilon_L) \hemant{\varepsilon_{L} \umax}.   
\label{eq:gen_proof_temp17}
\end{align}
This completes the proof.
\end{itemize}
\end{proof} 
\section{Unmixing Gaussians in Fourier domain: Noisy case} \label{sec:multi_gauss_unmix_analysis_noisy}
We now analyze the noisy setting where we acquire noisy values of the 
of the Fourier transform of $f$ at the sampling location (frequency) $s$.  
In particular, at stage $p$ ($1 \leq p \leq L$) in Algorithm \ref{algo:mmp_gaussian_unmix}, and frequency $s$, 
let $w_p(s)$ denote the additive observation noise on the clean Fourier sample $f(s)$.  
Denoting the noisy measurement by $\widetilde{f}(s)$, this means that at stage $p$,
\begin{align*}
\widetilde{f}(s) 
&= f(s) + w_p(s) \\
&= \sum_{l=1}^{L} \kerfour_l(s) \left(\sum_{j=1}^K u_{l,j} \exp(\iota 2\pi s t_{l,j})\right) + w_p(s).
\end{align*}
In addition to the terms defined at the beginning of Section \ref{sec:gen_noiseless_case}, 
we will need an additional term (defined below) which will be used in the statement of our theorem.
\begin{align}
F'_l(\varepsilon) &:= C_{l,1}' + C_{l,2}' \log^{1/2}\left(\frac{3 D_l}{\varepsilon}\right); \quad 
l = 2,\dots,L-1, \label{eq:gen_proof_temp28bis} 
\end{align}
where $C_{l,1}'$, $C_{l,2}'$ are as defined in \eqref{eq:gen_proof_temp29}.
%
%
\begin{theorem} \label{thm:gen_case_main_noisy}
For a constant $c \in (0,1)$, let $0 < \varepsilon_L < c\sep_L/2$, $\frac{2}{\sep_L - 2\varepsilon_L} \leq m_L < \mup{L}$ 
and $s_L = 0$. Moreover, for $l = L-1, \dots, 1$, say we choose $\varepsilon_l,m_l,s_l$ as follows.
\begin{enumerate} 
\item $0 < \varepsilon_l < c\sep_l/2$ additionally satisfies the following conditions.
\begin{enumerate} 

\item\label{itm:main_thm_epscond_2_noisy} \hemant{$(2\pi \mup{L} + E_{L-1}(\varepsilon_{L-1})) \varepsilon_{L-1} \urel 
\leq \varepsilon_{L} e^{-2\pi^2(\kerfnvar_{L}^2-\kerfnvar_1^2)\mup{L}^2} \frac{\kerfnvar_{L} B(\urel,K)}{2 K (L-1) \kerfnvar_{L-1}}$.} 

\item\label{itm:main_thm_epscond_3_noisy} If $l < L-1$, then $\varepsilon_{l} \leq \varepsilon_{l+1}$, 
\hemant{$E_l(\varepsilon_{l}) \varepsilon_{l} \leq E_{l+1}(\varepsilon_{l+1}) \varepsilon_{l+1}$ and 
$$(F'_{l+1}(\varepsilon_{l})  + E_{l}(\varepsilon_{l})) \varepsilon_{l} \urel 
\leq \varepsilon_{l+1} e^{-(\kerfnvar_{l+1}^2-\kerfnvar_1^2)\frac{F_{l+1}^{'^2}(\varepsilon_{l+1})}{2}} \frac{\kerfnvar_{l+1} B(\urel,K)}{3 K l \kerfnvar_{l}}.$$}  
\end{enumerate}

\item $\frac{2}{\sep_l - 2\varepsilon_l} \leq m_l < \mup{l}$, and 
$S_l \leq s_l \leq \widetilde{c} S_l$ (for constant $\widetilde{c} > 1$) where
\begin{align*}
S_l = m_l + \frac{1}{(2\pi^2(\kerfnvar_{l+1}^2 - \kerfnvar_l^2))^{1/2}} 
\log^{1/2}\left(\hemant{\frac{b_l K (L-l) \kerfnvar_L \urel}{\varepsilon_l \kerfnvar_l B(\urel,K)}} \right),
\end{align*}
with \hemant{$b_l = 2$ if $l = 1$, and $b_l = 3$} otherwise.
\end{enumerate}
Assume that the noise satisfies the conditions
\begin{align}
    \left\Vert \frac{w_1(\cdot)}{\kerfour_1(s_1+\cdot)}\right\Vert_{\infty} 
		& \le \varepsilon_1 \hemant{\frac{\umin B(\urel,K)}{2}}, \label{noise_bnd} \\
\left\Vert \frac{w_l(\cdot)}{\kerfour_l(s_l+\cdot)}\right\Vert_{\infty} 
& \le \varepsilon_l \hemant{\frac{\umin B(\urel,K)}{3}}, 
\quad l=2,\ldots,L-1, \quad \text{and} \label{noise_bnd1_tmp} \\
\left\Vert \frac{w_L(\cdot)}{\kerfour_L(s_L+\cdot)}\right\Vert_{\infty} 
& \le \varepsilon_L \hemant{\frac{\umin B(\urel,K)}{2}}. \label{noiseL}
\end{align}
Then, for each $l=1,\dots,L$, there exists a permutation $\perm_l: [K] \rightarrow [K]$ such that 
\begin{align*}
d_w(\widehat t_{l,\perm_l(j)}, t_{l,j}) &\leq \varepsilon_l, \quad 
\abs{\widehat u_{l,\perm_l(j)} - u_{l,j}} < E_l(\varepsilon_l) \hemant{\varepsilon_l \umax}; \quad j=1,\dots,K.
\end{align*}
\end{theorem}
Since the organisation and the arguments of the proof of Theorem 
\ref{thm:gen_case_main_noisy} are almost identical to that of Theorem \ref{thm:gen_case_main}, 
we defer this proof to Appendix \ref{app:sec:gen_unmix_noisy_thm} where we pinpoint the main differences between the noiseless 
and the noisy settings. 
%
\paragraph{Interpreting Theorem \ref{thm:gen_case_main_noisy}}
Theorem \ref{thm:gen_case_main_noisy} is almost the same as Theorem \ref{thm:gen_case_main} 
barring the conditions on the magnitude of external noise (and minor differences in some constants). 
Specifically, conditions \eqref{noise_bnd} - \eqref{noiseL} state that at stage $l$, 
the magnitude of the noise should be small relative to the desired accuracy parameter $\varepsilon_l$. 
This is examined in more detail below. For convenience, we will now assume $\umin,\umax \asymp 1$. 
\begin{enumerate}
\item \textbf{(Condition on $\norm{w_l(\cdot)}_{\infty}$ for $1 \leq l \leq L-1$.)}
Let us start with the case $2 \leq l \leq L-1$. Condition \eqref{noise_bnd1_tmp} states that 
\begin{equation} \label{eq:noisy_interp_1}
\left| \frac{w_l(i)}{\kerfour_l(s_l+i)} \right| \lesssim \frac{\varepsilon_l}{\sqrt{K}}; \quad i=-m_l,\dots,m_l-1.
\end{equation}
Now, $\abs{\frac{w_l(i)}{\kerfour_l(s_l+i)}} = \frac{\abs{w_l(i)}}{\sqrt{2\pi}\kerfnvar_l} e^{2\pi^2(s_l + i)^2 \kerfnvar_l^2}$. 
Since $s_l > m_l$, therefore $s_l + i > 0$ for the given range of $i$. 
Hence, $(s_l+i)^2 < (s_l+m_l)^2 < F_l^{'^2}(\varepsilon_l)$ for each $i$. 
\hemant{Since $F_l^{'}(\varepsilon) \asymp F_l(\varepsilon)$, therefore} using the order wise  
dependency from \eqref{eq:orddep_F_gen}, we obtain for each $i$ that 
\begin{equation} \label{eq:noisy_interp_2}
\frac{\abs{w_l(i)}}{\sqrt{2\pi}\kerfnvar_l} e^{2\pi^2(s_l + i)^2 \kerfnvar_l^2} 
\lesssim \frac{\abs{w_l(i)}}{\kerfnvar_l}\left(\frac{K^{3/2}L \kerfnvar_L}{\varepsilon_l\kerfnvar_l}\right)^{C(\kerfnvar_l,\kerfnvar_{l+1},\sep_l)},
\end{equation}
where $C(\kerfnvar_l,\kerfnvar_{l+1},\sep_l) > 0$ depends only on $\kerfnvar_l,\kerfnvar_{l+1},\sep_l$. 
Hence from \eqref{eq:noisy_interp_1}, 
\eqref{eq:noisy_interp_2}, we see that \eqref{noise_bnd1_tmp} is satisfied if
\begin{equation} \label{eq:main_noisy_interp_inter}
\norm{w_l}_{\infty} 
\lesssim \frac{(\varepsilon_l\kerfnvar_l)^{1+C(\kerfnvar_l,\kerfnvar_{l+1},\sep_l)}}{\sqrt{K} (K^{3/2} L \kerfnvar_L)^{C(\kerfnvar_l,\kerfnvar_{l+1},\sep_l)}}; \quad 2 \leq l \leq L-1.
\end{equation}
In a similar manner, one can easily show that \eqref{noise_bnd} is satisfied if
\begin{equation} \label{eq:main_noisy_interp_1}
\norm{w_1}_{\infty} 
\lesssim \frac{(\varepsilon_1\kerfnvar_1)^{1+C(\kerfnvar_1,\kerfnvar_{2},\sep_1)}}{\sqrt{K} (K^{3/2} L \kerfnvar_L)^{C(\kerfnvar_1,\kerfnvar_{2},\sep_1)}}.
\end{equation}

\item \textbf{(Condition on $\norm{w_L}_{\infty}$.)} 
In this case, $s_L = 0$ and so $(s_L + i)^2 \leq m_L^2$. Therefore for each $i=-m_L,\dots,m_L-1$, we obtain
\begin{equation} \label{eq:noisy_interp_3}
\left| \frac{w_L(i)}{\kerfour_L(s_L+i)} \right|
= \frac{\abs{w_L(i)}}{\sqrt{2\pi}\kerfnvar_L} e^{2\pi^2(i)^2 \kerfnvar_L^2} 
\lesssim \frac{\abs{w_L(i)}}{\kerfnvar_L} e^{2\pi^2 \kerfnvar_L^2/\sep_L^2}.
\end{equation}
Hence from \eqref{eq:noisy_interp_3}, we see that \eqref{noiseL} is satisfied if
\begin{equation} \label{eq:main_noisy_interp_L}
\norm{w_L}_{\infty} \lesssim \kerfnvar_L e^{-\kerfnvar_L^2/\sep_L^2} \frac{\varepsilon_L}{\sqrt{K}}.
\end{equation}
\end{enumerate}
\eqref{eq:main_noisy_interp_1}, \eqref{eq:main_noisy_interp_inter}, \eqref{eq:main_noisy_interp_L} show the 
conditions that the noise level is required to satisfy at the different levels. From the discussion following 
Theorem \ref{thm:gen_case_main}, we know that the $\varepsilon_i$'s gradually become smaller and smaller as we move 
from $i=L$ to $i=1$ (with $\varepsilon_1$ being the smallest). Therefore the condition on $\norm{w_1}_{\infty}$ 
is the strictest, while the condition on $\norm{w_L}_{\infty}$ is the mildest.

\paragraph{Corollary for the case $L = 1$. }
As noted earlier, the case $L = 1$ is not interesting in the absence of external noise as we 
can exactly recover the source parameters. The situation is more interesting in the presence of noise 
as shown in the following Corollary of Theorem \ref{thm:gen_case_main_noisy} for the case $L=1$.
%
\begin{corollary}
For a constant $c \in (0,1)$, let $0 < \varepsilon_1 < c\sep_1/2$, 
$\frac{2}{\sep_1 - 2\varepsilon_1} \leq m_1 < \frac{2}{\sep_1(1-c)} + 1$ 
and $s_1 = 0$. Moreover, denoting $C = 10 + \frac{1}{\sqrt{2}}$, assume that the noise satisfies
\begin{align} \label{eq:ext_noisecond_l_1}
    \left\Vert \frac{w_1(\cdot)}{\kerfour_1(\cdot)}\right\Vert_{\infty} 
		& \le \varepsilon_1 \frac{\umin}{10 C\sqrt{K}} \left( 1 + 48 \hemant{\urel} \right)^{-1}.
\end{align}
Then, there exists a permutation $\perm_1: [K] \rightarrow [K]$ such that for $j=1,\dots,K$, we have
\begin{align*}
d_w(\widehat t_{1,\perm_1(j)}, t_{1,j}) 
&\leq \varepsilon_1, \\
\abs{\widehat u_{1,\perm_1(j)} - u_{1,j}} 
&< \hemant{\left(4\pi K \left(\frac{2}{\sep_1(1-c)}+1 \right) + \frac{2}{C\sqrt{K}}\left(\urel + 16\urel^2\right)^{-1} \right) \varepsilon_1 \umax.}
\end{align*}
\end{corollary}
Assuming $\umin,\umax \asymp 1$, we see from \eqref{eq:main_noisy_interp_L} 
that \eqref{eq:ext_noisecond_l_1} is satisfied if  
$\norm{w_1}_{\infty} \lesssim \kerfnvar_1 e^{-\kerfnvar_1^2/\sep_1^2} \frac{\varepsilon_1}{\sqrt{K}}$. 
%
%
\section{Experiments} \label{sec:experiments}
In this section, we present some numerical experiments for our method\footnote{\hemant{Code available here: \url{https://hemant-tyagi.github.io}}} pertaining to the error in the 
recovery of the locations of the spikes. Our setup is as follows. We fix $L = 4$ groups, and consider 
$K \in \set{2,3,4,5}$. We fix the variance parameters of the Gaussian kernel to be $\kerfnvar_L = 0.01$
and $\kerfnvar_l = \kerfnvar_{l+1}/2$ for $l=1,\dots,L-1$. The minimum separation parameter $\sep$ is set 
to $0.05$ and $m_l$ is fixed to $1/\sep + 5$ for each $l = 1,\dots,L$. Furthermore, we fix the sampling 
parameters as $s_L = 0$, $\varepsilon_L = 0.01$, and for each $l=1,\dots,L-1$ choose
\begin{equation} \label{eq:exp_sl_const}
 \varepsilon_l = \varepsilon_{l+1}^2, \quad 
s_l = m_l + \frac{C}{\sqrt{2\pi^2 (\kerfnvar_{l+1}^2 - \kerfnvar_{l}^2)}}\log^{1/2}\left(\frac{\kerfnvar_L}{\kerfnvar_l \varepsilon_l}\right)
\end{equation}
for a suitably chosen constant $C > 0$. This is in line with our theory since $\varepsilon_l$ is smaller than 
$\varepsilon_{l+1}$ and $s_l$ too is of the form specified in Theorem \ref{thm:gen_case_main}.

In each Monte Carlo run, we choose $K \in \set{2,3,4,5}$, and then randomly generate $K$ spikes for each $l = 1,\dots,L$. 
In particular, the amplitudes of the spikes are generated by first uniformly sampling values in $[\umin,\umax]$ with $\umin = 3$ 
and $\umax = 10$, and then randomly assigning each one of them a negative sign with probability $1/2$. 
Moreover, the spike locations are sampled uniformly at random in $(0,1)$ with the minimum separation $\sep_l$ 
for each group ensured to be greater than or equal to $\sep$. 
Thereafter, each Fourier sample (for every group) is corrupted with zero-mean i.i.d Gaussian noise. 
For each group $l$, upon obtaining the estimated spike locations $\est{t}_{l,j}$ and 
amplitudes $\est{u}_{l,j}$ for $j=1,\dots,K$ (via Algorithm \ref{algo:mmp_gaussian_unmix}), 
we ``match'' (for each group) the estimated set of spikes with the input spikes based on a 
simple heuristic. We first find the 
estimated spike location that has the smallest wrap around distance from an input spike 
location -- this gives us a match. This pair is then removed, and we repeat the process on the remaining sets of spikes. 
This finally gives us a permutation $\perm:[K] \rightarrow [K]$ where $t_{l,j}$ would 
ideally be close to $\est{t}_{l,\perm(j)}$. Finally, we evaluate the performance of our algorithm by examining 
(a) the maximum wrap around distance $d_{w,l,\max} := \max_{j} d_w(t_{l,j}, \est{t}_{l,\phi(j)})$, and 
(b) the average wrap around distance $d_{w,l,\text{avg}}:= (1/K) \sum_{j=1}^K d_w(t_{l,j}, \est{t}_{l,\phi(j)})$. 
This is repeated over $400$ Monte Carlo trials.

Figures \ref{fig:max_loc_err_noiseless}, \ref{fig:mean_loc_err_noiseless} show scatter plots for the 
above notions of error and the minimum separation $\sep_l$ for each group, in the absence of external noise.
Figures \ref{fig:max_loc_err_noise}, \ref{fig:mean_loc_err_noise} show the same, but with external 
Gaussian noise (standard deviation $5 \times 10^{-5}$). Both these plots are for $C = 0.6$ 
in \eqref{eq:exp_sl_const}, we found this choice to give the best result. We observe that for $l=1,2$, 
the spike locations are recovered near exactly in all the trials and for all choices of $K$ -- both 
in the noiseless and noisy settings. For $l=3,4$, the performance is reasonably good for $K=2,3$. 
In particular, for both the noiseless and noisy settings, 
$d_{w,l,\max}$ is less than $0.05$ in at least $82\%$ of the trials while $d_{w,l,\text{avg}}$ is 
less than $0.05$ in at least $93\%$ of trials. While the performance is seen to drop as $K$ increases -- especially 
in terms of $d_{w,l,\max}$ -- the performance in terms of $d_{w,l,\text{avg}}$ is still significantly 
better than $d_{w,l,\max}$. In particular, for $K = 4$, $d_{w,l,\text{avg}}$ is 
less than $0.05$ in at least $86\%$ of the trials (in both noiseless and noisy settings), while 
for $K = 5$, the same is true in at least $73\%$ of the trials.

As mentioned earlier, we found the choice $C = 0.6$ to give the best performance, in general. In the noiseless 
scenario however, the choice $C = 1$ results in near exact recovery for all groups as shown in 
Appendix \ref{sec:app_exps} (see Figures \ref{fig:max_loc_err_noiseless_C1},\ref{fig:mean_loc_err_noiseless_C1}). 
However in the presence of external Gaussian noise of standard deviation $5 \times 10^{-5}$, the recovery performance 
breaks down for groups $l=2,3,4$ (see Figures \ref{fig:max_loc_err_noise_C1},\ref{fig:mean_loc_err_noise_C1}). 
Since we are sampling relatively deeper in the Fourier tail (compared to that when $C = 0.6$), the deconvolution 
step blows up the noise significantly, leading to the worse recovery performance. Finally, in the noiseless case, 
we observed that the recovery performance breaks down for larger values of $C$ (i.e., for $C > 4$) due to numerical 
errors creeping in the deconvolution step.

%
\begin{figure}[!ht]
\centering
\subcaptionbox[]{$l = 1, K = 2$}[ 0.24\textwidth ]
{\includegraphics[width=0.24\textwidth]{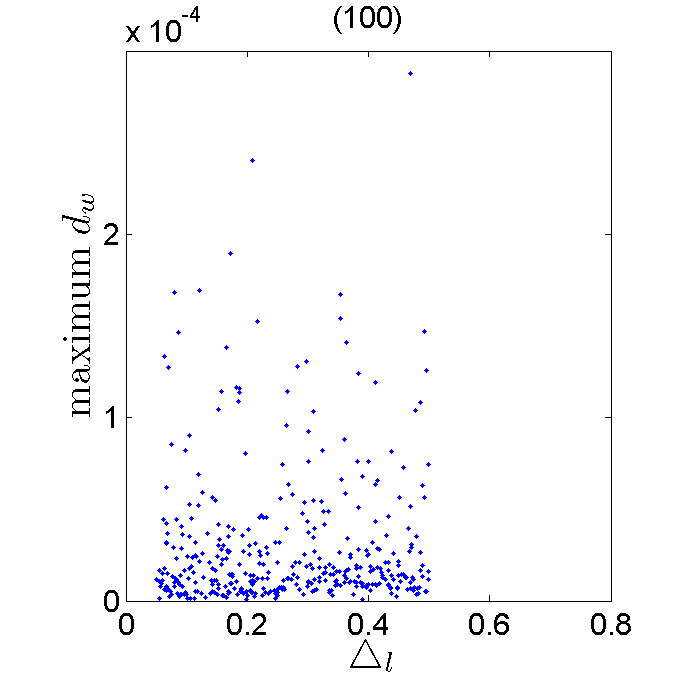} }
\subcaptionbox[]{$l = 2, K = 2$}[ 0.24\textwidth ]
{\includegraphics[width=0.24\textwidth]{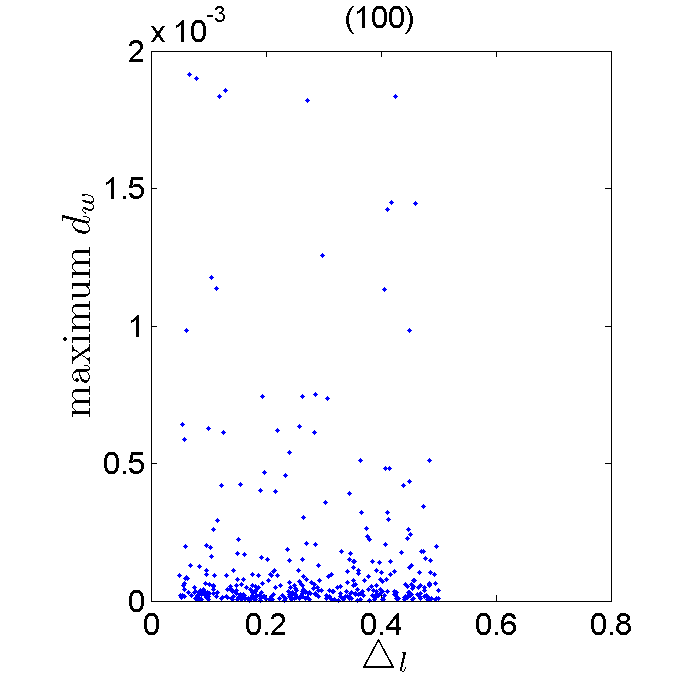} }
%
\subcaptionbox[]{$l = 3, K = 2$}[ 0.24\textwidth ]
{\includegraphics[width=0.24\textwidth]{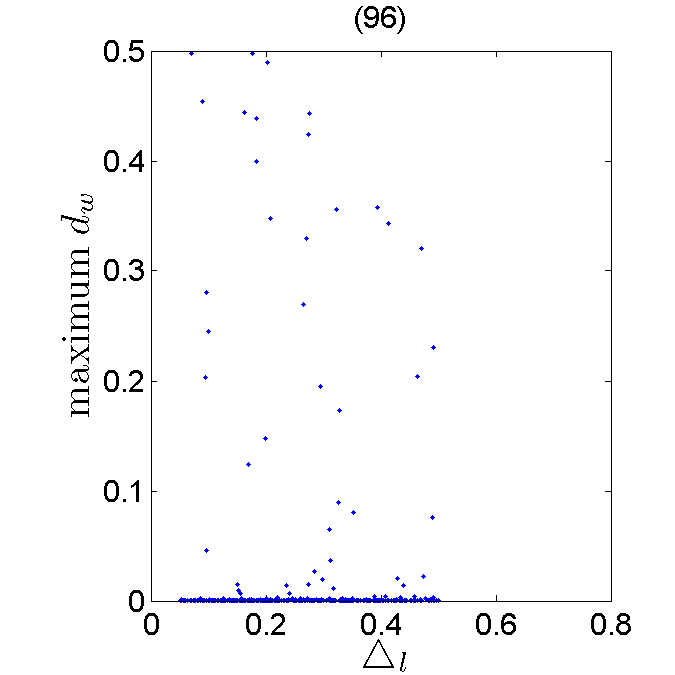} }
%
\subcaptionbox[]{$l = 4, K = 2$}[ 0.24\textwidth ]
{\includegraphics[width=0.24\textwidth]{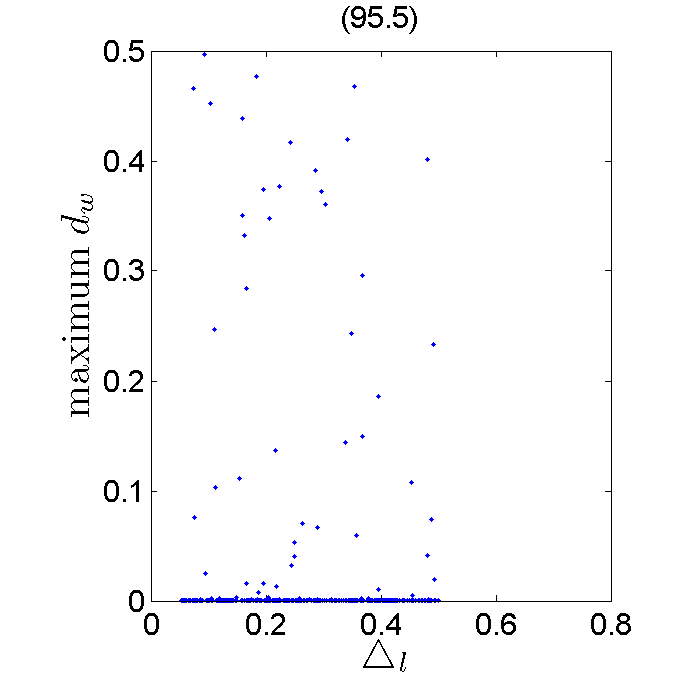} }
%

\subcaptionbox[]{$l = 1, K = 3$}[ 0.24\textwidth ]
{\includegraphics[width=0.24\textwidth]{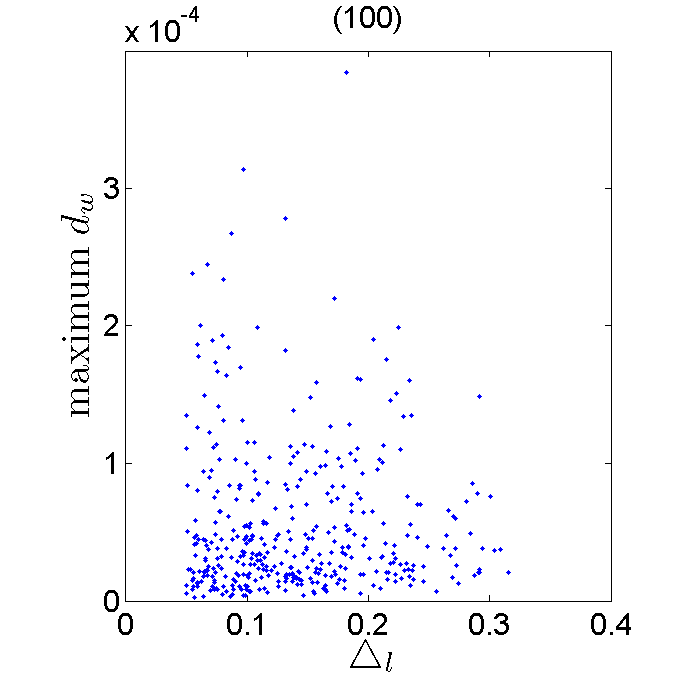} }
\subcaptionbox[]{$l = 2, K = 3$}[ 0.24\textwidth ]
{\includegraphics[width=0.24\textwidth]{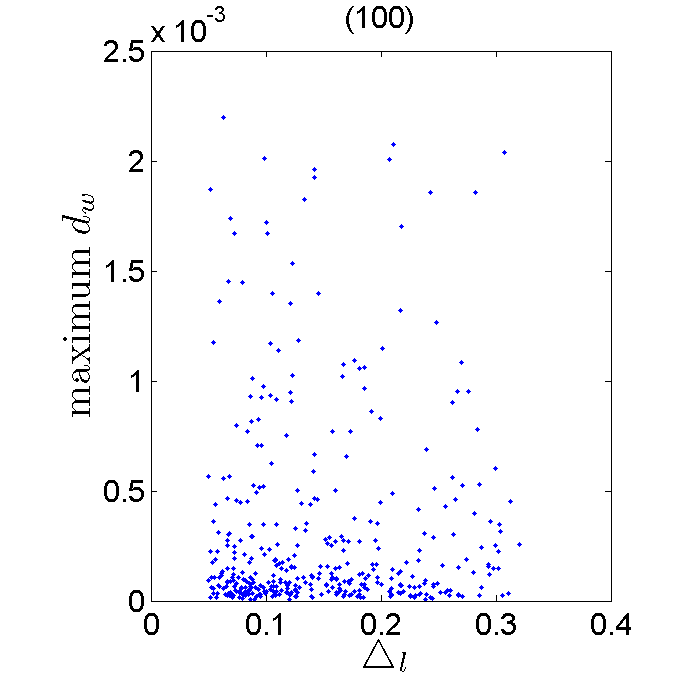} }
%
\subcaptionbox[]{$l = 3, K = 3$}[ 0.24\textwidth ]
{\includegraphics[width=0.24\textwidth]{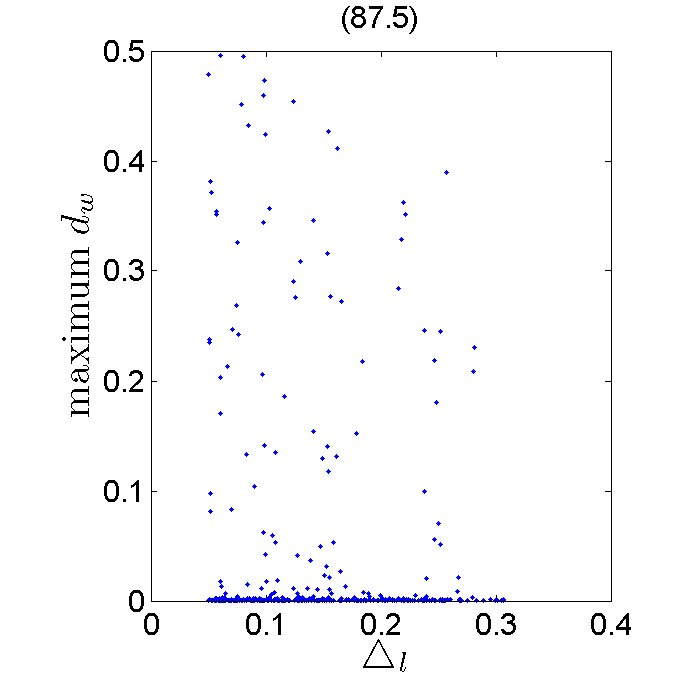} }
%
\subcaptionbox[]{$l = 4, K = 3$}[ 0.24\textwidth ]
{\includegraphics[width=0.24\textwidth]{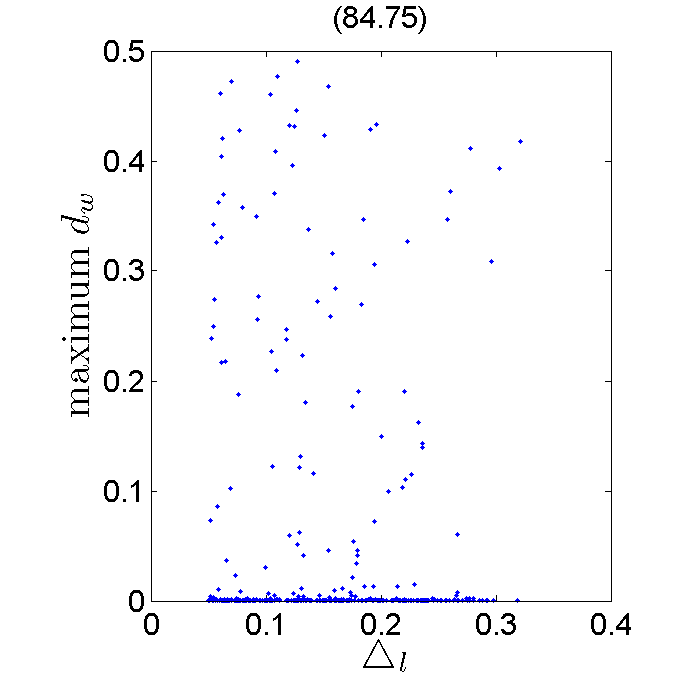} }
%

\subcaptionbox[]{$l = 1, K = 4$}[ 0.24\textwidth ]
{\includegraphics[width=0.24\textwidth]{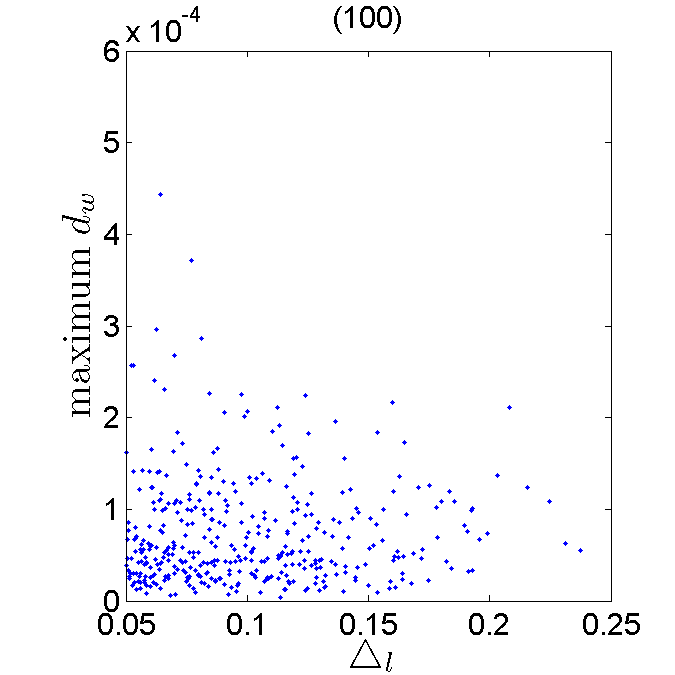} }
\subcaptionbox[]{$l = 2, K = 4$}[ 0.24\textwidth ]
{\includegraphics[width=0.24\textwidth]{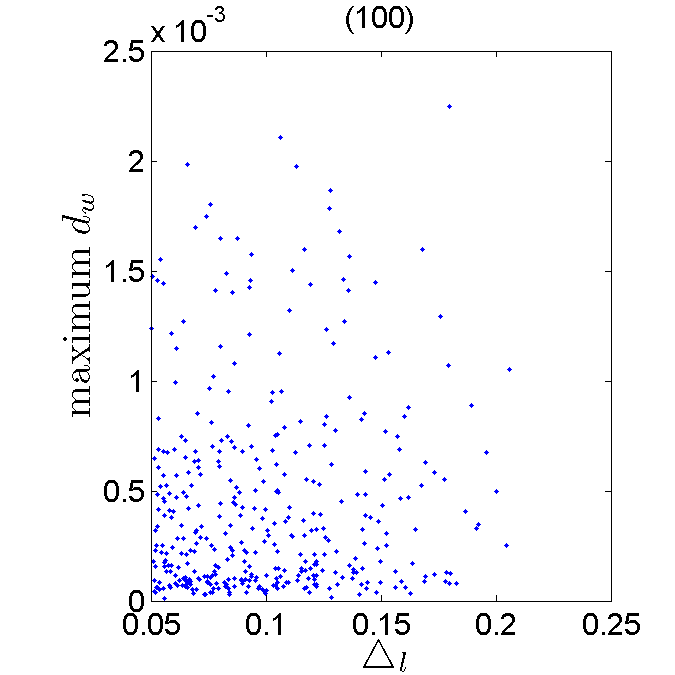} }
%
\subcaptionbox[]{$l = 3, K = 4$}[ 0.24\textwidth ]
{\includegraphics[width=0.24\textwidth]{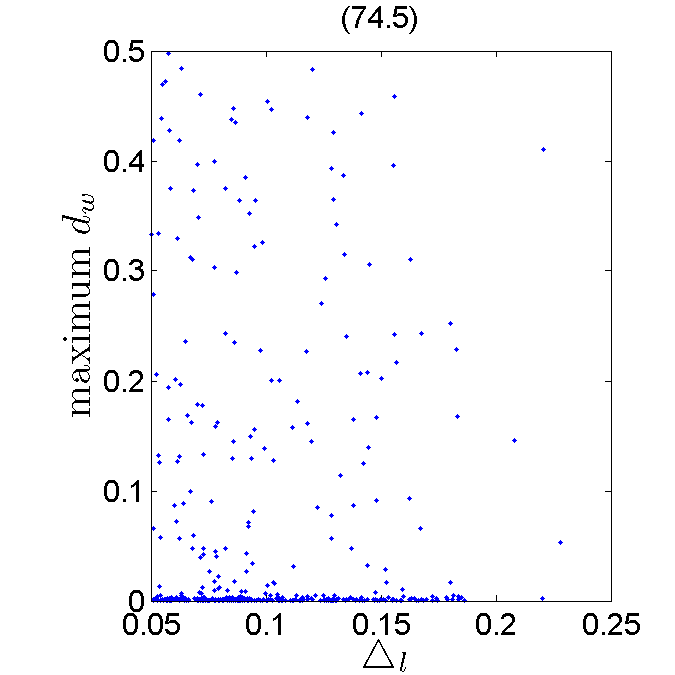} }
%
\subcaptionbox[]{$l = 4, K = 4$}[ 0.24\textwidth ]
{\includegraphics[width=0.24\textwidth]{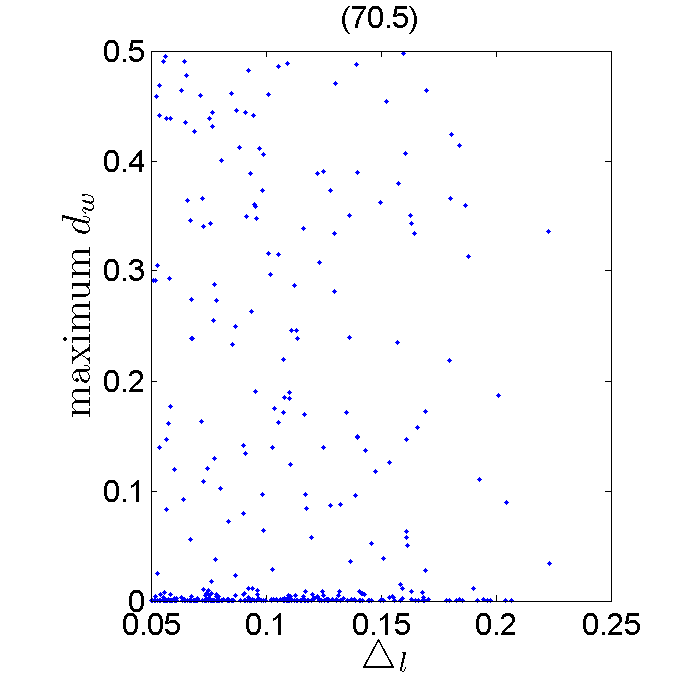} }
%

\subcaptionbox[]{$l = 1, K = 5$}[ 0.24\textwidth ]
{\includegraphics[width=0.24\textwidth]{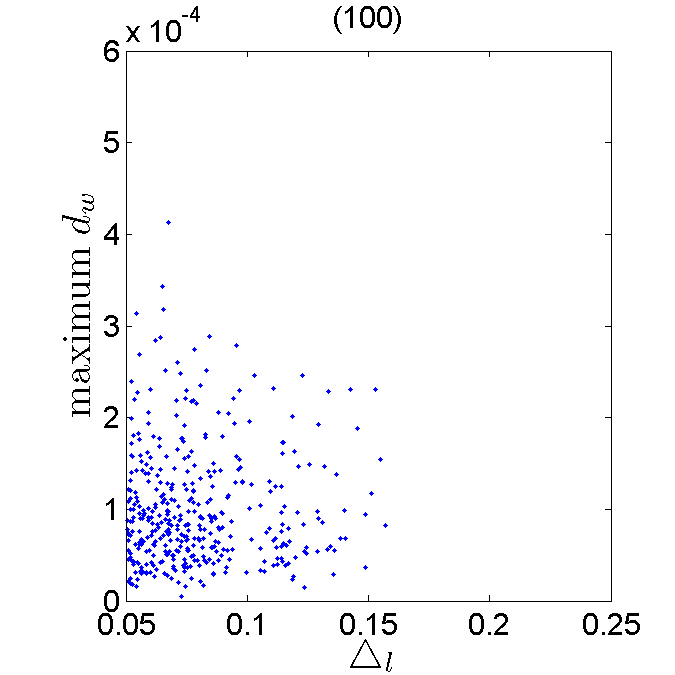} }
\subcaptionbox[]{$l = 2, K = 5$}[ 0.24\textwidth ]
{\includegraphics[width=0.24\textwidth]{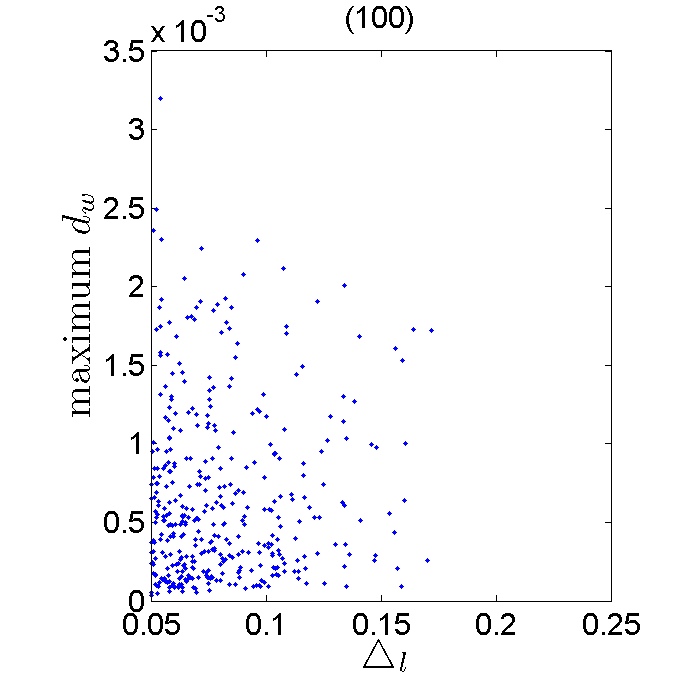} }
%
\subcaptionbox[]{$l = 3, K = 5$}[ 0.24\textwidth ]
{\includegraphics[width=0.24\textwidth]{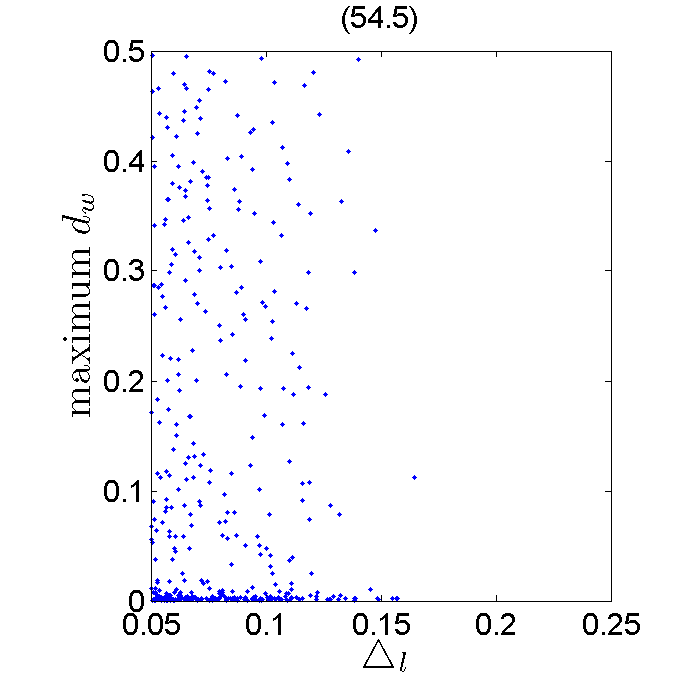} }
%
\subcaptionbox[]{$l = 4, K = 5$}[ 0.24\textwidth ]
{\includegraphics[width=0.24\textwidth]{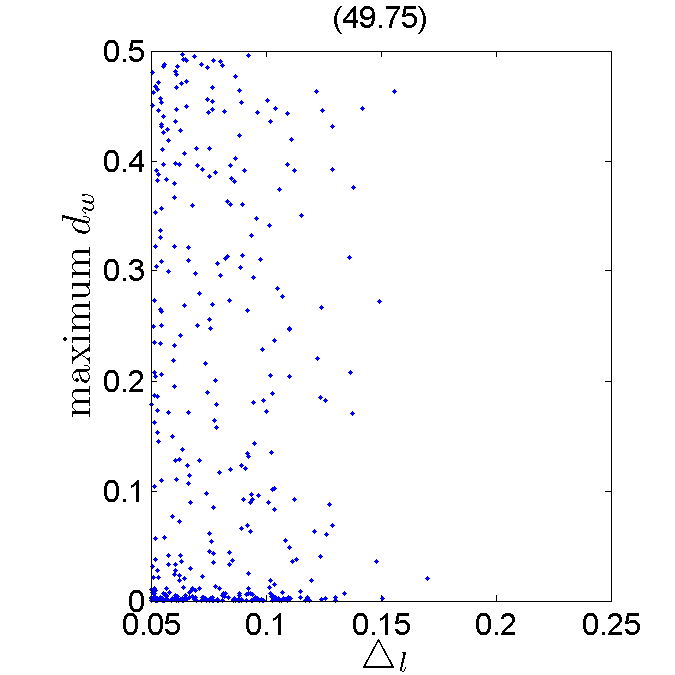} }
\captionsetup{width=0.98\linewidth}
\caption[Short Caption]{Scatter plots for maximum wrap around error ($d_{w,l,\max}$) v/s minimum separation ($\sep_l$) 
for $400$ Monte Carlo trials, 
with no external noise. This is shown for $K \in \set{2,3,4,5}$ with $L = 4$ and $C = 0.6$. For each sub-plot, we mention 
the percentage of trials with $d_{w,l,\max} \leq 0.05$ in parenthesis.}
\label{fig:max_loc_err_noiseless}
\end{figure}

%
\begin{figure}[!ht]
\centering
\subcaptionbox[]{$l = 1, K = 2$}[ 0.24\textwidth ]
{\includegraphics[width=0.24\textwidth]{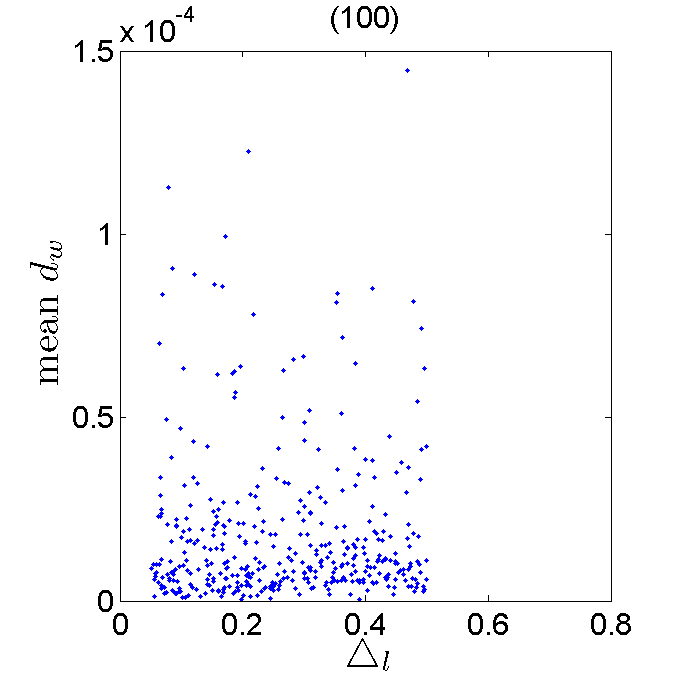} }
\subcaptionbox[]{$l = 2, K = 2$}[ 0.24\textwidth ]
{\includegraphics[width=0.24\textwidth]{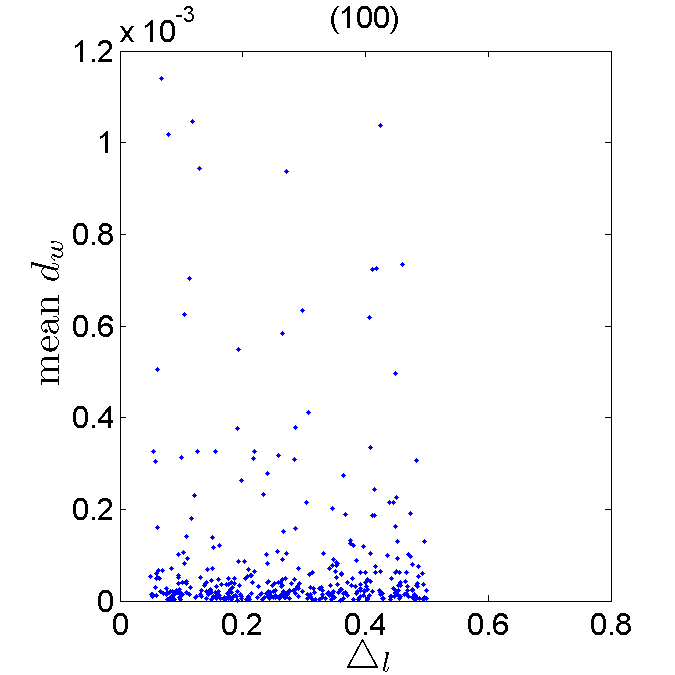} }
%
\subcaptionbox[]{$l = 3, K = 2$}[ 0.24\textwidth ]
{\includegraphics[width=0.24\textwidth]{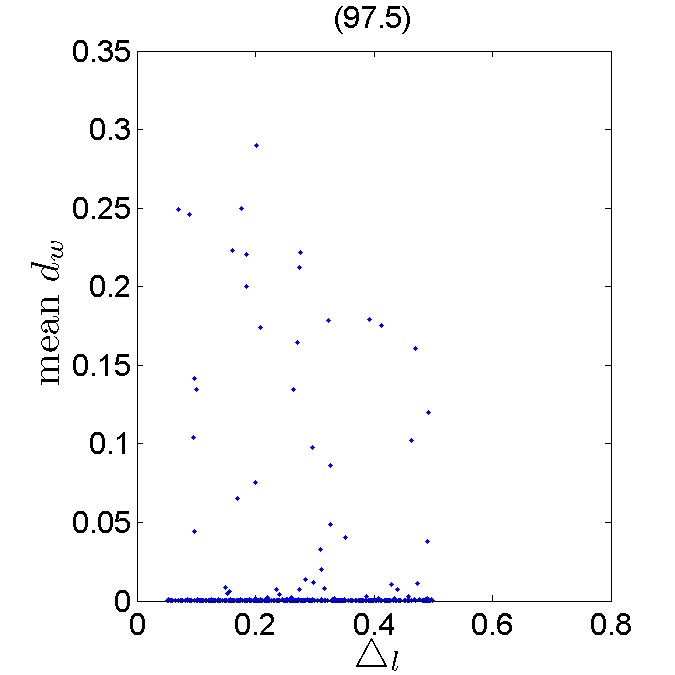} }
%
\subcaptionbox[]{$l = 4, K = 2$}[ 0.24\textwidth ]
{\includegraphics[width=0.24\textwidth]{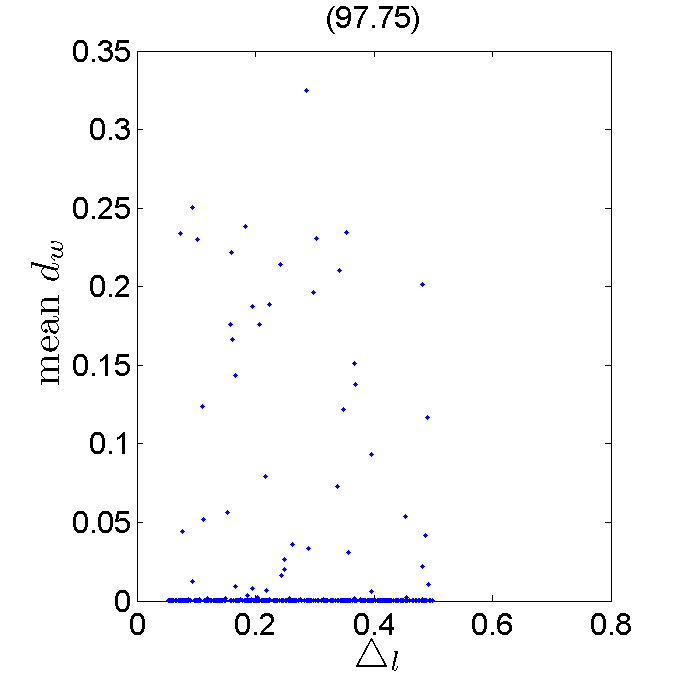} }
%

\subcaptionbox[]{$l = 1, K = 3$}[ 0.24\textwidth ]
{\includegraphics[width=0.24\textwidth]{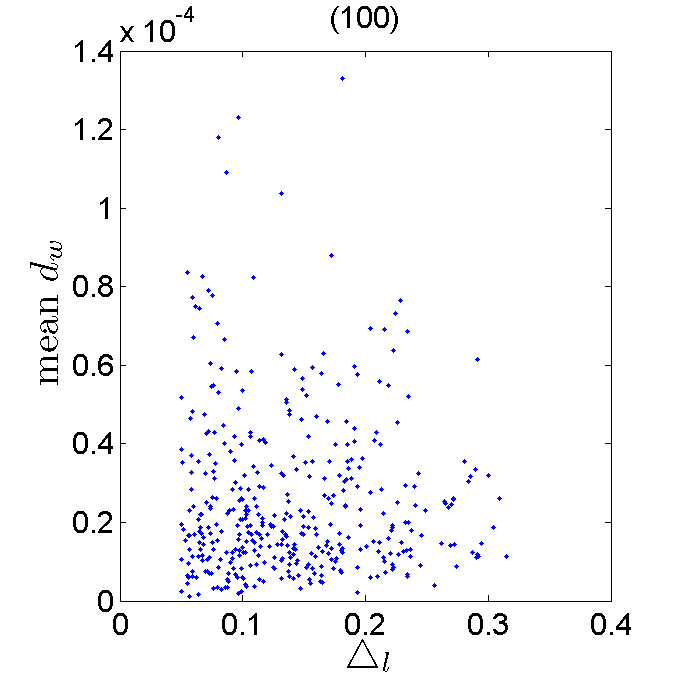} }
\subcaptionbox[]{$l = 2, K = 3$}[ 0.24\textwidth ]
{\includegraphics[width=0.24\textwidth]{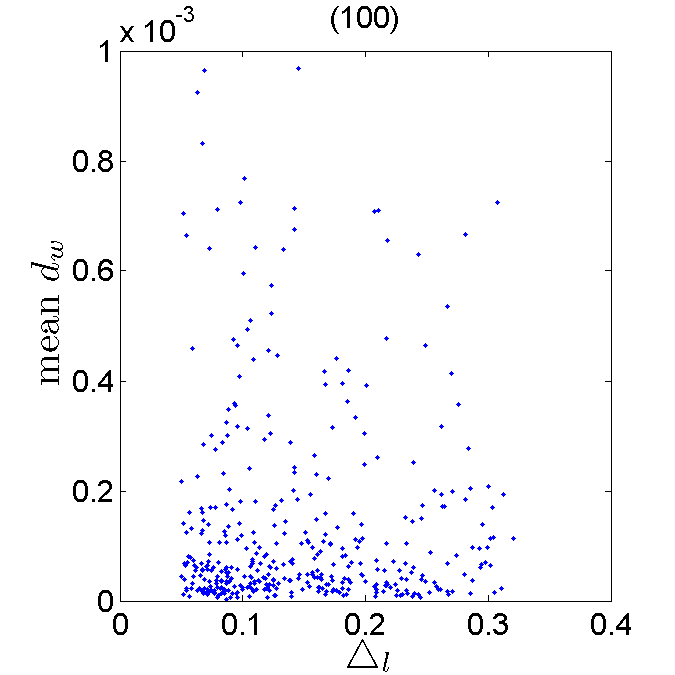} }
%
\subcaptionbox[]{$l = 3, K = 3$}[ 0.24\textwidth ]
{\includegraphics[width=0.24\textwidth]{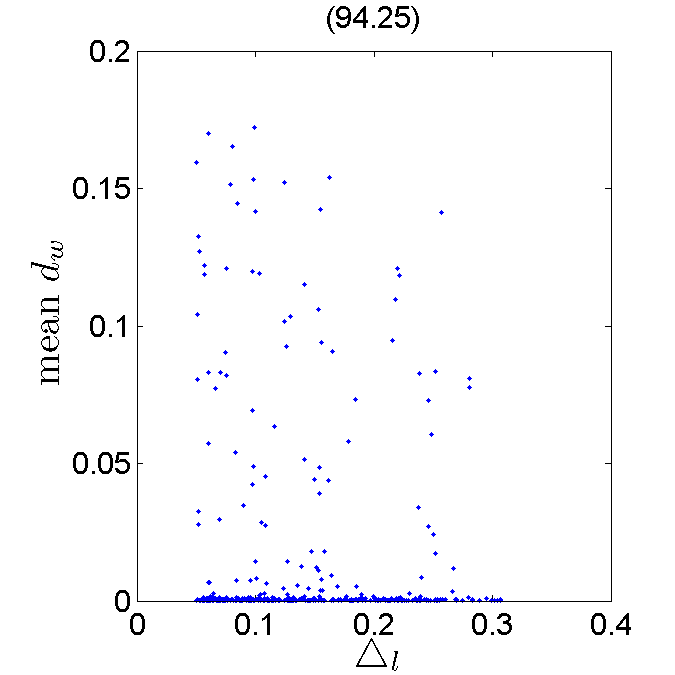} }
%
\subcaptionbox[]{$l = 4, K = 3$}[ 0.24\textwidth ]
{\includegraphics[width=0.24\textwidth]{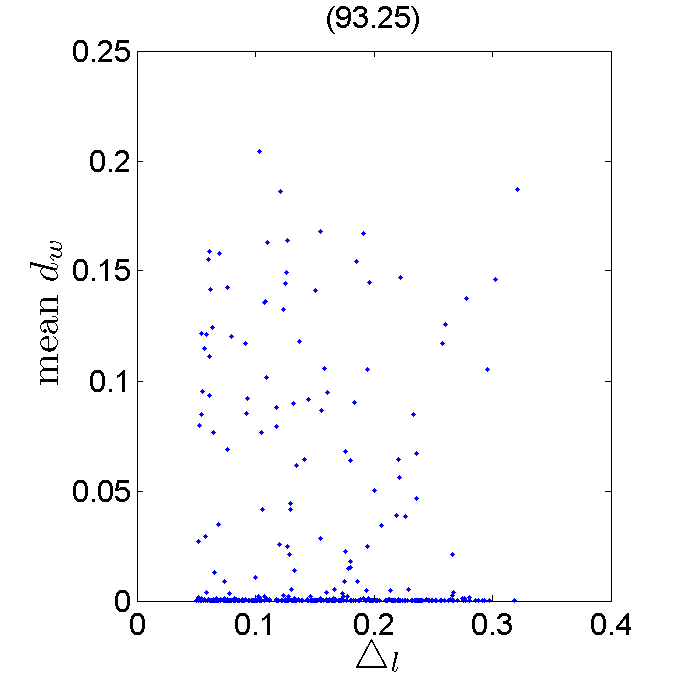} }
%

\subcaptionbox[]{$l = 1, K = 4$}[ 0.24\textwidth ]
{\includegraphics[width=0.24\textwidth]{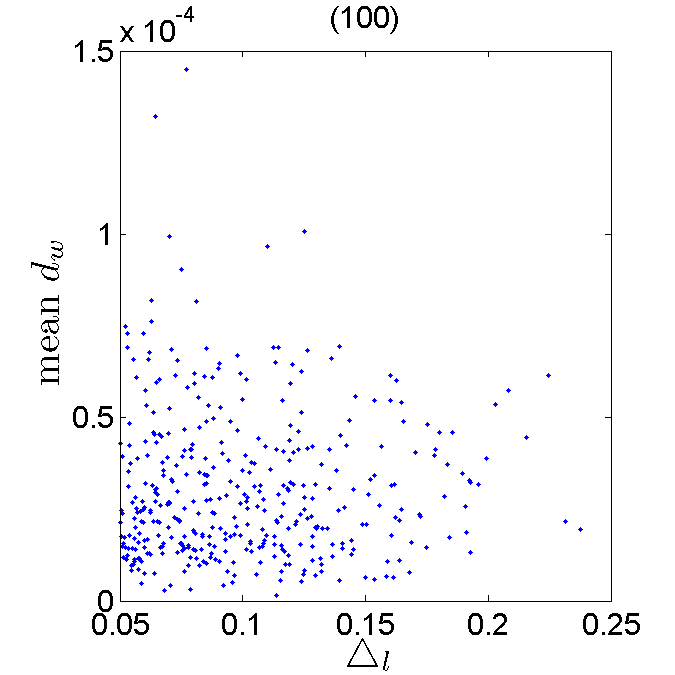} }
\subcaptionbox[]{$l = 2, K = 4$}[ 0.24\textwidth ]
{\includegraphics[width=0.24\textwidth]{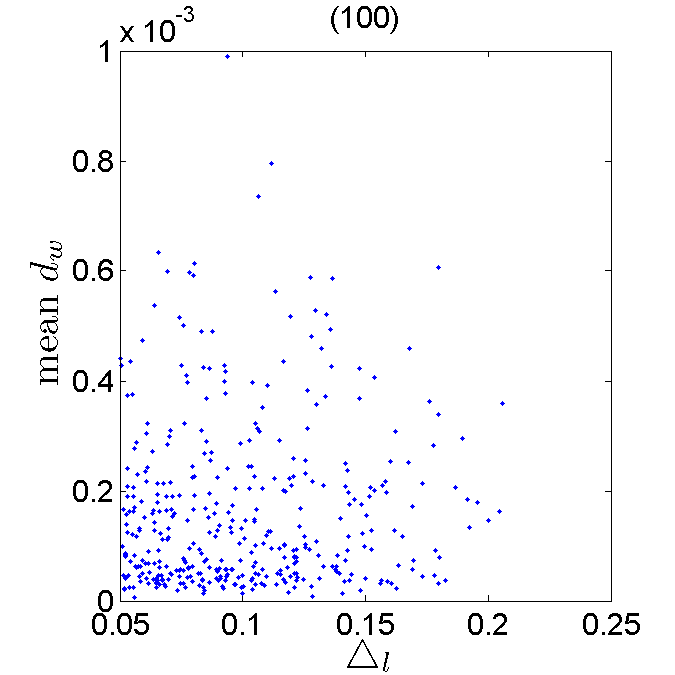} }
%
\subcaptionbox[]{$l = 3, K = 4$}[ 0.24\textwidth ]
{\includegraphics[width=0.24\textwidth]{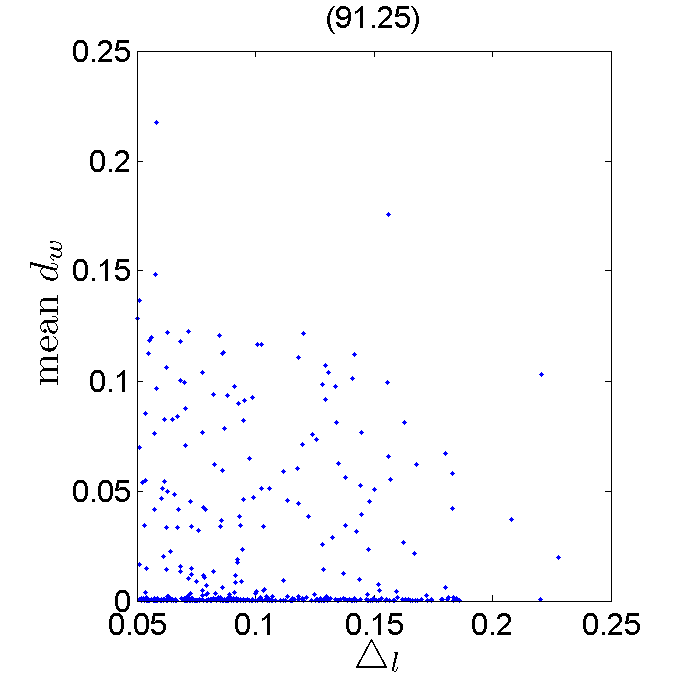} }
%
\subcaptionbox[]{$l = 4, K = 4$}[ 0.24\textwidth ]
{\includegraphics[width=0.24\textwidth]{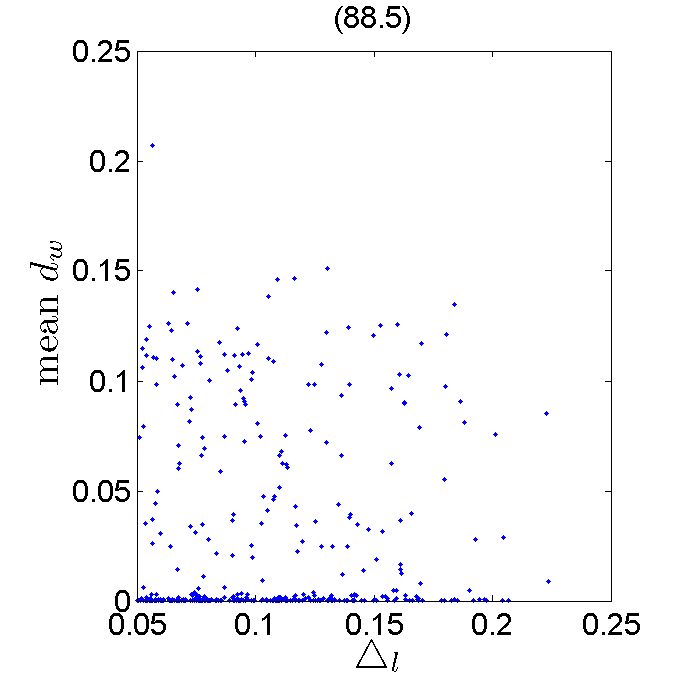} }
%

\subcaptionbox[]{$l = 1, K = 5$}[ 0.24\textwidth ]
{\includegraphics[width=0.24\textwidth]{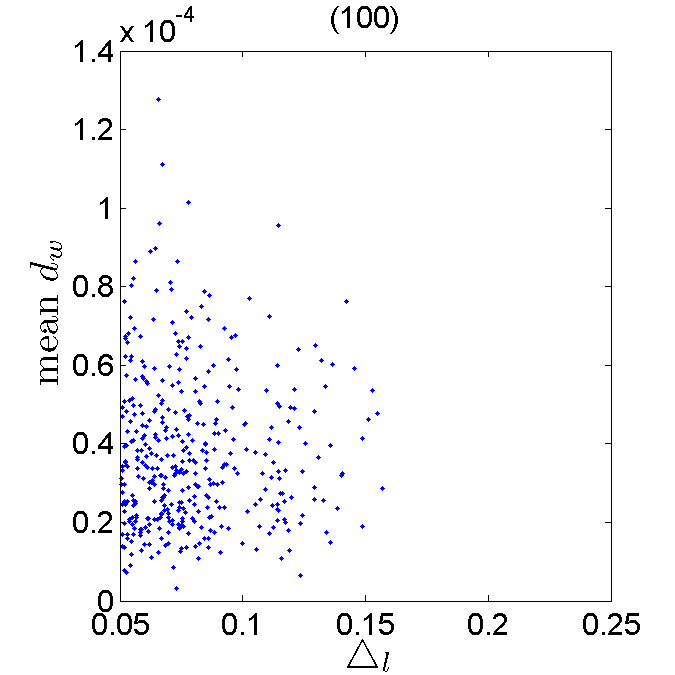} }
\subcaptionbox[]{$l = 2, K = 5$}[ 0.24\textwidth ]
{\includegraphics[width=0.24\textwidth]{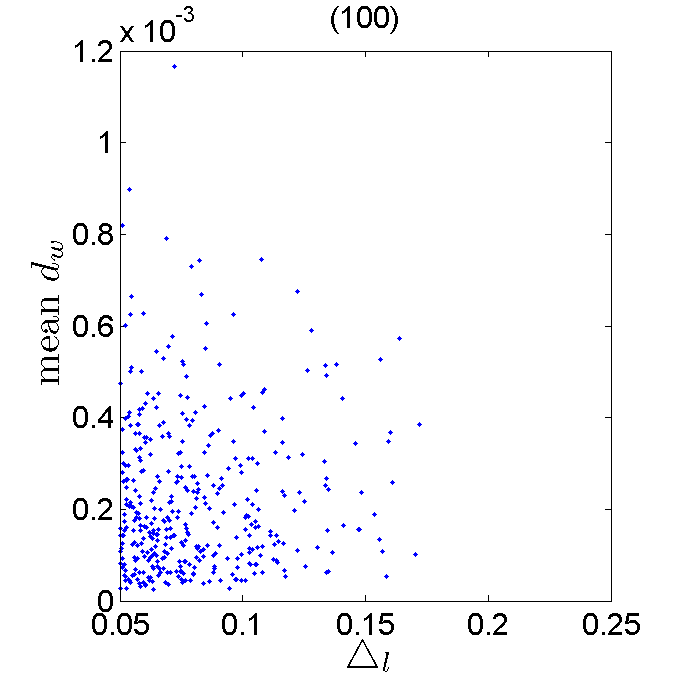} }
%
\subcaptionbox[]{$l = 3, K = 5$}[ 0.24\textwidth ]
{\includegraphics[width=0.24\textwidth]{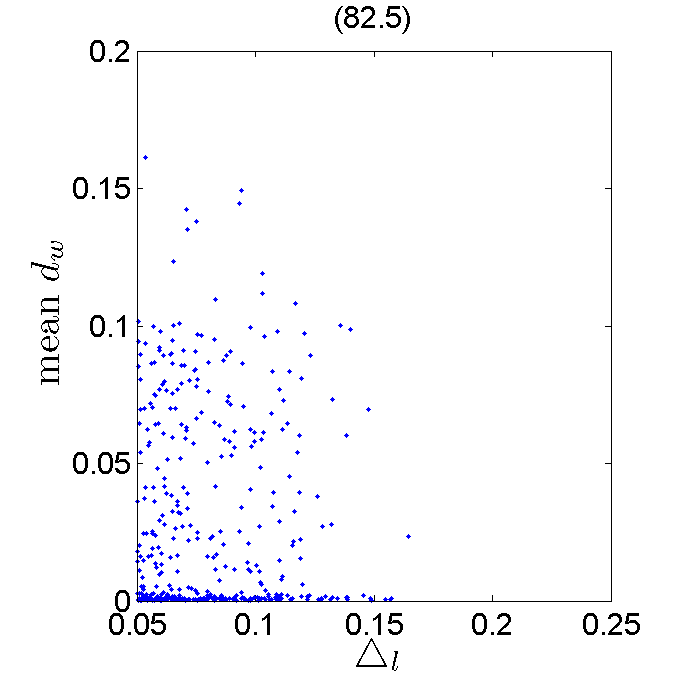} }
%
\subcaptionbox[]{$l = 4, K = 5$}[ 0.24\textwidth ]
{\includegraphics[width=0.24\textwidth]{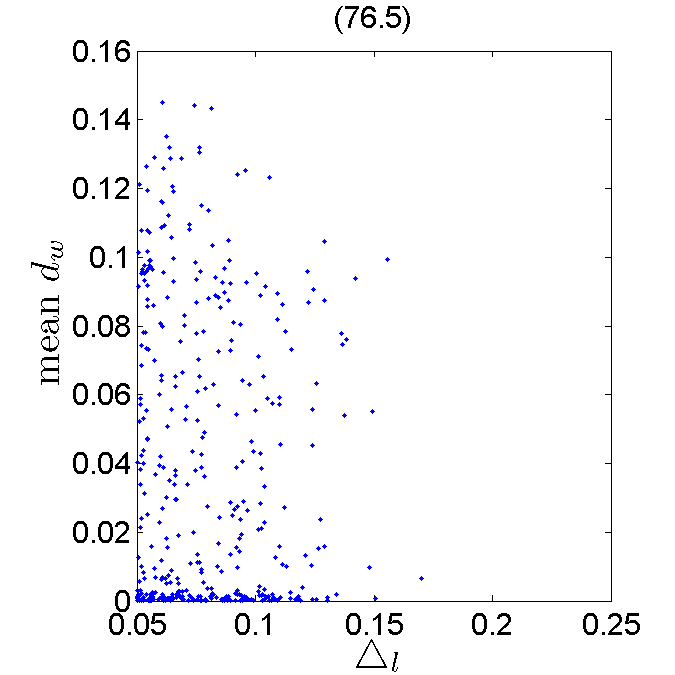} }
\captionsetup{width=0.98\linewidth}
\caption[Short Caption]{Scatter plots for the mean wrap around error ($d_{w,l,\text{avg}}$) v/s minimum separation ($\sep_l$) 
for $400$ Monte Carlo trials, with no external noise. This is shown for $K \in \set{2,3,4,5}$ with $L = 4$ and $C = 0.6$.
For each sub-plot, we mention the percentage of trials with $d_{w,l,\text{avg}} \leq 0.05$ in parenthesis.}
\label{fig:mean_loc_err_noiseless}
\end{figure}

%
\begin{figure}[!ht]
\centering
\subcaptionbox[]{$l = 1, K = 2$}[ 0.24\textwidth ]
{\includegraphics[width=0.24\textwidth]{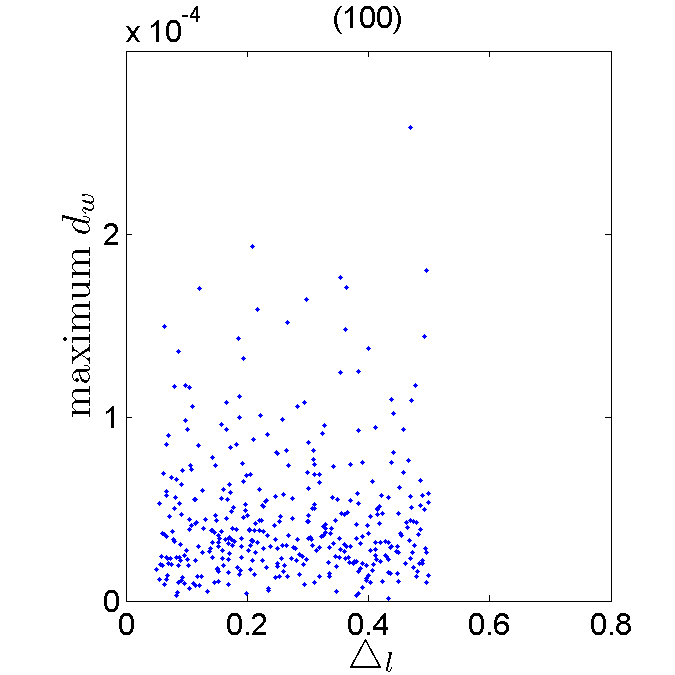} }
\subcaptionbox[]{$l = 2, K = 2$}[ 0.24\textwidth ]
{\includegraphics[width=0.24\textwidth]{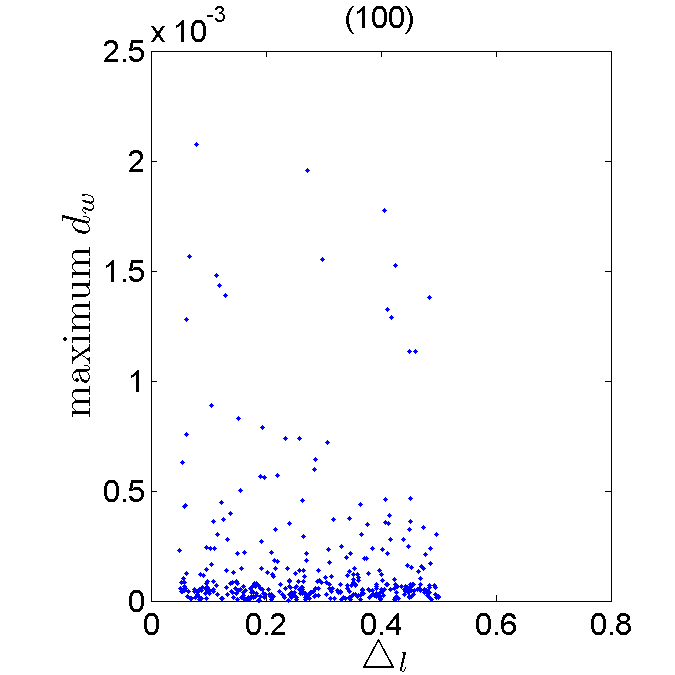} }
%
\subcaptionbox[]{$l = 3, K = 2$}[ 0.24\textwidth ]
{\includegraphics[width=0.24\textwidth]{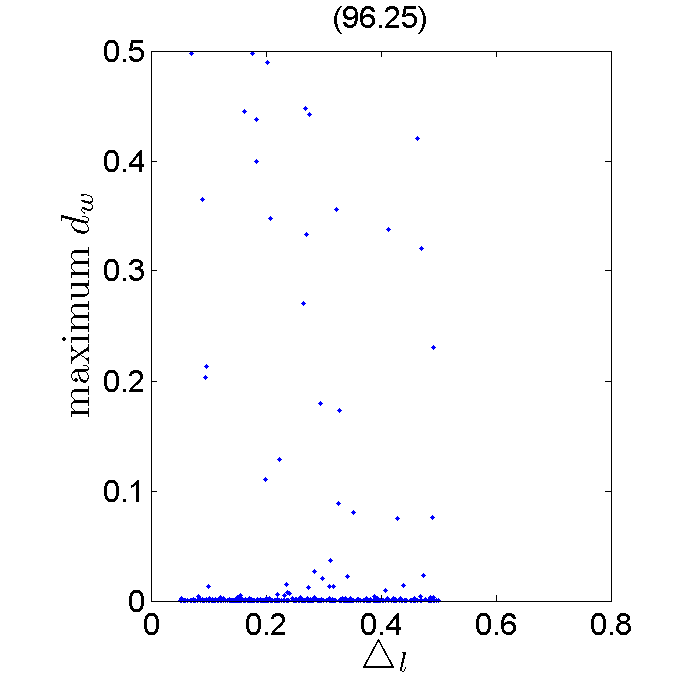} }
%
\subcaptionbox[]{$l = 4, K = 2$}[ 0.24\textwidth ]
{\includegraphics[width=0.24\textwidth]{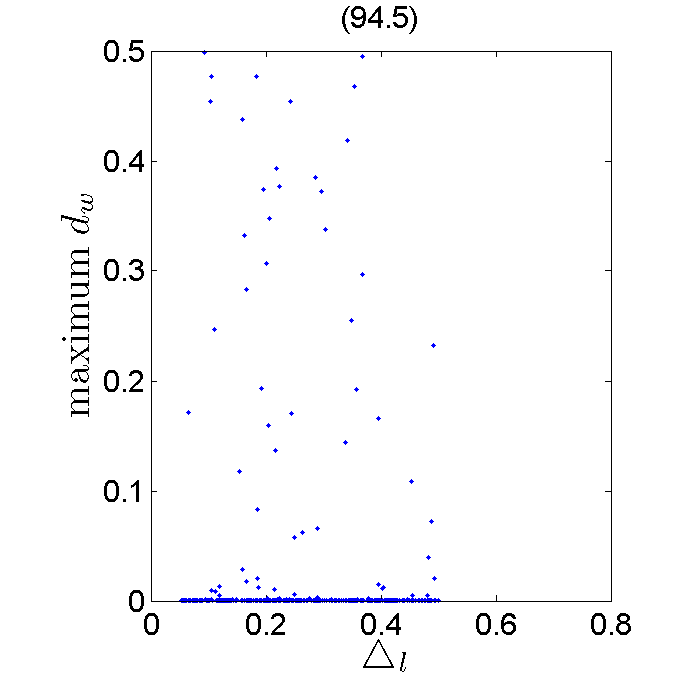} }
%

\subcaptionbox[]{$l = 1, K = 3$}[ 0.24\textwidth ]
{\includegraphics[width=0.24\textwidth]{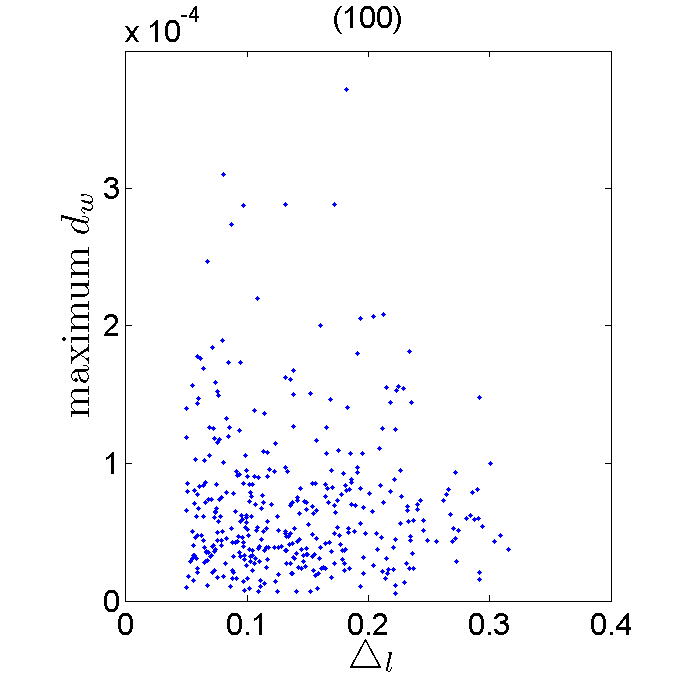} }
\subcaptionbox[]{$l = 2, K = 3$}[ 0.24\textwidth ]
{\includegraphics[width=0.24\textwidth]{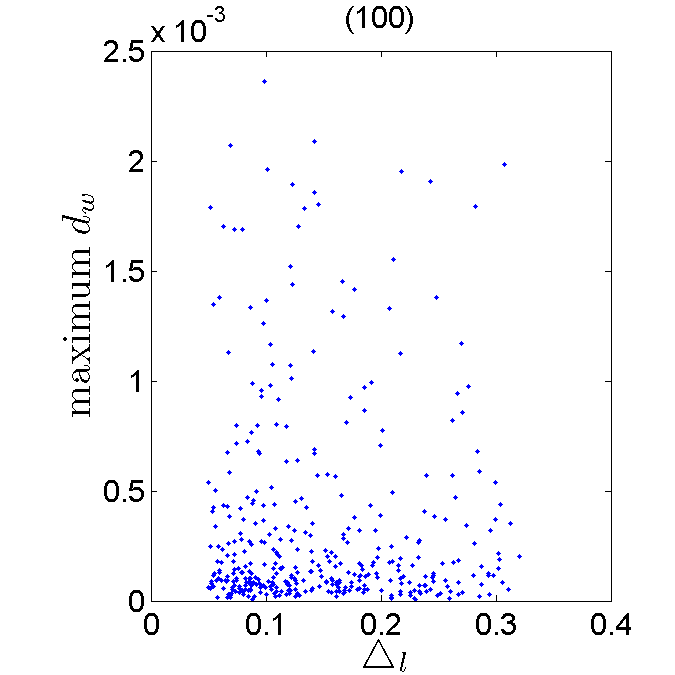} }
%
\subcaptionbox[]{$l = 3, K = 3$}[ 0.24\textwidth ]
{\includegraphics[width=0.24\textwidth]{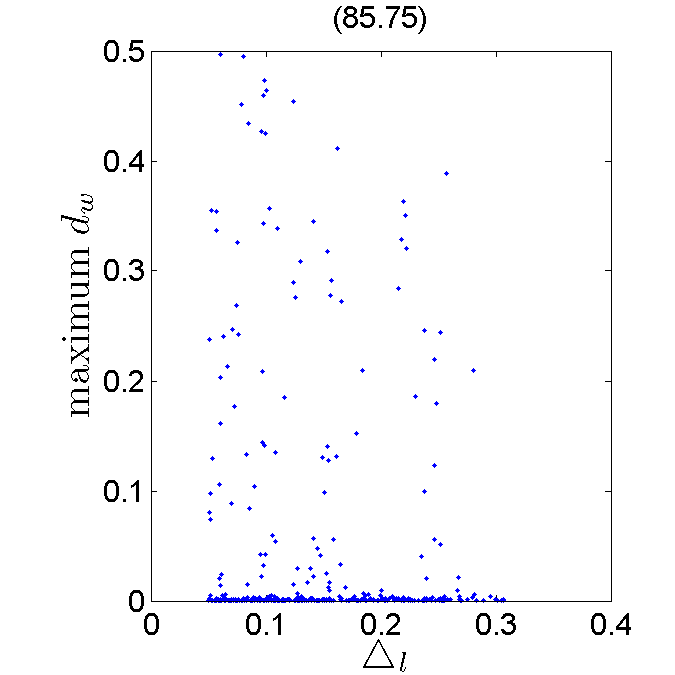} }
%
\subcaptionbox[]{$l = 4, K = 3$}[ 0.24\textwidth ]
{\includegraphics[width=0.24\textwidth]{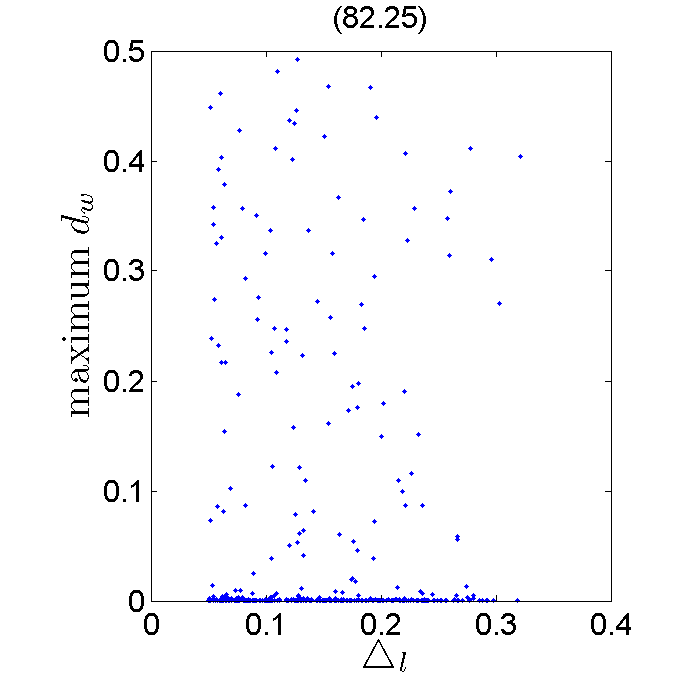} }
%

\subcaptionbox[]{$l = 1, K = 4$}[ 0.24\textwidth ]
{\includegraphics[width=0.24\textwidth]{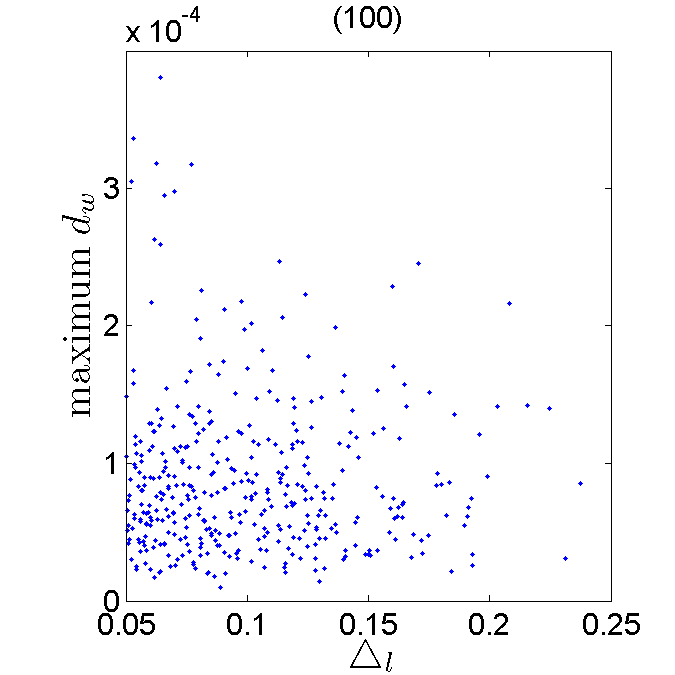} }
\subcaptionbox[]{$l = 2, K = 4$}[ 0.24\textwidth ]
{\includegraphics[width=0.24\textwidth]{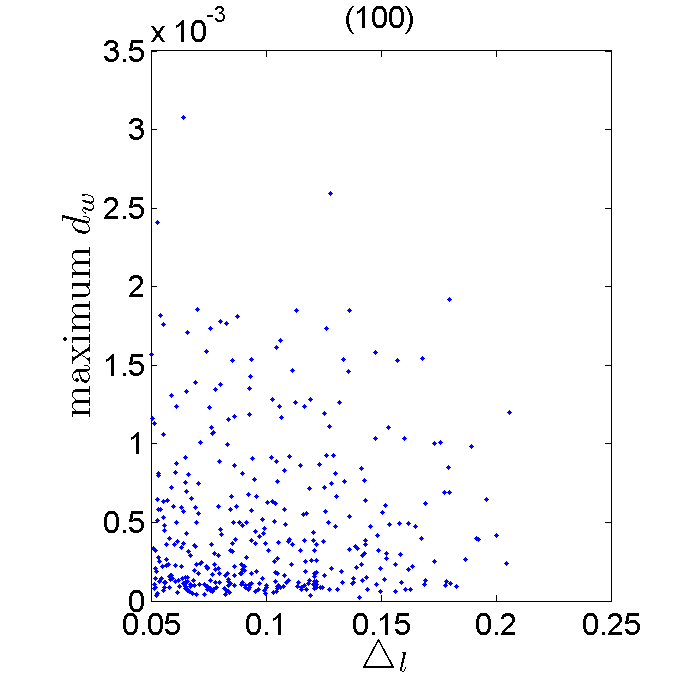} }
%
\subcaptionbox[]{$l = 3, K = 4$}[ 0.24\textwidth ]
{\includegraphics[width=0.24\textwidth]{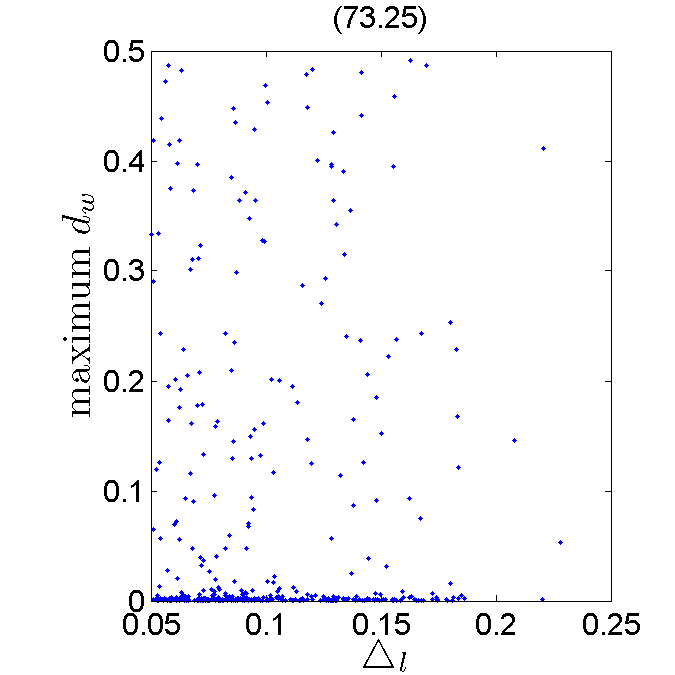} }
%
\subcaptionbox[]{$l = 4, K = 4$}[ 0.24\textwidth ]
{\includegraphics[width=0.24\textwidth]{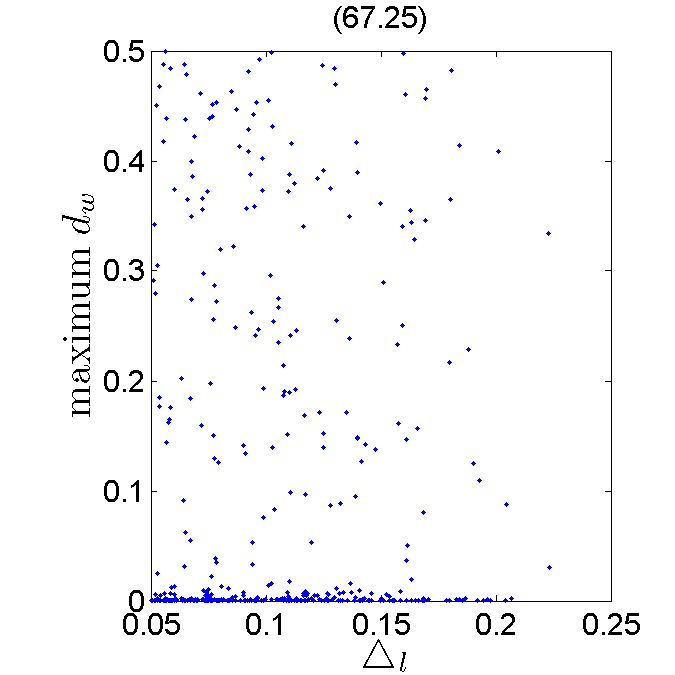} }
%

\subcaptionbox[]{$l = 1, K = 5$}[ 0.24\textwidth ]
{\includegraphics[width=0.24\textwidth]{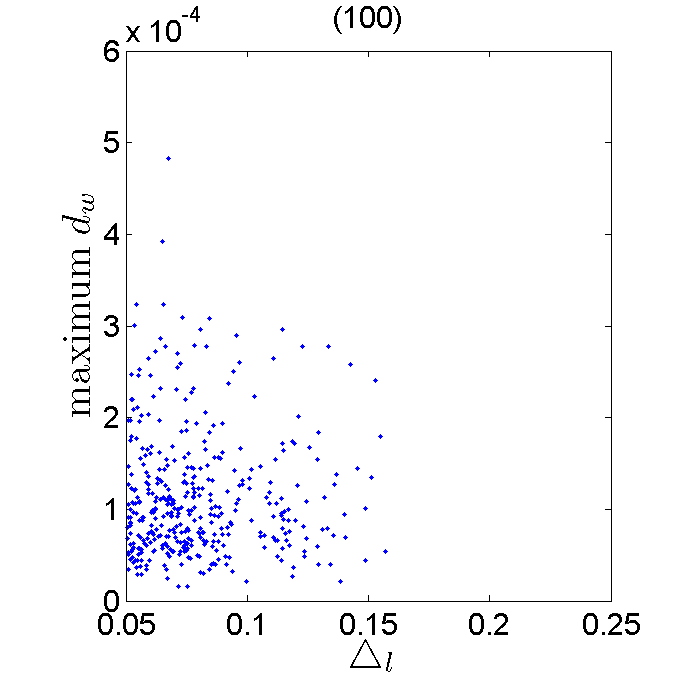} }
\subcaptionbox[]{$l = 2, K = 5$}[ 0.24\textwidth ]
{\includegraphics[width=0.24\textwidth]{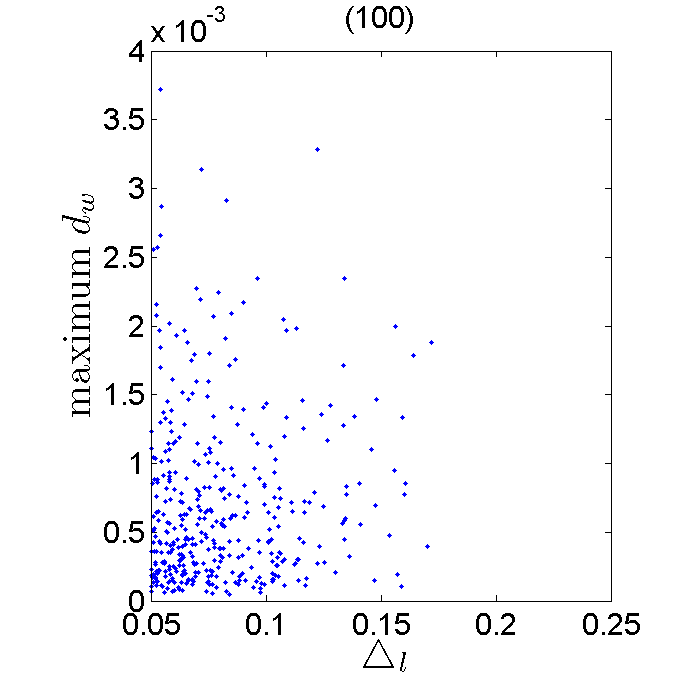} }
%
\subcaptionbox[]{$l = 3, K = 5$}[ 0.24\textwidth ]
{\includegraphics[width=0.24\textwidth]{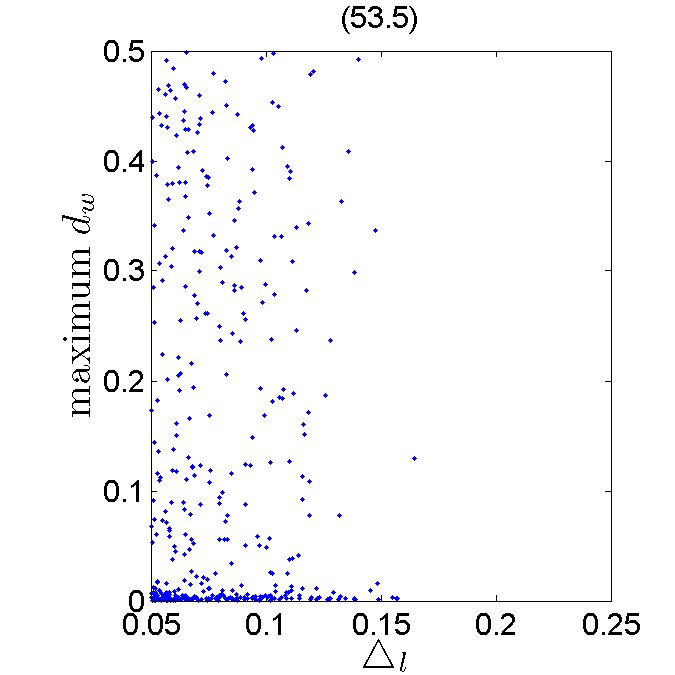} }
%
\subcaptionbox[]{$l = 4, K = 5$}[ 0.24\textwidth ]
{\includegraphics[width=0.24\textwidth]{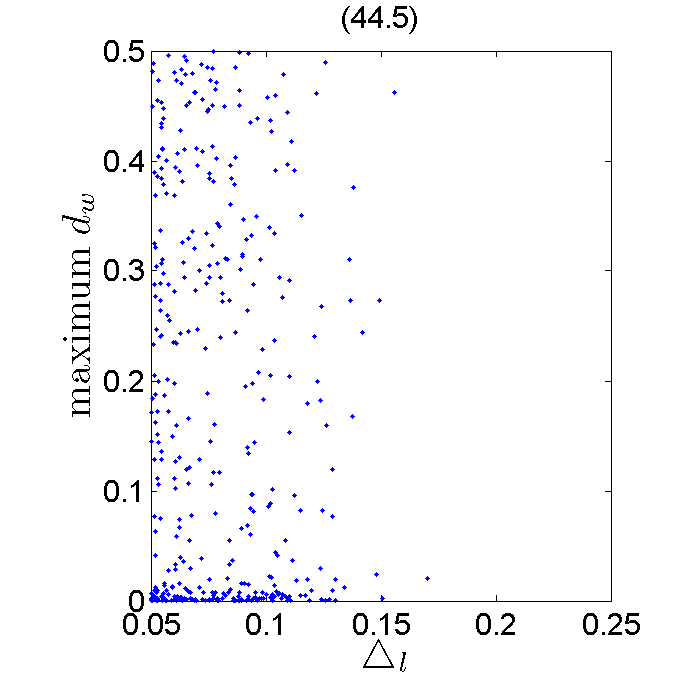} }
\captionsetup{width=0.98\linewidth}
\caption[Short Caption]{Scatter plots for maximum wrap around error ($d_{w,l,\max}$) v/s 
minimum separation ($\sep_l$) for $400$ Monte Carlo trials, 
with external Gaussian noise (standard deviation $5 \times 10^{-5}$). This is shown for $K \in \set{2,3,4,5}$ 
with $L = 4$ and $C = 0.6$. For each sub-plot, we mention the percentage of trials with 
$d_{w,l,\max} \leq 0.05$ in parenthesis.}
\label{fig:max_loc_err_noise}
\end{figure}

%
\begin{figure}[!ht]
\centering
\subcaptionbox[]{$l = 1, K = 2$}[ 0.24\textwidth ]
{\includegraphics[width=0.24\textwidth]{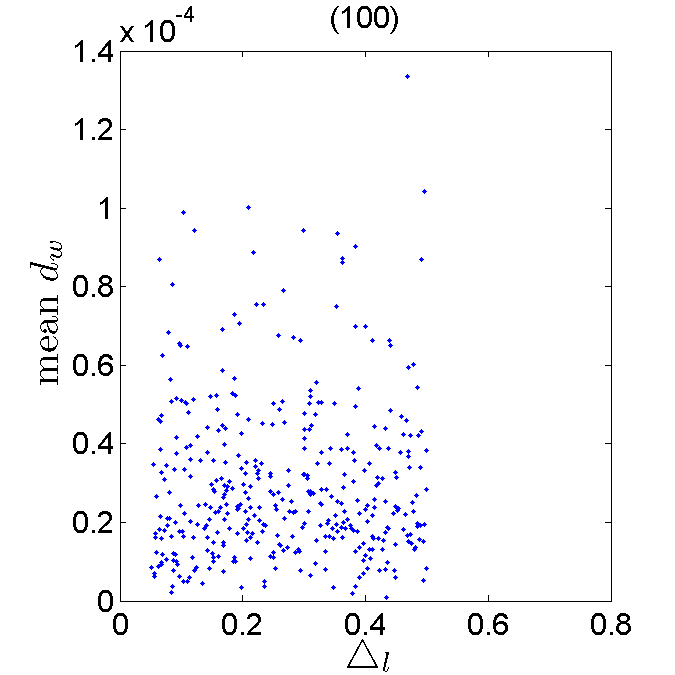} }
\subcaptionbox[]{$l = 2, K = 2$}[ 0.24\textwidth ]
{\includegraphics[width=0.24\textwidth]{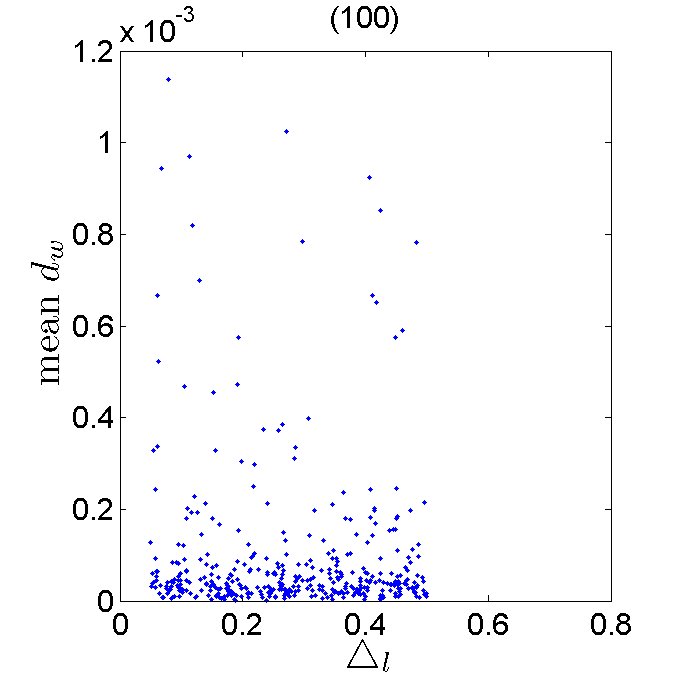} }
%
\subcaptionbox[]{$l = 3, K = 2$}[ 0.24\textwidth ]
{\includegraphics[width=0.24\textwidth]{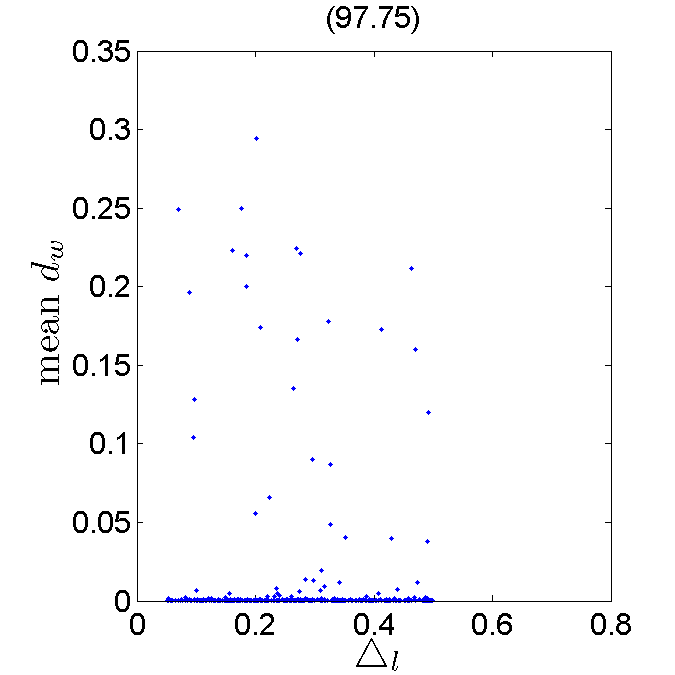} }
%
\subcaptionbox[]{$l = 4, K = 2$}[ 0.24\textwidth ]
{\includegraphics[width=0.24\textwidth]{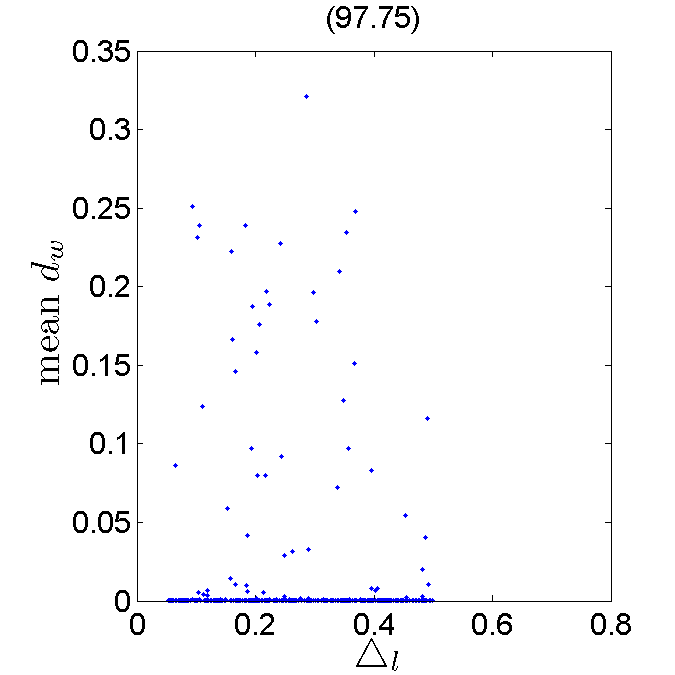} }
%

\subcaptionbox[]{$l = 1, K = 3$}[ 0.24\textwidth ]
{\includegraphics[width=0.24\textwidth]{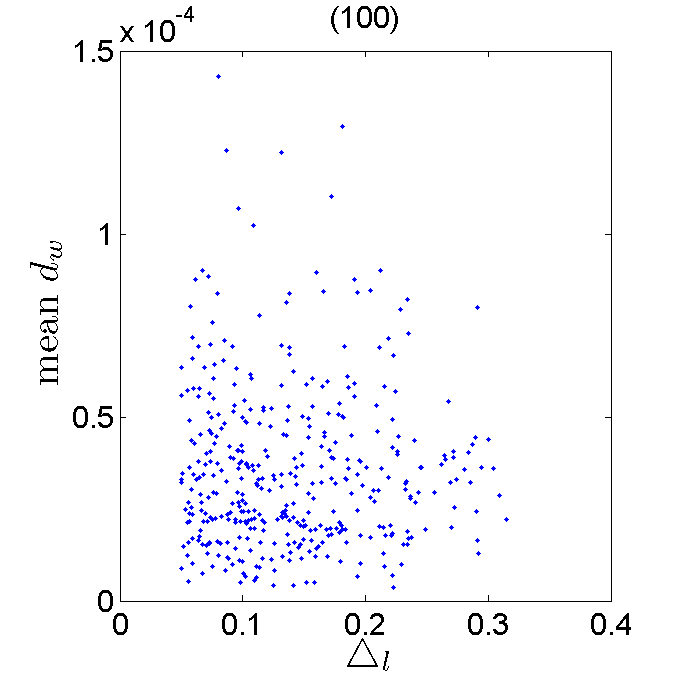} }
\subcaptionbox[]{$l = 2, K = 3$}[ 0.24\textwidth ]
{\includegraphics[width=0.24\textwidth]{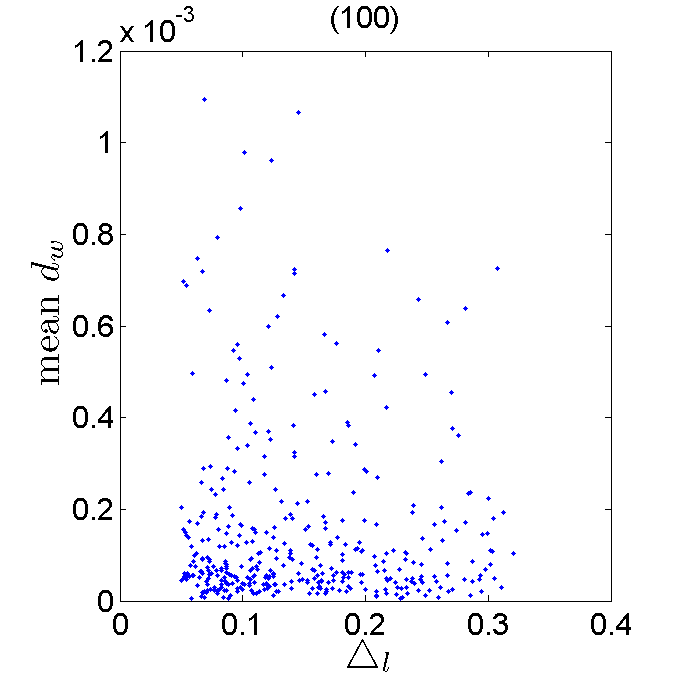} }
%
\subcaptionbox[]{$l = 3, K = 3$}[ 0.24\textwidth ]
{\includegraphics[width=0.24\textwidth]{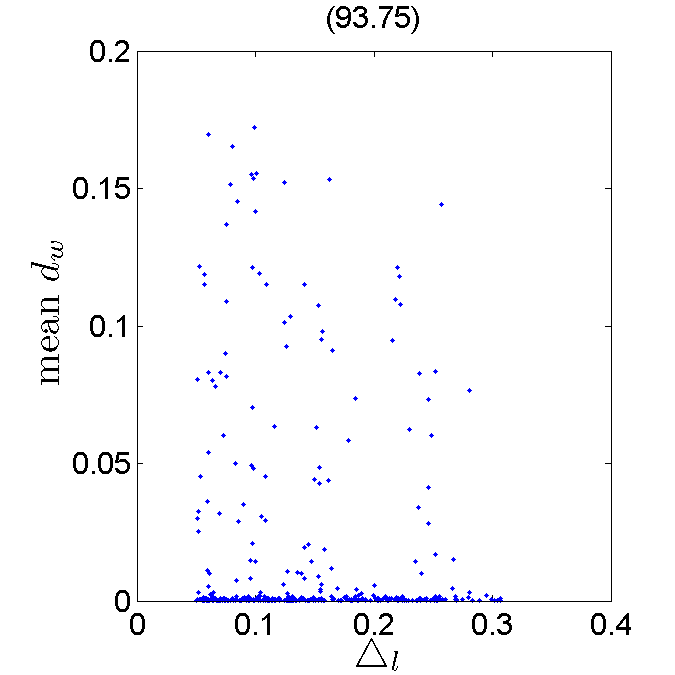} }
%
\subcaptionbox[]{$l = 4, K = 3$}[ 0.24\textwidth ]
{\includegraphics[width=0.24\textwidth]{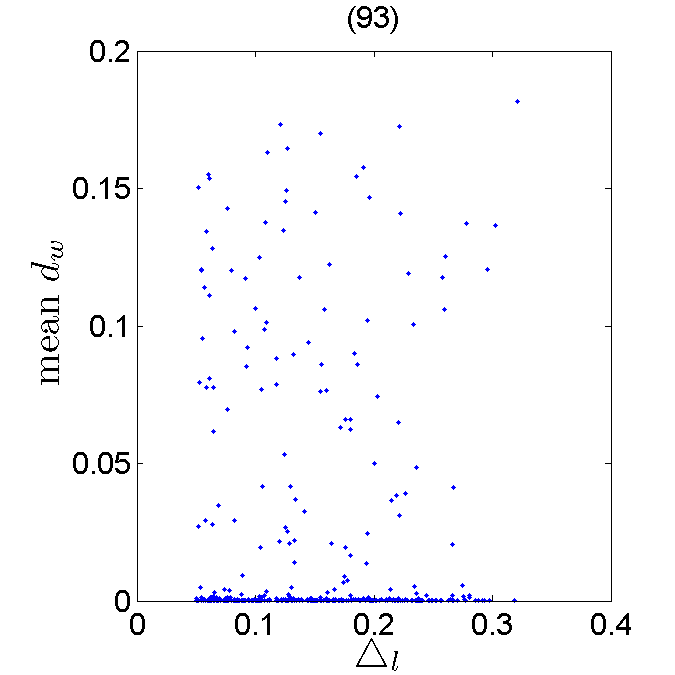} }
%

\subcaptionbox[]{$l = 1, K = 4$}[ 0.24\textwidth ]
{\includegraphics[width=0.24\textwidth]{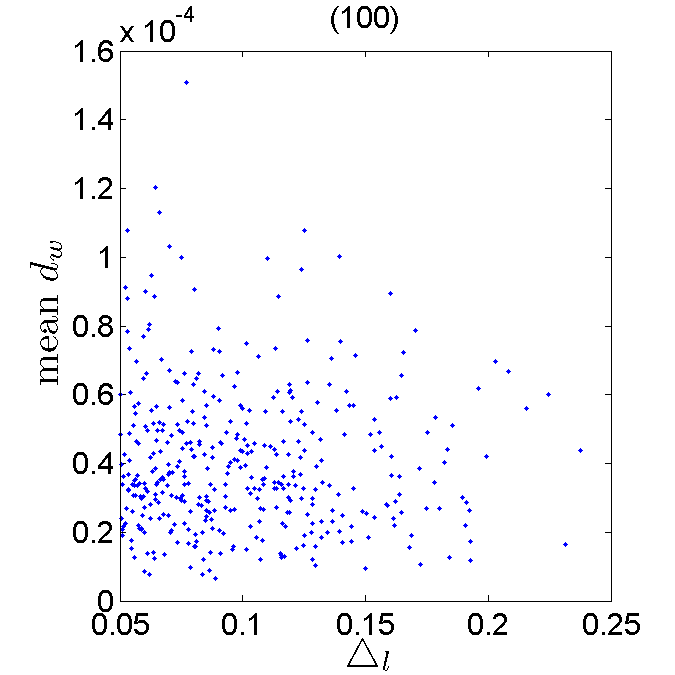} }
\subcaptionbox[]{$l = 2, K = 4$}[ 0.24\textwidth ]
{\includegraphics[width=0.24\textwidth]{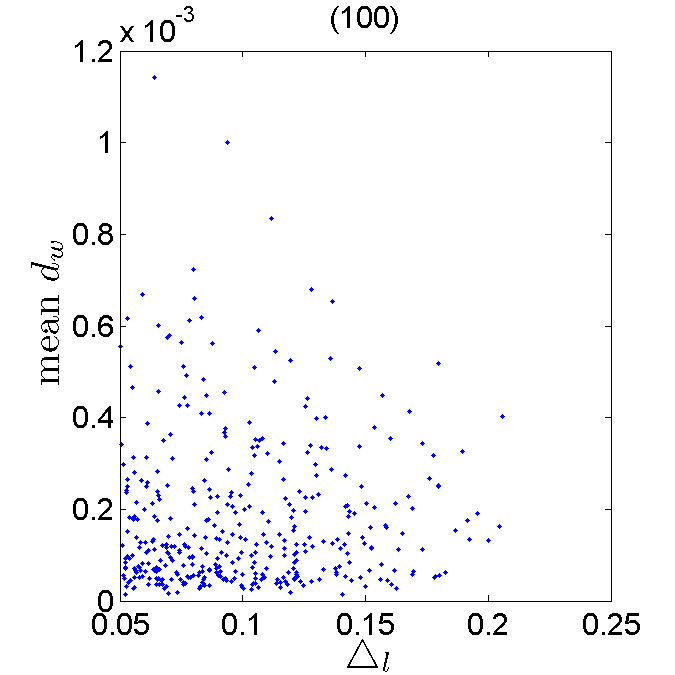} }
%
\subcaptionbox[]{$l = 3, K = 4$}[ 0.24\textwidth ]
{\includegraphics[width=0.24\textwidth]{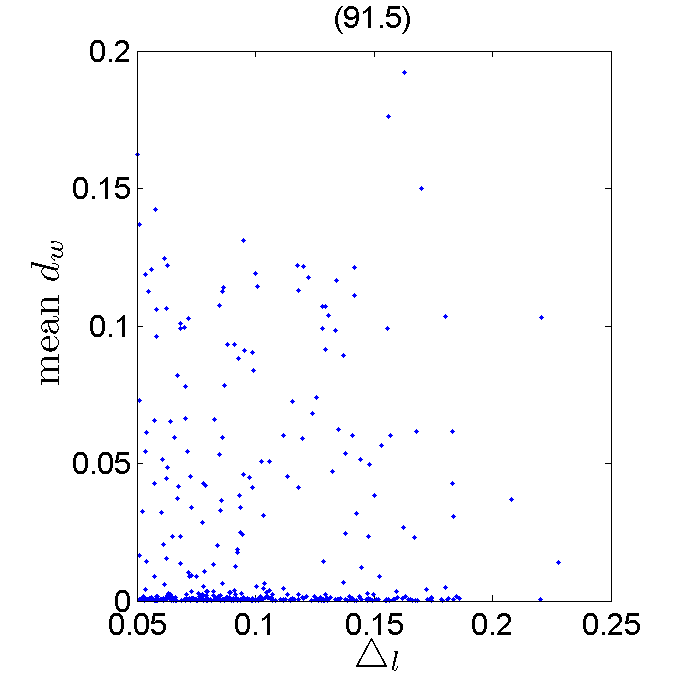} }
%
\subcaptionbox[]{$l = 4, K = 4$}[ 0.24\textwidth ]
{\includegraphics[width=0.24\textwidth]{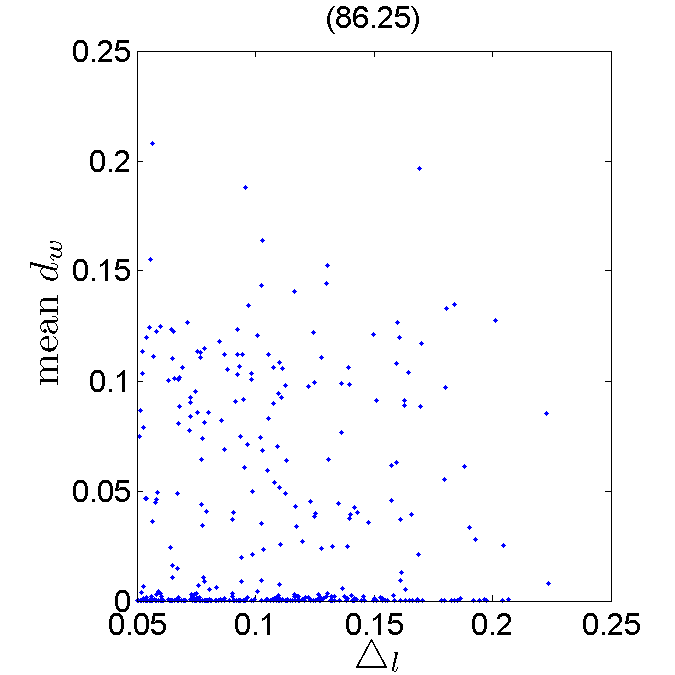} }
%

\subcaptionbox[]{$l = 1, K = 5$}[ 0.24\textwidth ]
{\includegraphics[width=0.24\textwidth]{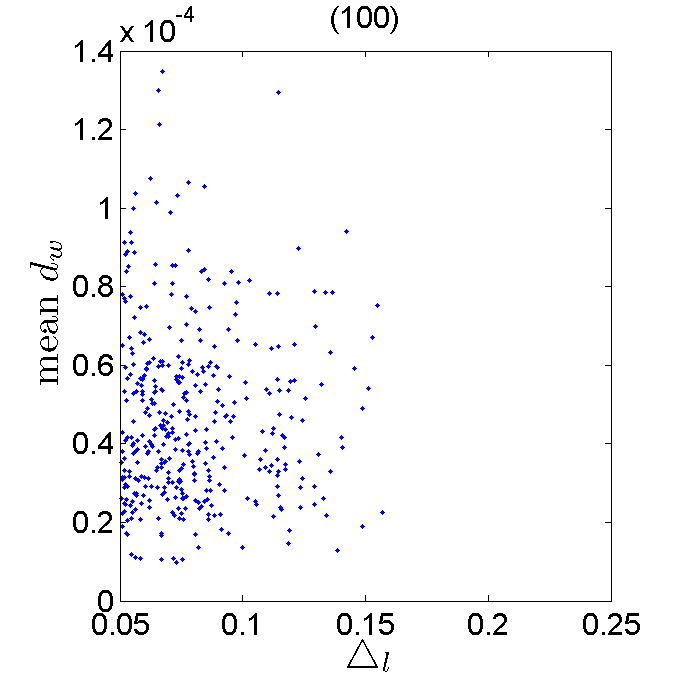} }
\subcaptionbox[]{$l = 2, K = 5$}[ 0.24\textwidth ]
{\includegraphics[width=0.24\textwidth]{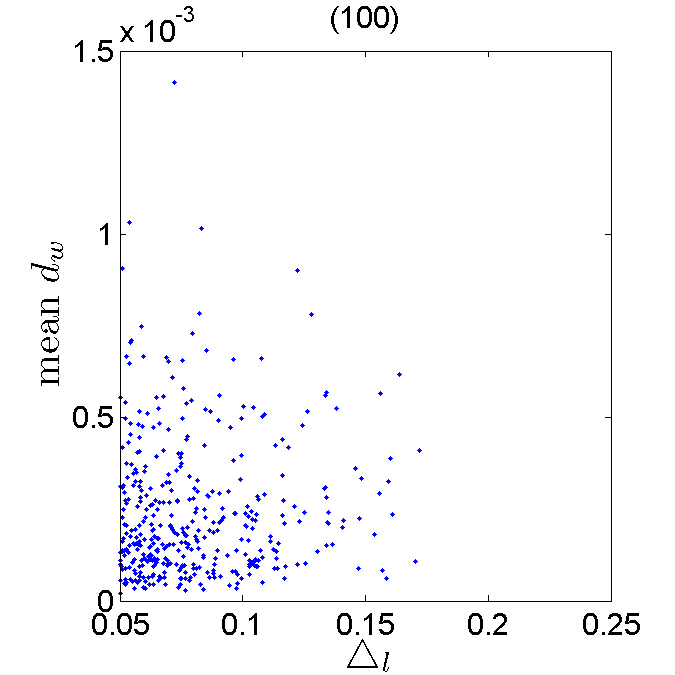} }
%
\subcaptionbox[]{$l = 3, K = 5$}[ 0.24\textwidth ]
{\includegraphics[width=0.24\textwidth]{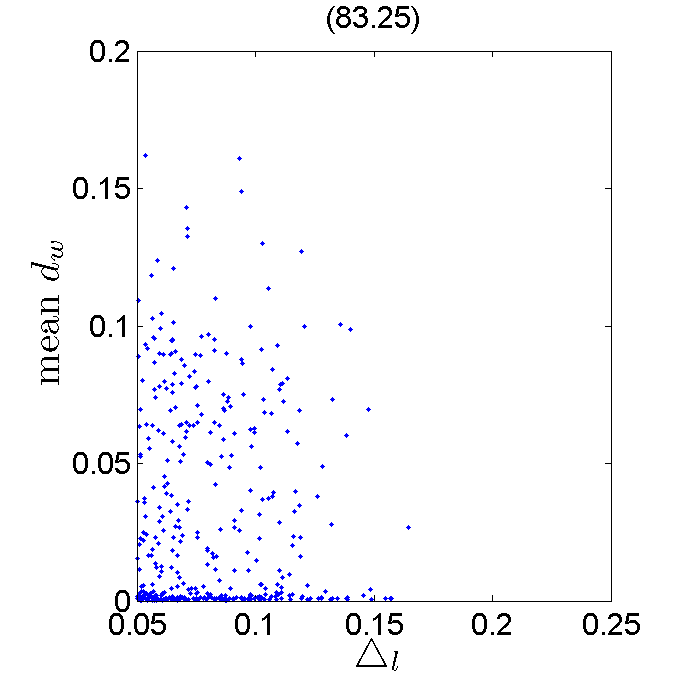} }
%
\subcaptionbox[]{$l = 4, K = 5$}[ 0.24\textwidth ]
{\includegraphics[width=0.24\textwidth]{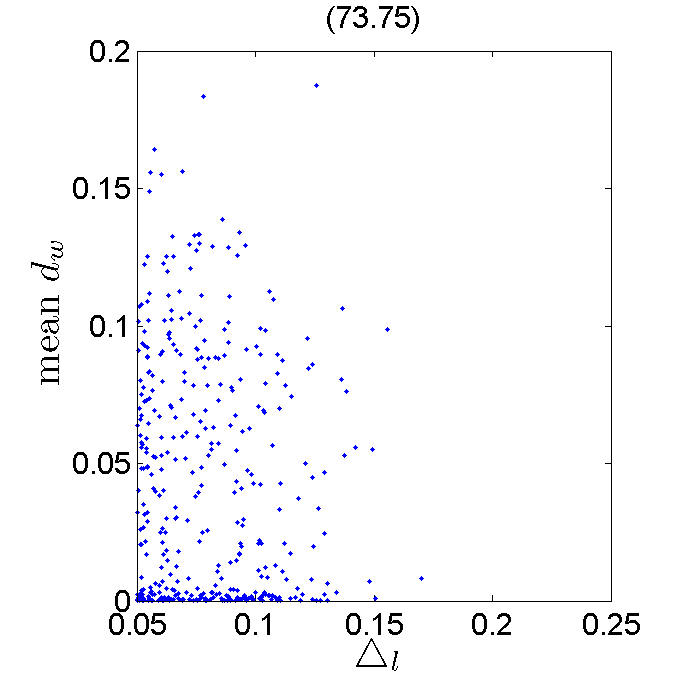} }
\captionsetup{width=0.98\linewidth}
\caption[Short Caption]{Scatter plots for mean wrap around error ($d_{w,l,\text{avg}}$) v/s minimum separation ($\sep_l$) 
for $400$ Monte Carlo trials, with external Gaussian noise (standard deviation $5 \times 10^{-5}$). 
This is shown for $K \in \set{2,3,4,5}$ with $L = 4$ and $C = 0.6$. For each sub-plot, we mention the percentage of trials with 
$d_{w,l,\text{avg}} \leq 0.05$ in parenthesis.}
\label{fig:mean_loc_err_noise}
\end{figure} 
\section{Discussion and Concluding remarks}
We now compare with closely related work and conclude with directions for future research.
%
%
\subsection{Related work on the multi-kernel unmixing super-resolution problem} \label{subsec:rel_work}
Despite is natural role in many practical problems,  the study of super-resolution under the 
presence of multiple kernels has not attracted much theoretical activity until recently. 
In \cite{li2017stable}, the authors introduce an interesting variant of the atomic norm approach to 
sparse measure reconstruction. They prove that the resulting convex optimization problem recovers 
the original measure in the noiseless case and they provide an interesting error bound in the noisy setting. 
The assumptions underlying this latter result is a standard separation assumption on the spike localization 
and a uniform random prior on the Fourier 
coefficients (when considered to lie in $\mathbb R/\mathbb Z$) of the point spread functions.
In comparison, our assumptions are quite different. In particular, we do not make any 
assumption about the randomness of the Fourier coefficients of the point spread function. 
Moreover, we use Moitra's Modified MP method instead of the atomic norm penalization considered in \cite{li2017stable}. 
As a main benefit of our approach, we do not need any hyper-parameter tuning when the signal is 
sufficiently larger than the noise level\footnote{what sufficiently larger means is elaborated 
on in Theorem \ref{thm:gen_case_main_noisy}}.

Another interesting work on multi-kernel super-resolution is the technique developed in \cite{slawski2010sparse}, where the setting is very close to the one of the present paper. A set of relevant modifications of the LASSO estimator and Matching Pursuit method, combined with post-processing techniques, are proposed in \cite{slawski2010sparse} and shown to perform well on real datasets. However, to the best of our knowledge, the practical value of these methods is not rigorously supported by theoretical results.

%
%


\subsection{Future directions} \label{subsec:fut_direcs}
In this paper, we provide a simple and intuitive algorithm for multi-kernel super-resolution, and also 
provide strong theoretical results for our approach. There are several directions for extending the results 
in this paper, we list two of them below. 

Firstly, our analysis assumes that the kernel variance parameters (i.e., $\kerfnvar_l$) are known exactly. In general, 
the analysis can be extended to the case where upper and lower estimates are available for each $\kerfnvar_l$.
Hence one could consider estimating the variance terms, and using these estimates with our algorithm. 
The choice of the method for estimating $(\kerfnvar_l)_l$ should be investigated with great care. 
One possible avenue is to use Lepski's method \cite{goldenshluger2011bandwidth} and its many recent 
variants and improvements (see for e.g., \cite{comte2013anisotropic}).

Next, it is natural to extend the techniques developed here to the multi-dimensional setting as 
multivariate signals are of high importance in many applications such as DNA sequencing and Mass Spectrometry.
This could be investigated using for instance \hemant{multivariate Prony-like methods} such as in
\cite{peter2015prony,kunis2016multivariate,Hua92multi,Andersson18,Ehler18,Pottsmultiexp13,Cuyt18}. 
\steph{Another interesting avenue is the one discovered in \cite{batenkov2018v2} where refined bounds on the condition number of the Vandermonde matrices are devised. It would be of particular interest to understand how such bounds could be employed in our framework in order to accomodate separation conditions below the threshold discovered by Moitra in \cite{moitra15}.}

\bibliographystyle{plain}
\bibliography{unmix_gauss_biblio}

\appendix

\section{Modified matrix pencil method} \label{sec:mmp_results}
%
%
%
\subsection{Proof of Moitra's theorem} \label{app:subsec_moitraMP_proof}
We now outline the steps of Moitra's proof \cite[Theorem 2.8]{moitra15} for completeness. In particular, we note that 
the result in \cite[Theorem 2.8]{moitra15} does not detail the constants appearing in the bounds, while we will do so here. 

To begin with, recall the notion of chordal metric for measuring distance 
between complex numbers.
%
\begin{definition} \label{def:chord_metric}
The chordal metric for $u,v \in \mathbb{C}$ is defined as
\begin{equation*}
\chord(u,v) := \frac{\abs{u-v}}{\sqrt{1+\abs{u}^2} \sqrt{1+\abs{v}^2}}.
\end{equation*}
\end{definition}
Denoting $s(u),s(v) \in \mathbb{S}^2$ to be points with $u,v$ as their respective stereographic projections on the plane, 
we also have $\chord(u,v) = \frac{1}{2}\norm{s(u) - s(v)}$. It is useful to note that if $(u,v) \in \mathbb{R}^2$ is 
the stereographic projection of a point $(a,b,c) \in \mathbb{S}^2$, then
\begin{align*}
(a,b,c) = \left(\frac{2u}{1+u^2+v^2}, \frac{2v}{1+u^2+v^2}, \frac{-1+u^2+v^2}{1+u^2+v^2} \right).
\end{align*} 
%
\begin{definition} \label{def:match_dist}
Let $(\lambda_i)_{i=1}^n$ denote generalized eigenvalues of ($A,B$), and also let $({\widehat \lambda}_i)_{i=1}^n$ 
denote generalized eigenvalues of  $(\widehat A, \widehat B)$. Then the matching distance with respect to $\chord$ is defined as 
\begin{equation*}
\mdist[(A,B), (\widehat A, \widehat B)] := \min_{\perm} \max_i \chord(\lambda_i, {\widehat \lambda}_{\perm(i)}),
\end{equation*}
where $\perm:[n] \rightarrow [n]$ denotes a permutation.
\end{definition}
The key to the analysis in \cite[Theorem 2.8]{moitra15} are the following two technical results. 
\begin{theorem}(\cite[Corollary VI.2.5]{stewart1990matrix}) \label{thm:sun_stewart}
Let $(A, B)$ and $(\hemant{\widehat A, \widehat B})$ be 
regular\footnote{\hemant{A pair $(A,B)$ is regular if there exist scalars 
$\alpha,\beta$ with $(\alpha,\beta) \neq (0,0)$ such that 
$\text{det}(\beta A - \alpha B) \neq 0$ \cite[Definition VI.1.2]{stewart1990matrix}. For $(A,B)$ to be regular, 
it suffices that $A$ and/or $B$ is full rank.}} 
pairs and further suppose that for some nonsingular
$\hemant{X,Y}$ we have $(X A \hemant{Y^H}, X B \hemant{Y^H}) = (I, D)$ where D is diagonal. 
Also let $e_i$ be the $i^{th}$ row of $X(A -\widehat A)\hemant{Y^H}$, 
$f_i$ be the $i^{th}$ row of $X(B - \widehat B)\hemant{Y^H}$ and set $\rho = \max_i\{\Vert e_i\Vert_1 + \Vert f_i\Vert_1\}$. 
If the following regions
\begin{align*}
		\mathcal G_i & = \left\{\mu \mid \chi(\hemant{D_{ii}},\mu) \le \rho \right\}
\end{align*}
are disjoint then the matching distance of the generalized eigenvalues of $(\hemant{\widehat B, \widehat A})$ to $\{D_{ii}\}_i$ 
is at most $\rho$ with respect to the chordal metric.
\end{theorem}
\hemant{Recall that for any pair $(A,B)$, and non-singular $X,Y$, the pair $(X A \hemant{Y^H}, X B \hemant{Y^H})$ 
is \emph{equivalent} to $(A,B)$. In particular, they both have the same generalized eigenvalues \cite[Theorem VI.1.8]{stewart1990matrix}.}
%
%
%
\begin{theorem}(\cite[Theorem 1.1]{moitra15}) \label{thm:moitra_vander_condbd}
Provided that $m > \frac{1}{\sep} + 1$, we have $\sigmamax^2 \leq m+\frac{1}{\sep}-1$ and 
$\sigmamin^2 \geq m-\frac{1}{\sep}-1$. Consequently, the condition number of $V$ satisfies 
\begin{align*}
	\condnum^2 & \le \frac{m+\frac{1}{\sep}-1}{m-\frac{1}{\sep}-1}.
\end{align*}
\end{theorem}
%
%
%
\paragraph{Step 1: Recovering each $t_j$.} 
Let $\Uhat$ denote the top $K$ left singular vectors of $\widetilde H_0$. The following lemma shows that there exists an 
orthonormal basis for the column span of $V$ that is well aligned to $\Uhat$.
\begin{lemma}(\cite[Lemma 2.7]{moitra15}) \label{lem:subspace_bound}
		\hemant{If $\norm{E}_2 < \sigma_K(VD_{u'}V^H)$, then there} exists a matrix $U$ such that 
		\begin{align*}
		\norm{U-\widehat U}_2 & \le \frac{2 \norm{E}_2 }{\sigma_K(V\hemant{D_{u'}}V^H)- \norm{E}_2}
		\end{align*}
		and the columns of $U$ form an orthonormal basis for $V$ and those of $\widehat U$ form an orthogonal 
		basis of the $K$ ``largest'' singular vectors of \hemant{$\widetilde H_0$}.
	\end{lemma}
%
%
%
\begin{remark}
In the original statement in \cite[Lemma 2.7]{moitra15} the denominator term is  $\sigma_K(V\hemant{D_{u'}}V^H)$. 
But from Wedin's bound \cite{Wedin1972}, we observe that the denominator should be 
$\sigma_K(V\hemant{D_{u'}}V^H + E)$ instead. Since 
$\sigma_K(V\hemant{D_{u'}}V^H + E) \geq \sigma_K(V\hemant{D_{u'}}V^H) - \norm{E}_2$ (Weyl's inequality \cite{Weyl1912}),
the statement in Lemma \ref{lem:subspace_bound} follows.
\end{remark}
By projecting $H_0, H_1$ to $U$ we obtain
\begin{align*}
	A  = U^H H_0 U =  U^HV\hemant{D_{u'}}V^H U \quad \text{and} \quad B  = U^H H_1 U =  U^HV\hemant{D_{u'}}D_\alpha V^H U.
\end{align*}
As easily checked, the generalized eigenvalues of ($B,A$) are $\alpha_1,\dots,\alpha_K$. Similarly,
\begin{align*}
	\widehat A = \widehat U^H\widetilde H_0\widehat U = \widehat U^H(V\hemant{D_{u'}}V^H+E)\widehat U \quad \text{and} \quad \widehat B = \widehat U^H\widetilde H_1\widehat U  = \widehat U^H(V\hemant{D_{u'}}D_\alpha V^H+F)\widehat U.
\end{align*}
The generalized eigenvalues of ($\widehat B, \widehat A$) are perturbed versions of that of $\hemant{(B,A)}$.  
Since\footnote{\hemant{Note that $D_{u'}$ and $D_u$ have the same singular values 
as $\abs{u'_i} = \abs{u_i}$ for each $i$.}} $\sigma_K(V\hemant{D_{u'}}V^H)  \ge \sigma_{\min}^2 u_{\min}$, 
therefore \hemant{if $\norm{E}_2 < \sigma^2_{\min} u_{\min}$ holds, then} Lemma \ref{lem:subspace_bound} 
gives us the bound 
\begin{align} \label{eq:tau_def}
	\norm{U-\widehat U}_2 & \le \frac{2 \norm{E}_2}{\sigma^2_{\min} u_{\min}- \norm{E}_2} = \hemant{\tau}.
\end{align}
%
%
\hemant{Let us define the matrices $X := \hemant{D^{-1/2}_{u'}}V^\dagger U$, 
$Y:= \hemant{(D^{-1/2}_{u'})^H} V^\dagger U$. Clearly, both $X$ and Y are non-singular. Moreover, one can easily 
verify that $X A Y^H = I$ and $X B Y^H = D_{\alpha}.$} We now bound the $\hemant{\ell_2}$-norm of each row of 
$X(A-\widehat A)\hemant{Y^H}$ and $X(B-\widehat B)\hemant{Y^H}$ using what we have seen so far. 
To begin with, following the steps in \cite{moitra15}, we obtain
\begin{align*}
  \hemant{\norm{A - \widehat A}_2} &\leq \hemant{\Vert E \Vert_2 + 	\sigma^2_{\max} \umax \hemant{\tau} (2+\hemant{\tau})} \\
	\Rightarrow \Vert X(A-\widehat A) \hemant{Y^H} \Vert_2 & \le \frac{\Vert E \Vert_2 + 
	\sigma^2_{\max} \hemant{\umax \tau} (2+\hemant{\tau})}{u_{\min}\ \sigma^2_{\min}}  
	= \frac{\Vert E\Vert_2}{\umin \sigmamin^2} + \kappa^2 \frac{u_{\max}}{u_{\min}} \ \hemant{\tau}(2+\hemant{\tau}).
\end{align*}
The same bound holds on $\hemant{\norm{X(B-\widehat B)\hemant{Y^H}}_2}$ with $E$ replaced by $F$.
%
Let $e_i,f_i$ denote the $i^{th}$ row of $X(A-\widehat A)\hemant{Y^H}$, $X(B-\widehat B)\hemant{Y^H}$ respectively and let 
$\rho = \max_{i} \left\{ \Vert e_i\Vert_1+ \Vert f_i\Vert_1\right\}$. \hemant{Since for each $i$, 
$\norm{e_i}_2 \leq \norm{X(A-\widehat A)\hemant{Y^H}}_2$ and $\norm{f_i}_2 \leq \norm{X(B-\widehat B)\hemant{Y^H}}_2$, 
therefore by using} the fact $\rho \le \sqrt{K} \ \max_{i} \left\{ \Vert e_i\Vert_2 + \Vert f_i\Vert_2\right\}$, 
we get 
\begin{align}
	\rho & \le \sqrt{K} \left(\frac{\Vert E\Vert_2+\Vert F\Vert_2}{u_{\min} \ \sigma^2_{\min}} + 
	\hemant{\frac{2\kappa^2 u_{\max}}{u_{\min}}\left[\tau(2+\tau)\right]}\right). \label{boundrho}
\end{align}
%
%
Assume that 
%
	$\Vert E\Vert_2, \ \Vert F\Vert_2 \le \delta \ \frac{u_{\min}\sigma^2_{\min}}{2}$ 
%
for some $\delta \in (0,1)$. Then,
\begin{align*}
	\hemant{\tau(2+\tau)} \leq \frac{\delta}{1-\delta/2}\left(2+\frac{\delta}{1-\delta/2}\right) 
	= \frac{2\delta}{(1-\delta/2)^2} 	\le 8\ \delta.
\end{align*}
%
\hemant{Applying these bounds to \eqref{boundrho} leads to the bound} 
\begin{align*}
	\rho & \le \sqrt{K} \ \delta \left( 1+16 \ \kappa^2 \frac{u_{\max}}{u_{\min}}\right).
\end{align*}
On the other hand, one can easily verify that $\Vert E\Vert_2, \Vert F\Vert_2 \le m \noisemax$. Therefore, if 
$\noisemax \le \delta \ \frac{u_{\min}\ \sigma^2_{\min}}{2 m}$, then we get
$\rho \le \sqrt{K} \ \delta \left( 1+16 \ \kappa^2 \frac{u_{\max}}{u_{\min}}\right).$

%
%
\hemant{Our immediate goal now is to use Theorem \ref{thm:sun_stewart}. $(A,B)$ are clearly regular 
since $A,B$ are full rank matrices. We will show that if $\delta$ is small enough, then $\widehat A$ is 
full rank, which in turn implies that $(\widehat A, \widehat B)$ is regular. To this end, recall that 
if $\norm{E}_2 < \sigma^2_{\min} u_{\min}$ then 
$\norm{A - \widehat A}_2 \leq \Vert E \Vert_2 + \sigma^2_{\max} \umax \tau (2+\tau)$ 
holds with $\tau$ as in \eqref{eq:tau_def}. If $\norm{E}_2 \leq \frac{\delta}{2} \sigma^2_{\min} u_{\min}$ for 
$\delta \in (0,1)$ then we have seen that this implies $\tau(2+\tau) \leq 8\delta$; this in turn implies that 
$\norm{A - \widehat A}_2 \leq 8\delta \sigma_{\max}^2\umax + \frac{\delta}{2}\umin\sigma_{\min}^2$. Finally, 
using Weyl's inequality, we have that 
\begin{equation*}
\sigma_{K}(\widehat A) 
\geq \sigma_K(A) - \norm{A - \widehat A}_2 
\geq \sigma^2_{\min} u_{\min} - 8\delta \sigma_{\max}^2\umax - \frac{\delta}{2}\umin\sigma_{\min}^2.
\end{equation*}
Hence $\sigma_{K}(\widehat A) > 0$ if
\begin{equation} \label{eq:delt_cond_Ahat_frank}
\delta 
< \frac{\sigma^2_{\min} u_{\min}}{8\sigma_{\max}^2\umax + \frac{1}{2}\umin\sigma_{\min}^2} 
= \frac{2}{16\kappa^2\frac{\umax}{\umin} + 1}.
\end{equation}
}

Now let us note that $\chord(\alpha_i,\alpha_j) \geq 2 d_w(t_i,t_j) \geq 2\sep$ (see Appendix \ref{app:rel_dw_chord}) for 
$i \neq j$. So if $u$ is such that $\chord(\alpha_i,u) \leq \epsilon$, then this would imply  
$\chord(\alpha_j,u) \geq 2\sep - \epsilon$, and hence if additionally $\epsilon < \sep$ holds, then the regions 
$\set{u \ | \ \chord(\alpha_i,u) \leq \epsilon}$ will be disjoint. Denote $(\widehat \lambda_j)_{j=1}^{n}$ to be the 
generalized eigenvalues of ($\hemant{\widehat B, \widehat A}$).  Therefore, if $\delta \in (0,1)$ satisfies 
$\delta < \frac{\sep}{\sqrt{K}} \left( 1+16 \ \kappa^2 \frac{u_{\max}}{u_{\min}}\right)^{-1}$ 
(note that the bound is already less than $1$ \hemant{and also subsumes \eqref{eq:delt_cond_Ahat_frank}}),  
then from Theorem \ref{thm:sun_stewart} we know that there exists a permutation $\perm: [k] \rightarrow [k]$ such that 
$\chord(\widehat \lambda_{\perm(i)}, \alpha_i) \leq \sqrt{K} \ \delta \left( 1+16 \ \kappa^2 \frac{u_{\max}}{u_{\min}}\right)$ for $i = 1,\dots,n$.
Let $\widehat \alpha_i = {\widehat \lambda_i}/\abs{\widehat \lambda_i}$, and set 
$\sqrt{K} \ \delta \left( 1+16 \ \kappa^2 \frac{u_{\max}}{u_{\min}}\right) = \varepsilon/C$ where $C = 10 + \frac{1}{2\sqrt{2}}$ and 
$0 \leq \varepsilon < \min\set{1,C\sep}$. This means $\chord(\widehat \lambda_{\perm(i)}, \alpha_i) \leq \varepsilon/C$ holds, and so,  
from\footnote{The Proposition requires $0 \leq \varepsilon/C \leq 1/4$, which is the case here.} Proposition \ref{app:prop_wrap_chord_notdisk} in 
Appendix \ref{app:rel_dw_chord_notdisk}, this implies $d_w(\est{t}_{\perm(i)},t_i) \leq \varepsilon$. 
Putting it together with the earlier condition on $\delta$, we get that if 
\begin{equation} \label{eq:noisebd_moitra_1}
\noisemax \leq \varepsilon \frac{\umin \sigma_{\min}^2}{2mC\sqrt{K}} \left( 1+16 \ \kappa^2 \frac{u_{\max}}{u_{\min}}\right)^{-1} 
\end{equation}
is satisfied, for $0 \leq \varepsilon < \min\set{1,C\sep}$, then $d_w(\est{t}_{\perm(i)},t_i) \leq \varepsilon$.
%
\paragraph{Step 2: Recovering each $u'_j$.} 
Note that $d_w(\widehat t_i, \widehat t_j) \geq \sep - 2\varepsilon$ for all $i\neq j$, so we assume $\varepsilon < \sep/2$ from now. 
Also recall that we form the Vandermonde matrix $\widehat V \in \mathbb{C}^{m \times K}$ using $\widehat \alpha_j = \exp(-\iota 2\pi \widehat t_j)$. 
Then, the estimate $\widehat u = \widehat V^{\dagger} v \in \mathbb{C}^K$ satisfies 
$$\widehat u' = \widehat V^{\dagger} v = \widehat V^{\dagger} Vu' + \widehat V^{\dagger} \eta.$$

Denote $\tilde{\perm} = \perm^{-1}$ and let $u'_{\tilde{\perm}} \in \mathbb{C}^K$ denote the permuted version of $u'$ w.r.t $\tilde{\perm}$. 
Also, let $V_{\tilde{\perm}}$ be formed by permuting the columns of $V$ w.r.t $\tilde{\perm}$. Then, 
\begin{align}
\widehat u' - u'_{\tilde{\perm}} 
&= ( \widehat V^{\dagger} V_{\tilde{\perm}} u'_{{\tilde{\perm}}} - u'_{\tilde{\perm}}) + \widehat V^{\dagger} \eta \nonumber \\
&= (\widehat V^{\dagger} V_{\tilde{\perm}} - \widehat V^{\dagger} \widehat V)u'_{\tilde{\perm}} + \widehat V^{\dagger} \eta \nonumber \\
\Rightarrow \norm{\widehat u' - u'_{\tilde{\perm}}}_2 &\leq \norm{\widehat V^{\dagger}}_2 (\norm{V_{\tilde{\perm}} - \widehat V}_2 
\norm{u'_{\tilde{\perm}}}_2 + \norm{\eta}_2).  
\label{eq:uest_bd_temp1}
\end{align}
Now, $\widehat V^{\dagger} = (\widehat V^H \widehat V)^{-1} \widehat V^H$ so 
$\norm{\widehat V^{\dagger}}_2 = (\sigmamin(\widehat V))^{-1} \leq (m - \frac{1}{\sep-2\varepsilon} -1)^{-1/2}$ (cf., Theorem \ref{thm:moitra_vander_condbd}). 
Next, the magnitude of each entry in $V_{\tilde{\perm}} - \widehat V$ can be verified to be upper bounded by 
$2\pi m \max_{i} d_w(t_i,\widehat t_{\tilde{\perm}(i)}) \leq 2\pi m \varepsilon$ (see Proposition \ref{app:prop_useful_res} in Appendix \ref{subsec:app_another_use_res}). 
Thus, $\norm{V_{\tilde{\perm}} - \widehat V}_2 \leq \norm{V_{\tilde{\perm}} - \widehat V}_F \leq 2\pi m^{3/2} K^{1/2} \varepsilon$. 
Also, $\norm{u'_{\tilde{\perm}}}_2 \leq \sqrt{K}\umax$ and using \eqref{eq:noisebd_moitra_1}, we obtain
$$\norm{\eta}_2 \leq \sqrt{m} \noisemax \leq \varepsilon \frac{u_{\min} \sigma^2_{\min}}{2 C \sqrt{m K}}\left( 1+16 \ \kappa^2 \frac{u_{\max}}{u_{\min}}\right)^{-1}.$$
Plugging these bounds in \eqref{eq:uest_bd_temp1}, we obtain
\begin{align} \label{eq:uest_bd_temp2}
\norm{\widehat u' - u'_{\tilde{\perm}}}_{\infty} \leq \norm{\widehat u' - u'_{\tilde{\perm}}}_2 \leq \frac{2\pi m^{3/2} K \umax \varepsilon + \varepsilon \frac{u_{\min} \sigma^2_{\min}}{2 C \sqrt{m K}}\left( 1+16 \ \kappa^2 \frac{u_{\max}}{u_{\min}}\right)^{-1}}{(m - \frac{1}{\sep-2\varepsilon} -1)^{1/2}}.
\end{align}
\paragraph{Step 3: Recovering each $u_j$.} 
Given $\widehat u' \in \mathbb{C}^K$, we obtain our final estimate $\widehat u_j = \exp(-\iota 2\pi s_0 \widehat t_j) \widehat u'_j$ for $j=1,\dots,K$.
Denoting 
\begin{align*}
\widehat D &= \text{diag}(\exp(-\iota2\pi s_0 \widehat t_1),\cdots,\exp(-\iota2\pi s_0 \widehat t_k)), \\  
D_{\tilde{\perm}} &= \text{diag}(\exp(-\iota 2\pi s_0 t_{\tilde{\perm}(1)}),\cdots,\exp(-\iota 2\pi s_0 t_{\tilde{\perm}(k)})), 
\end{align*}
note that $\widehat u = \widehat{D} \widehat u'$ and $u_{\tilde{\perm}} = D_{\tilde{\perm}} u'_{\tilde{\perm}}$. From this, we obtain
\begin{align}
\widehat u - u_{\tilde{\perm}} 
&= \widehat{D} \widehat u' - D_{\tilde{\perm}} u'_{\tilde{\perm}} \nonumber \\
&= \widehat{D} \widehat u' - \widehat{D}u'_{\tilde{\perm}} + \widehat{D}u'_{\tilde{\perm}} - D_{\tilde{\perm}} u'_{\tilde{\perm}} \nonumber \\
&= \widehat D (\widehat u' - u'_{\tilde{\perm}}) + (\widehat D - D_{\tilde{\perm}}) u'_{\tilde{\perm}} \nonumber \\
\Rightarrow \norm{\widehat u - u_{\tilde{\perm}}}_{\infty} 
&\leq \norm{\widehat u' - u'_{\tilde{\perm}}}_{\infty} + \norm{\widehat D - D_{\tilde{\perm}}}_{\infty} \norm{u'_{\tilde{\perm}}}_{\infty} \label{eq:uest_bd_temp3}
\end{align}
where in the last line, we used $\norm{\widehat D}_{\infty} = 1$. Using Proposition \ref{app:prop_useful_res} in Appendix \ref{subsec:app_another_use_res} ,
we readily obtain the bound $\norm{\widehat D - D_{\tilde{\perm}}}_{\infty} \leq 2\pi s_0 \max_i d_w(\widehat t_i, t_{\tilde{\perm}(i)}) \leq 2\pi s_0 \varepsilon$.
Using $\norm{u'_{\tilde{\perm}}}_{\infty} = \umax$ and \eqref{eq:uest_bd_temp2} in \eqref{eq:uest_bd_temp3}, we finally obtain 
\begin{align} \label{eq:uest_bd_fin}
\norm{\widehat u - u_{\tilde{\perm}}}_{\infty} 
&\leq \frac{2\pi m^{3/2} K \umax \varepsilon + \varepsilon \frac{u_{\min} \sigma^2_{\min}}{2 C \sqrt{m K}}\left( 1+16 \ \kappa^2 \frac{u_{\max}}{u_{\min}}\right)^{-1}}{(m - \frac{1}{\sep-2\varepsilon} -1)^{1/2}} 
+ 2\pi\umax s_0 \varepsilon.
\end{align}
%
%

\subsection{Proof of Corollary \ref{corr:moitra_MP}} \label{app:subsec_corr_moitra_MP}
We need only make the following simple observations. 
Firstly, $\sigmamin^2 \leq \sigmamax^2 \leq m + \frac{1}{\sep} + 1 < 2m$, since $m > 1/\sep + 1$. 
Also, since $\condnum^2 \geq 1$, we get 
\begin{align*}
\frac{u_{\min} \sigma^2_{\min} \varepsilon}{2 C\sqrt{m K}} \left(1+16 \ \kappa^2 \frac{u_{\max}}{u_{\min}}\right)^{-1} 
&< \frac{u_{\min}}{C\sqrt{K}} \left( 1+16 \frac{u_{\max}}{u_{\min}}\right)^{-1} \sqrt{m} \varepsilon.
\end{align*}
Next, one can easily verify that for $m \geq \frac{2}{\sep-2\varepsilon} + 1$, we have 
$\frac{\sqrt{m}}{\sqrt{m - \frac{1}{\sep-2\varepsilon} - 1}} \leq 2$. Using these observations in  \eqref{eq:mp_thm_bd}, along with  
$m < \frac{2}{\sep(1-c)} + 1$, we obtain the stated bound on $\norm{\hat u_{\perm} - u}_{\infty}$.

Next, let us note that the condition $m \geq \frac{2}{\sep - 2\varepsilon} + 1 \geq \frac{2}{\sep} + 1$ leads to the following 
simple observations.
\begin{itemize}
\item[(i)] 
\begin{equation*}
\frac{m}{\sigma_{\min}^2} \leq \frac{m}{m-\frac{1}{\sep} - 1} \leq \frac{1}{1-\frac{1/\sep + 1}{2/\sep + 1}} \leq 5/2.
\end{equation*}

\item[(ii)]
\begin{equation*}
\condnum^2 = \frac{m + \frac{1}{\sep} - 1}{m - \frac{1}{\sep} - 1} \leq 3.
\end{equation*}
\end{itemize}
Plugging these in \eqref{eq:moitra_thm_noise_cond}, we have that it suffices for $\noisemax$ to satisfy
\begin{equation*}
\noisemax \leq \varepsilon \frac{\umin}{5C\sqrt{K}} \left( 1+48 \frac{u_{\max}}{u_{\min}}\right)^{-1}
\leq \varepsilon \frac{\umin \sigma_{\min}^2}{2mC\sqrt{K}} \left( 1+16 \ \kappa^2 \frac{u_{\max}}{u_{\min}}\right)^{-1}.
\end{equation*}
%
\section{Auxiliary results} \label{sec:aux_results}
%
\subsection{Useful relation involving wrap around metric and chordal metric for points on the unit disk} \label{app:rel_dw_chord}
For $t_1,t_2 \in [0,1)$ let $\alpha_1 = \exp(\iota 2\pi t_1)$, $\alpha_2 = \exp(\iota 2\pi t_2)$ denote their 
representations on the unit disk. We have 
\begin{align*}
\abs{\alpha_1 - \alpha_2} 
&= \abs{1 - \exp(\iota 2\pi (t_2-t_1))} \\
&= \sqrt{(1 - \cos(2\pi (t_2-t_1)))^2 + \sin^2(2\pi (t_2-t_1))} \\
&= \sqrt{2 - 2\cos(2\pi (t_2-t_1))} \\ 
&= 2\abs{\sin (\pi (t_2-t_1))} \\
&= 2\sin (\pi \abs{t_2-t_1}) \qquad (\text{Since } t_2-t_1 \in (-1,1)) \\
&= 2\sin (\pi - \pi \abs{t_2-t_1}) = 2\sin (\pi d_w(t_1,t_2)).
\end{align*}
Since $\sin x \leq x$ for $x \geq 0$,  therefore we obtain $\abs{\alpha_1 - \alpha_2} \leq 2 \pi d_w(t_1,t_2)$. 
Also, since $d_w(t_1,t_2) \in [0,1/2]$ and $\sin x \geq \frac{2x}{\pi}$ for $x \in [0,\pi/2]$, we obtain 
$d_w(t_1,t_2) \leq \abs{\alpha_1 - \alpha_2}/4$. To summarize, 
\begin{equation} \label{eq:chord_wrap_temp1}
\frac{\abs{\alpha_1 - \alpha_2}}{2\pi} \leq d_w(t_1,t_2) \leq \frac{\abs{\alpha_1 - \alpha_2}}{4}.
\end{equation}
Finally, from the definition of chordal metric (see Definition \ref{def:chord_metric}) we know that 
$\chord(\alpha_1,\alpha_2) = \abs{\alpha_1 - \alpha_2}/2$. Thus from \eqref{eq:chord_wrap_temp1}, we obtain 
\begin{align} \label{eq:chord_wrap_temp2}
\frac{\chord(\alpha_1,\alpha_2)}{\pi} \leq d_w(t_1,t_2) \leq \frac{\chord(\alpha_1,\alpha_2)}{2}.
\end{align}
\subsection{Useful relation involving wrap around metric and chordal metric when one point lies on 
the unit disk, and the other does not} \label{app:rel_dw_chord_notdisk}
We will prove the following useful Proposition.
\begin{proposition} \label{app:prop_wrap_chord_notdisk}
For $\alpha_1,\alpha_2 \in \mathbb{C}$ where $\abs{\alpha_1} = 1$, let $t_1,t_2 \in [0,1)$ 
be such that $\alpha_1 = \exp(\iota 2\pi t_1)$, $\alpha_2/\abs{\alpha_2} = \exp(\iota 2\pi t_2)$. 
If $\chord(\alpha_1,\alpha_2) \leq \epsilon$ for some $0 \leq \epsilon \leq 1/4$, then $d_w(t_1,t_2) \leq (20 + \frac{1}{\sqrt{2}})\epsilon$. 
\end{proposition}
\begin{proof}
We begin by representing $\alpha_1,\alpha_2$ in Cartesian coordinates in $\mathbb{R}^2$ where
\begin{align*}
	\alpha_1  = (a,b), \quad	\alpha_2 = (x,y).
\end{align*}
%
Let $s(\alpha_1), s(\alpha_2) \in \mathbb{S}^2$ be such that 
$\alpha_1, \alpha_2$ are their respective stereographic projections. Then, 
\begin{align*}
	s(\alpha_2) & = \left(\frac{2x}{1+x^2+y^2},\frac{2y}{1+x^2+y^2}, \frac{-1+x^2+y^2}{1+x^2+y^2} \right) \\
	s(\alpha_1) & = (a,b,0).
\end{align*}
Since $\chi(\alpha_1,\alpha_2) = \frac{\norm{s(\alpha_1) - s(\alpha_2)}_2}{2} \le \epsilon$, this implies 
\begin{align} \label{app:aux_temp2}
	\frac{\vert -1+x^2+y^2\vert}{1+x^2+y^2}  \le 2\epsilon, \quad 
	\left\vert \frac{2x}{1+x^2+y^2} -a\right\vert \le 2\epsilon, \quad  
	\left\vert \frac{2y}{1+x^2+y^2} -b\right\vert \le 2\epsilon.
\end{align}
Let us assume $\epsilon < 1/2$ from now.
\begin{itemize}
\item The first inequality in \eqref{app:aux_temp2} implies  
\begin{align} \label{app:aux_temp5}
	\frac{1-2\epsilon}{1+2\epsilon} & \le x^2+y^2 \le \frac{1+2\epsilon}{1-2\epsilon}.
\end{align}

\item The second inequality in \eqref{app:aux_temp2} implies  
\begin{align} \label{app:aux_temp6}
	(a-2\epsilon) \left(\frac{1+x^2+y^2}{2} \right) & \le 	x\le (a+2\epsilon) \left(\frac{1+x^2+y^2}{2}\right).
\end{align}
		
\item  The third inequality in \eqref{app:aux_temp2} implies 
\begin{align} \label{app:aux_temp7}
		(b-2\epsilon) \left(\frac{1+x^2+y^2}{2} \right) & \le y\le (b+2\epsilon) \left(\frac{1+x^2+y^2}{2}\right).
\end{align}
\end{itemize}
Now from \eqref{app:aux_temp5},\eqref{app:aux_temp6} we get
\begin{align}
\frac{x}{\sqrt{x^2+y^2}} -a 
&\le \frac{(a+2\epsilon)\sqrt{1+2\epsilon}}{(1-2\epsilon)^{3/2}} -a \nonumber \\
&= \frac{a\sqrt{1+2\epsilon}}{(1-2\epsilon)^{3/2}} - a + \frac{2\epsilon\sqrt{1+2\epsilon}}{(1-2\epsilon)^{3/2}} \nonumber \\
&\leq \left(\frac{\sqrt{1+2\epsilon}}{(1-2\epsilon)^{3/2}} - 1\right) + \frac{2\epsilon\sqrt{1+2\epsilon}}{(1-2\epsilon)^{3/2}} \label{app:aux_temp8}
\end{align}
where we used $\abs{a} \leq 1$. Note that $\sqrt{1+2\epsilon} \leq 1 + \epsilon$. Moreover, we have 
the following 
\begin{claim}
For $0 \leq \epsilon \leq 1/4$, we have $(1-2\epsilon)^{-3/2} \leq (1+12\sqrt{2}\epsilon)$.
\end{claim}
\begin{proof}
Consider the function $g(\epsilon) = (1+12\sqrt{2}\epsilon) - (1-2\epsilon)^{-3/2}$. We have 
$g'(\epsilon) = 12\sqrt{2} - 3(1-2\epsilon)^{-5/2}$ and so $g'(\epsilon) \geq 0$ 
when $0 \leq \epsilon \leq 1/4$. Thus $g$ is increasing for this range of $\epsilon$, 
and so $g(\epsilon) \geq g(0) = 0$, which completes the proof.
\end{proof}
Applying these observations to \eqref{app:aux_temp8}, we get 
\begin{align}
\frac{x}{\sqrt{x^2+y^2}} -a  
&\leq [(1+\epsilon)(1+12\sqrt{2}\epsilon) - 1] + 2\epsilon(1+\epsilon)2^{3/2} \nonumber \\
&= [(12\sqrt{2} + 1)\epsilon + 12\sqrt{2}\epsilon^2] + 4\sqrt{2}(\epsilon + \epsilon^2) \nonumber \\
&\leq [(12\sqrt{2} + 1)\epsilon + 3\sqrt{2}\epsilon] + 5\sqrt{2}\epsilon \qquad (\text{Using } \epsilon \leq 1/4) \nonumber \\
&\leq (20\sqrt{2} + 1)\epsilon. \label{app:aux_temp9}
\end{align}
The reader is invited to verify that $a - \frac{x}{\sqrt{x^2+y^2}} \leq (20\sqrt{2} + 1)\epsilon$, which together 
with \eqref{app:aux_temp9} implies that $\abs{\frac{x}{\sqrt{x^2+y^2}} -a} \leq (20\sqrt{2} + 1)\epsilon$. In an identical 
fashion, one obtains the same bound on $\abs{\frac{y}{\sqrt{x^2+y^2}} - b}$. From these observations, we then obtain 
%
%
%
%
\begin{align*}
\chord(\alpha_1,\alpha_2/\abs{\alpha_2}) 
= \frac{1}{2}\sqrt{\left(\frac{x}{\sqrt{x^2+y^2}} -a \right)^2 + \left(\frac{y}{\sqrt{x^2+y^2}} - b\right)^2} 
\leq \left(20 + \frac{1}{\sqrt{2}}\right)\epsilon.
\end{align*}
Using \eqref{eq:chord_wrap_temp2}, this in turn implies that 
$d_w(\hemant{t_1,t_2}) \leq (10 + \frac{1}{2\sqrt{2}}) \epsilon$.
\end{proof}

\subsection{More useful results} \label{subsec:app_another_use_res}
\begin{proposition} \label{app:prop_useful_res}
For any $t_1,t_2 \in [0,1)$, and integer $n$, we have
\begin{align*}
\abs{\exp(\iota 2\pi n t_1) - \exp(\iota 2\pi n t_2)} \leq 2\abs{n}\pi d_w(t_1,t_2).
\end{align*}
\end{proposition}
\begin{proof}
On one hand, 
\begin{align}
\abs{\exp(\iota 2\pi n t_1) - \exp(\iota 2\pi n t_2)} 
&= \abs{1 - \exp(\iota 2\pi n(t_2-t_1))} \nonumber \\
&= 2\abs{\sin (\pi n  (t_2-t_1))} \nonumber \\
&\leq 2\pi \abs{n} \abs{t_2 - t_1} \quad (\text{Since } \abs{\sin x} \leq \abs{x}, \ \forall x\in \mathbb{R}). \label{eq:app_ano_useres_temp1}
\end{align}
On the other hand, note that 
\begin{align} \label{eq:app_ano_useres_temp2}
2\abs{\sin (\pi n  (t_2-t_1))} 
= 2\abs{\sin (\pi \abs{n}  \abs{t_2-t_1})} 
= 2\abs{\sin (\pi \abs{n} (1 - \abs{t_2-t_1}))} 
\leq 2 \pi \abs{n} (1 - \abs{t_2 - t_1}).
\end{align}
The bound follows from \eqref{eq:app_ano_useres_temp1}, \eqref{eq:app_ano_useres_temp2} and by noting the definition of $d_w$.
\end{proof}
%
%
\begin{proposition} \label{app:prop_useful_res_2}
For any $t_1,t_2 \in [0,1)$, integer $n$, and $u_1,u_2 \in \mathbb{C}$, we have
\begin{align*}
\abs{u_1\exp(\iota 2\pi n t_1) - u_2 \exp(\iota 2\pi n t_2)} \leq 2\pi \abs{u_1} \abs{n} d_w(t_1,t_2) + \abs{u_1-u_2}.
\end{align*}
\end{proposition}
\begin{proof}
Indeed, 
\begin{align*}
&\abs{u_1\exp(\iota 2\pi n t_1) - u_2 \exp(\iota 2\pi n t_2)} \\
&= \abs{u_1\exp(\iota 2\pi n t_1) - u_1\exp(\iota 2\pi n t_2) + u_1\exp(\iota 2\pi n t_2) - u_2 \exp(\iota 2\pi n t_2)} \\
&\leq \abs{u_1} \abs{\exp(\iota 2\pi n t_1) - \exp(\iota 2\pi n t_2)} + \abs{u_1-u_2} \\
&\leq 2\pi \abs{u_1} \abs{n} d_w(t_1,t_2) + \abs{u_1-u_2}
\end{align*}
where in the last inequality, we used Proposition \ref{app:prop_useful_res}.
\end{proof}
\section{Conditions on $\varepsilon_l$ in Theorem \ref{thm:gen_case_main}} \label{app:sec:conditions_eps_gen}
We treat conditions \ref{itm:main_thm_epscond_2}, \ref{itm:main_thm_epscond_3} separately below. In this section, 
for a given $i \in \mathbb{N}$, we will denote $C_i(u_i,v_i,w_i,\dots)$ to be a positive term 
depending only on the parameters $u_i,v_i,w_i,\dots$.
\begin{enumerate}
\item Condition \ref{itm:main_thm_epscond_2} is equivalent to
\begin{align}
\frac{K}{\sep_{L-1}}\varepsilon_{L-1} + \frac{1}{(\kerfnvar_L^{2} - \kerfnvar_{L-1}^2)^{1/2}} 
\log^{1/2}\left(\frac{K^{3/2}\kerfnvar_{L}}{\varepsilon_{L-1}\kerfnvar_{L-1}} \right)\varepsilon_{L-1} 
&\lesssim \frac{\varepsilon_{L} \kerfnvar_L e^{-2\pi^2(\kerfnvar_L^2 - \kerfnvar_1^2)/\sep_L^2}}{K^{3/2}L\kerfnvar_{L-1}} \label{eq:app_epsconds1b_1}\\
\Leftrightarrow C_1(K,\sep_{L-1},\kerfnvar_L,\kerfnvar_{L-1}) \log^{1/2}\left(\frac{K^{3/2}\kerfnvar_{L}}{\varepsilon_{L-1}\kerfnvar_{L-1}} \right)\varepsilon_{L-1} 
&\lesssim C_2(\kerfnvar_{L},\kerfnvar_{L-1},\kerfnvar_{1},K,L,\sep_L) \varepsilon_L. \label{eq:app_epsconds1b_2}
\end{align}
Since for any $0 < \widetilde\theta < 1$, we have\footnote{Here, we use the fact $\log x \leq n x^{1/n}$ for $n,x > 0$.} 
\begin{equation} \label{eq:app_epsconds1b_3}
\log\left(\frac{K^{3/2}\kerfnvar_{L}}{\varepsilon_{L-1}\kerfnvar_{L-1}}\right) 
\leq 
\frac{1}{\widetilde\theta}\left(\frac{K^{3/2}\kerfnvar_{L}}{\varepsilon_{L-1}\kerfnvar_{L-1}}\right)^{\widetilde\theta},
\end{equation}
therefore \eqref{eq:app_epsconds1b_2} is satisfied if for any given $\theta \in (0,1/2)$, it holds that
\begin{equation} \label{eq:app_epsconds1b_4}
\varepsilon_{L-1}^{1-\theta} \lesssim \sqrt{\theta} \frac{C_2(\kerfnvar_{L},\kerfnvar_{L-1},\kerfnvar_{1},K,L,\sep_L)}{C_1(K,\sep_{L-1},\kerfnvar_L,\kerfnvar_{L-1})} \left(\frac{\kerfnvar_{L-1}}{\kerfnvar_L  K^{3/2}}\right)^{\theta} \varepsilon_{L}.
\end{equation}
From this, \eqref{eq:eps_main_alpha} follows easily.

\item We now look at condition \ref{itm:main_thm_epscond_3} starting with the condition 
$E_l(\varepsilon_l) \varepsilon_l \leq E_{l+1}(\varepsilon_{l+1}) \varepsilon_{l+1}$. Using the order dependency of $E_{l}(\cdot)$ 
from \eqref{eq:E_L_ord_dep}, this is the same as
\begin{align} \label{eq:app_epsconds1c_1}
& C_3(K,\sep_l,\kerfnvar_l,\kerfnvar_{l+1}) \log^{1/2}\left(\frac{K^{3/2}\kerfnvar_{L} (L-l)}{\varepsilon_{l}\kerfnvar_{l}}\right) \varepsilon_l \\
& \hspace{1cm}\lesssim C_4(K,\sep_{l+1},\kerfnvar_{l+1},\kerfnvar_{l+2}) \log^{1/2}\left(\frac{K^{3/2}\kerfnvar_{L} (L-l)}{\varepsilon_{l+1}\kerfnvar_{l+1}}\right) \varepsilon_{l+1}.
\end{align}
Using \eqref{eq:app_epsconds1b_3}, it follows that  $E_l(\varepsilon_l) \varepsilon_l \leq E_{l+1}(\varepsilon_{l+1}) \varepsilon_{l+1}$ 
is ensured if 
\begin{equation}\label{eq:app_epsconds1c_2}
\varepsilon_{l}^{1-\theta} \lesssim 
\sqrt{\theta}\frac{C_4(K,\sep_{l+1},\kerfnvar_{l+1},\kerfnvar_{l+2})}{C_3(K,\sep_l,\kerfnvar_l,\kerfnvar_{l+1})} \left(\frac{\kerfnvar_{l}}{\kerfnvar_L (L-l) K^{3/2}}\right)^{\theta} 
\log^{1/2}\left(\frac{K^{3/2}\kerfnvar_{L} (L-l)}{\varepsilon_{l+1}\kerfnvar_{l+1}}\right) \varepsilon_{l+1}
\end{equation} 
holds for any given $\theta \in (0,1/2)$.

Now consider the other condition, namely \eqref{eq:app_epsconds1c_3}.
%
%
%
%
Using the order dependency of $F_{l+1}(\varepsilon_l)$ and $E_l(\varepsilon_l)$ from 
\eqref{eq:orddep_F_gen},\eqref{eq:E_L_ord_dep} respectively, this can be equivalently written as 
\begin{align} \label{eq:app_epsconds1c_4}
&\left(\frac{1}{\sep_{l+1}} + \frac{1}{(\kerfnvar_{l+2}^2 -\kerfnvar_{l+1}^2)^{1/2}} 
\log^{1/2}\left( \frac{K^{3/2} (L-l) \kerfnvar_L}{\varepsilon_l \kerfnvar_{l+1}}\right) \right)\varepsilon_l  \\
&\quad +\left(\frac{K}{\sep_l} + \frac{1}{(\kerfnvar_{l+1}^2 -\kerfnvar_{l}^2)^{1/2}} \log^{1/2}\left( \frac{K^{3/2} (L-l) \kerfnvar_L}{\varepsilon_l \kerfnvar_{l}}\right) \right) \varepsilon_l
\lesssim \varepsilon_{l+1} e^{-(\kerfnvar_{l+1}^2 -\kerfnvar_{1}^2)\frac{F^2_{l+1}(\varepsilon_{l+1})}{2}} 
\frac{\kerfnvar_{l+1}}{\kerfnvar_l l K^{3/2}}.\nonumber
\end{align}
Since $\kerfnvar_{l+1} > \kerfnvar_l$, therefore the L.H.S of \eqref{eq:app_epsconds1c_4} is  
\begin{equation} \label{eq:app_epsconds1c_5}
\lesssim C_5(\sep_l,\sep_{l+1},\kerfnvar_l,\kerfnvar_{l+1},\kerfnvar_{l+2},K) 
\log^{1/2}\left(\frac{K^{3/2} (L-l) \kerfnvar_L}{\varepsilon_l \kerfnvar_{l}}\right) \varepsilon_l.
\end{equation}
Moreover, since 
\begin{equation*}
   (\kerfnvar_{l+1}^2 -\kerfnvar_{1}^2) \frac{F^2_{l+1}(\varepsilon_{l+1})}{2} 
	\lesssim C_6(\kerfnvar_1,\kerfnvar_{l+1}, \kerfnvar_{l+2}, \sep_{l+1}) 
	\log\left(\frac{K^{3/2} (L-l) \kerfnvar_L}{\varepsilon_{l+1} \kerfnvar_{l+1}}\right),
\end{equation*}
therefore the R.H.S of \eqref{eq:app_epsconds1c_4} is 
\begin{align} 
&\gtrsim \varepsilon_{l+1} \left(\frac{\kerfnvar_{l+1}}{\kerfnvar_{l} l K^{3/2}}\right) 
\left(\frac{\varepsilon_{l+1}\kerfnvar_{l+1}}{K^{3/2} (L-l) \kerfnvar_L}\right)^{C_6(\kerfnvar_1,\kerfnvar_{l+1}, \kerfnvar_{l+2}, \sep_{l+1})} \nonumber \\
&\gtrsim (\varepsilon_{l+1})^{1+C_6(\kerfnvar_1,\kerfnvar_{l+1}, \kerfnvar_{l+2}, \sep_{l+1})} 
C_7(\kerfnvar_l,\kerfnvar_{l+1},\kerfnvar_L,K,L). \label{eq:app_epsconds1c_6}
\end{align}
Therefore from \eqref{eq:app_epsconds1c_5},\eqref{eq:app_epsconds1c_6}, it follows that a sufficient condition for 
\eqref{eq:app_epsconds1c_4} to hold is
\begin{equation} \label{eq:app_epsconds1c_7}
\log^{1/2}\left(\frac{K^{3/2} (L-l) \kerfnvar_L}{\varepsilon_l \kerfnvar_{l}}\right) \varepsilon_l 
\lesssim (\varepsilon_{l+1})^{1+C_6(\kerfnvar_1,\kerfnvar_{l+1}, \kerfnvar_{l+2}, \sep_{l+1})} 
\frac{C_7(\kerfnvar_l,\kerfnvar_{l+1},\kerfnvar_L,K,L)}{C_5(\sep_l,\sep_{l+1},\kerfnvar_l,\kerfnvar_{l+1},\kerfnvar_{l+2},K)}. 
\end{equation}
Using \eqref{eq:app_epsconds1b_3}, we have that \eqref{eq:app_epsconds1c_7} holds if for any given $\theta \in (0,1/2)$, 
\begin{equation} \label{eq:app_epsconds1c_8}
\varepsilon_l^{1-\theta} 
\lesssim 
\sqrt{\theta} 
\frac{C_7(\kerfnvar_l,\kerfnvar_{l+1},\kerfnvar_L,K,L)}{C_5(\sep_l,\sep_{l+1},\kerfnvar_l,\kerfnvar_{l+1},\kerfnvar_{l+2},K)} 
(\varepsilon_{l+1})^{1+C_6(\kerfnvar_1,\kerfnvar_{l+1}, \kerfnvar_{l+2}, \sep_{l+1})} 
\left(\frac{\kerfnvar_l}{\kerfnvar_L (L-l) K^{3/2}}\right)^{\theta}.
\end{equation}
Comparing \eqref{eq:app_epsconds1c_8} and \eqref{eq:app_epsconds1c_2}, we see that the dependence on $\varepsilon_{l+1}$ 
in \eqref{eq:app_epsconds1c_8} is in general stricter than that in \eqref{eq:app_epsconds1c_2}. 
\end{enumerate} 
%
\section{Proof of Theorem \ref{thm:gen_case_main_noisy}} \label{app:sec:gen_unmix_noisy_thm}
As for the proof of Theorem \ref{thm:gen_case_main}, the proof of Theorem \ref{thm:gen_case_main_noisy} is divided into three main steps. We will only point out the relevant modifications in the proof of Theorem \ref{thm:gen_case_main} that we need to make in order to account for the noise term.
\begin{itemize}
\item \textbf{Recovering source parameters for first group.} 
For $i=-m_1,\ldots,m_1-1$, equation \eqref{eq:gen_proof_temp1} becomes 
\begin{align*}
\frac{\widetilde f(s_1+i)}{\kerfour_1(s_1+i)} 
& = \sum_{j=1}^K u'_{1,j} \exp\left(\iota 2\pi i t_{1,j}\right) + \noiseint_{1,i}+\frac{w_1(i)}{\kerfour_1(s_1+i)}.
\end{align*}
Since the noise term satisfies \eqref{noise_bnd},
invoking Corollary \ref{corr:moitra_MP} brings that for $\varepsilon_1 < c \sep_1/2$, and 
$\frac{2}{\sep_1 - 2\varepsilon_1} + 1 \leq m_1 < \frac{2}{\sep_1(1-c)} + 1 (= \mup{1})$, we only need to impose that $s_1$ satisfies 
\begin{align*} 
K\umax (L-1) \frac{\kerfnvar_{L}}{\kerfnvar_1}\exp\left(-2\pi^2(s_1-m_1)^2(\kerfnvar_2^2-\kerfnvar_1^2)\right)
\leq \hemant{\varepsilon_1 \frac{\umin B(\urel,K)}{2}}
\end{align*}
to recover the same error bound \eqref{eq:gen_proof_temp3} as in the noiseless case.
The lower and upper bounds on $s_1$ are then obtained in a similar way as \eqref{eq:low_s1} and \eqref{eq:gen_proof_temp4} in the proof of Theorem 
\ref{thm:gen_case_main}, and one finally gets the same error bound \eqref{eq:gen_proof_temp5} as in the noiseless case.
%
%

\item \textbf{Recovering source parameters for $l^{th} (1 < l < L)$ group.}
Say we are at the $l^{th}$ iteration for $1 < l < L$, having estimated the source parameters 
up to the $(l-1)^{th}$ group. 
Say that for each $p = 1,\dots,l-1$ and $j=1,\dots,K$ the following holds.
\begin{align} \label{eq:gen_proof_temp10bis}
d_w(\est{t}_{p,\phi_p(j)}, t_{p,j}) \leq \varepsilon_p, \quad  \abs{\est{u}_{p,\phi_p(j)} - u_{p,j}} < E_p(\varepsilon_p) \hemant{\varepsilon_p \umax}, 
\end{align}
for some permutations $\phi_p: [K] \rightarrow [K]$, with 
\begin{enumerate} 
\item $\varepsilon_1 \leq \cdots \leq \varepsilon_{l-1}; \quad \quad E_1(\varepsilon_1) \hemant{\varepsilon_1} \leq \cdots \leq E_{l-1}(\varepsilon_{l-1}) \hemant{\varepsilon_{l-1}}$; 

\item $\varepsilon_{p} < c \sep_p/2$ ; 

\item \hemant{$(F_{q+1}'(\varepsilon_{q}) + E_{q}(\varepsilon_{q})) \hemant{\varepsilon_{q} \umax} 
\leq \varepsilon_{q+1} e^{-(\kerfnvar_{q+1}^2-\kerfnvar_1^2)\frac{{F_{q+1}'^2}(\varepsilon_{q+1})}{2}} \frac{\kerfnvar_{q+1} B(\urel,K)}{3 K q \kerfnvar_{q}}$},  $1 \leq q \leq l-2$.
\end{enumerate}
For $i=-m_l,\ldots,m_l-1$, we have  
\begin{align*}
&\frac{\widetilde f(s_l + i) - \sum_{p=1}^{l-1} \est{f}_p(s_l + i)}{\kerfour_l(s_l+i)} = \sum_{j=1}^K u'_{l,j} \exp(\iota2\pi i t_{l,j}) 
+ \noiseint_{l,i,past} + \noiseint_{l,i,fut}+ \frac{w_l(i)}{\kerfour_l(s_l+i)}.
\end{align*}
As in the proof of Theorem \ref{thm:gen_case_main}, $\noiseint_{l,i,past}$ denotes perturbation due to the estimation errors of 
the source parameters in the past, and $\noiseint_{l,i,fut}$ denotes perturbation due to the tails of the kernels that are yet 
to be processed. Imposing that $s_l \geq S_l$ where
\begin{align*}
S_l = m_l + \frac{1}{(2\pi^2(\kerfnvar_{l+1}^2 - \kerfnvar_l^2))^{1/2}} 
\log^{1/2}\left(\hemant{\frac{3 K \urel (L-l) \kerfnvar_L}{\varepsilon_l \kerfnvar_l B(\urel,K)}} \right)
\end{align*}
we get that, for $i=-m_l,\dots,m_l-1$, 
\begin{align} \label{eq:gen_proof_temp12bis}
\abs{\noiseint_{l,i,fut}} < \varepsilon_l \hemant{\frac{\umin B(\urel,K)}{3}}
\end{align}
where we used the same steps as in the proof of Theorem \ref{thm:gen_case_main}.
%
Now, since $s_l \leq \widetilde{c} S_l$ and $m_l < \mup{l}$, we obtain the bound
%
%
%
%
%
\begin{align*}
2\pi(s_l + m_l) 
&< 2\pi(\widetilde{c} + 1) \mup{l} + \frac{2\pi \widetilde{c}}{(2\pi^2(\kerfnvar_{l+1}^2 - \kerfnvar_l^2))^{1/2}} 
\log^{1/2}\left(\hemant{\frac{3 K \urel (L-l) \kerfnvar_L}{\varepsilon_l \kerfnvar_l B(\urel,K)}}\right)  \\
&= C'_{l,1} + C'_{l,2} \log^{1/2}\left(\frac{3D_{l}}{\varepsilon_l}\right) = F'_{l}(\varepsilon_l),
\end{align*}
where we recall the definition of $F'_l$, and constants $C_{l,1}', C_{l,2}', D_l > 0$ from \eqref{eq:gen_proof_temp28bis}, \eqref{eq:gen_proof_temp29}.  
Since $\varepsilon_l < D_l$, hence if $\varepsilon_l \geq \varepsilon_{l-1}$ holds, 
then $2\pi(s_l + m_l) < F'_{l}(\varepsilon_l) \leq F'_{l}(\varepsilon_{l-1})$. Using this in 
\eqref{eq:gen_proof_temp11}, we obtain
\begin{align*}
\abs{\noiseint_{l,i,past}} 
< \hemant{\left(K(l-1) \frac{\kerfnvar_{l-1}}{\kerfnvar_l} e^{(\kerfnvar_l^2 - \kerfnvar_1^2)\frac{F^{'^2}_l(\varepsilon_l)}{2}} \right) 
\left(F'_l(\varepsilon_{l-1})  + E_{l-1}(\varepsilon_{l-1})\right) \varepsilon_{l-1} \umax.} 
\end{align*}
Therefore if $\varepsilon_{l-1}$ satisfies the condition 
\begin{align*}
\hemant{(F'_l(\varepsilon_{l-1}) + E_{l-1}(\varepsilon_{l-1}))\varepsilon_{l-1} \umax 
\leq \varepsilon_l e^{-(\kerfnvar_l^2-\kerfnvar_1^2) \frac{F^{'^2}_l(\varepsilon_l)}{2}} \frac{\umin\kerfnvar_l B(\urel,K)}{3 K (l-1) \kerfnvar_{l-1}}} 
\end{align*}
then it implies 
$\abs{\noiseint_{l,i,past}} < \hemant{\varepsilon_l \frac{\umin B(\urel,K)}{3}}$. Together with \eqref{eq:gen_proof_temp12bis}, \eqref{noise_bnd1_tmp} this gives \hemant{for each $i$ that}  
\begin{align*}
\left|\noiseint_{l,i}+ \frac{w_l(i)}{\bar g_l(s_l+i)}\right| 
\leq \abs{\noiseint_{l,i,fut}} + \abs{\noiseint_{l,i,past}} + \left|\frac{w_l(i)}{\bar g_l(s_l+i)}\right|  
< \varepsilon_l \hemant{\umin B(\urel,K)}.
\end{align*}
\hemant{Thereafter, we obtain the same bound \eqref{eq:gen_proof_temp15} as in the noiseless case.}
%
%
%
%
%
%

\item \textbf{Recovering source parameters for last group.}
Say that for each $p = 1,\dots,L-1$ and $j=1,\dots,K$ the following holds.
\begin{align} \label{eq:gen_proof_temp18bis}
d_w(\est{t}_{p,\phi_p(j)}, t_{p,j}) \leq \varepsilon_p, \quad  \abs{\widehat u_{p,\phi_p(j)} - u_{p,j}} <  E_p(\varepsilon_p) \hemant{\varepsilon_p \umax},
\end{align}
for some permutations $\phi_p: [K] \rightarrow [K]$, with
\begin{enumerate} 
\item $\varepsilon_1 \leq \cdots \leq \varepsilon_{L-1}$; $E_1(\varepsilon_1) \hemant{\varepsilon_1} \leq \cdots \leq E_{L-1}(\varepsilon_{L-1}) \hemant{\varepsilon_{L-1}}$; \label{eq:gen_proof_temp19bis} 

\item $\varepsilon_{p} < c \sep_p/2$;  

\item \hemant{$(F'_{q+1}(\varepsilon_{q}) + E_{q}(\varepsilon_{q})) \hemant{\varepsilon_{q} \urel} 
\leq \varepsilon_{q+1} e^{-(\kerfnvar_{q+1}^2-\kerfnvar_1^2)\frac{F_{q+1}^{'^2}(\varepsilon_{q+1})}{2}} \frac{\kerfnvar_{q+1} B(\urel,K)}{3 K q \kerfnvar_{q}}$}, $1 \leq q \leq L-2$.
\end{enumerate}
We proceed by noting that for each $i = -m_L,\dots,m_L-1$
\begin{align*}
\frac{\widetilde f(s_L + i) - \sum_{p=1}^{L-1} \est{f}_p(s_L + i)}{\kerfour_L(s_L + i)} 
& = \sum_{j=1}^K u'_{L,j} \exp(\iota2\pi i t_{L,j}) + \noiseint_{L,i}+\frac{w_L(i)}{\bar g_L(s_L+i)}.
\end{align*}
%
%
Using Proposition \ref{app:prop_useful_res_2}, we have for each $p=1,\dots,L-1$ and $j=1,\dots,K$ that   
\begin{align}
\abs{u_{p,j} \exp(\iota 2\pi (s_L+i)t_{p,j}) - \est{u}_{p,\perm_p(j)} \exp (\iota 2\pi (s_L+i)\widehat t_{p,\perm_p(j)})} 
&< 2\pi \umax \abs{s_L+i} \varepsilon_p + E_p(\varepsilon_p) \hemant{\varepsilon_p \umax}, \label{eq:gen_proof_temp25bis}
\end{align}
where we used \eqref{eq:gen_proof_temp18bis}.
Since $s_L = 0$, hence $(s_L + i)^2 < m_L^2 < \mup{L}^2$ for all $i=-m_L,\dots,m_L-1$. 
Using \eqref{eq:gen_proof_temp25bis}, we then easily obtain as in the proof of Theorem \ref{thm:gen_case_main} 
\begin{align}
\abs{\noiseint_{L,i}} 
&\leq \left(K(L-1) \frac{\kerfnvar_{L-1}}{\kerfnvar_L} e^{2\pi^2(\kerfnvar_L^2 - \kerfnvar_1^2)\mup{L}^2} \right) 
\left(2\pi\mup{L} + E_{L-1}(\varepsilon_{L-1})\right) \hemant{\umax \varepsilon_{L-1}} 
\end{align}
where in the last inequality, we used \eqref{eq:gen_proof_temp19bis}.

Invoking Corollary \ref{corr:moitra_MP} and assuming $\varepsilon_L < c \sep_L/2$, 
it follows for the stated conditions on $m_L$, that it suffices if $\varepsilon_{L-1}$ satisfies 
\begin{align*}
\hemant{(2\pi\mup{L} + E_{L-1}(\varepsilon_{L-1})) \urel \varepsilon_{L-1} 
\leq 
\varepsilon_L e^{-2\pi^2(\kerfnvar_L^2 - \kerfnvar_1^2)\mup{L}^2} \frac{\kerfnvar_L B(\urel,K)}{2 K (L-1) \kerfnvar_{L-1}}.}
\end{align*}
\hemant{Indeed, combining this last bound with \eqref{noiseL}, we obtain for each $i$ that $\abs{\noiseint_{L,i} + \frac{w_L(i)}{\bar g_l(s_L+i)}} < \varepsilon_L \umin B(\urel,K)$. Thereafter, we obtain the same bound \eqref{eq:gen_proof_temp17} as in the noiseless case.}
\end{itemize} 
\clearpage
\section{More experiments} \label{sec:app_exps}
%
%
%
\begin{figure}[!ht]
\centering
\subcaptionbox[]{$l = 1, K = 2$}[ 0.24\textwidth ]
{\includegraphics[width=0.24\textwidth]{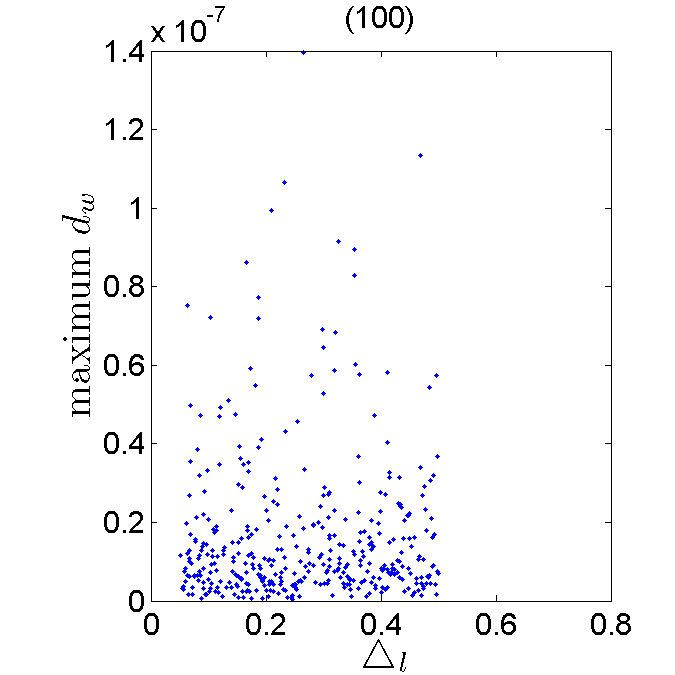} }
\subcaptionbox[]{$l = 2, K = 2$}[ 0.24\textwidth ]
{\includegraphics[width=0.24\textwidth]{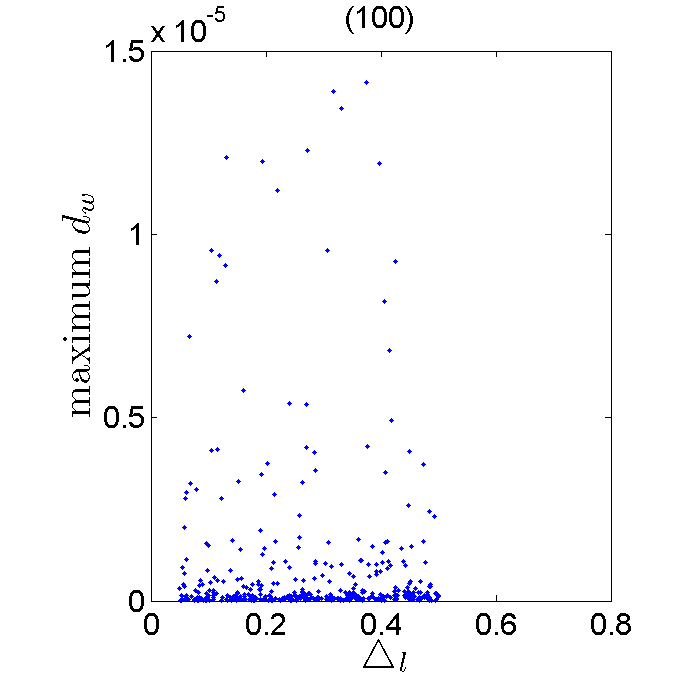} }
%
\subcaptionbox[]{$l = 3, K = 2$}[ 0.24\textwidth ]
{\includegraphics[width=0.24\textwidth]{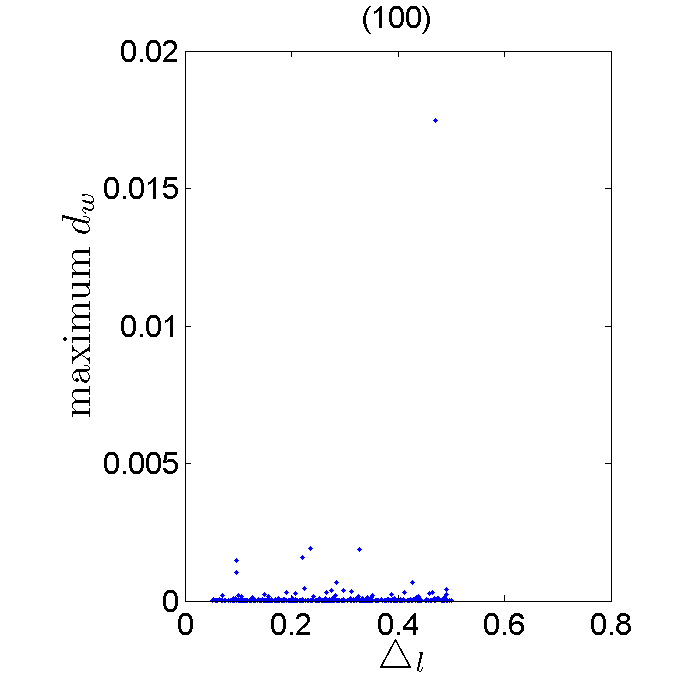} }
%
\subcaptionbox[]{$l = 4, K = 2$}[ 0.24\textwidth ]
{\includegraphics[width=0.24\textwidth]{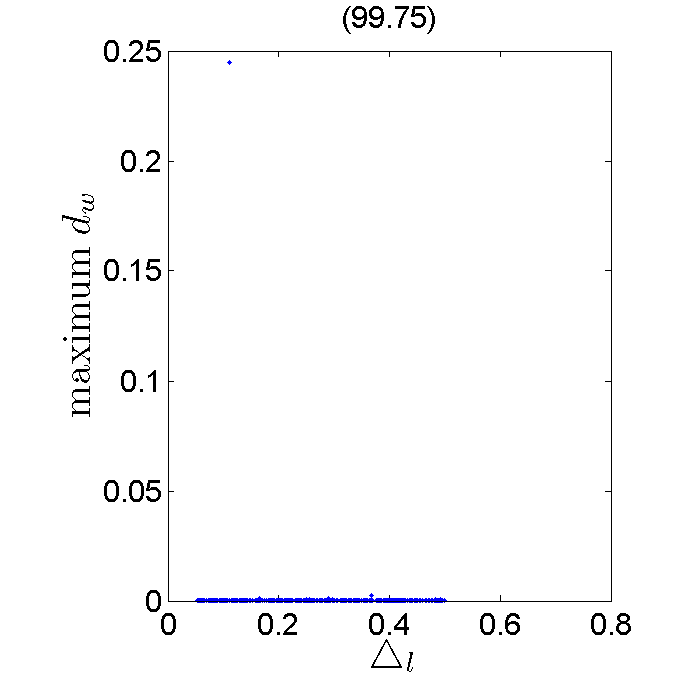} }
%

\subcaptionbox[]{$l = 1, K = 3$}[ 0.24\textwidth ]
{\includegraphics[width=0.24\textwidth]{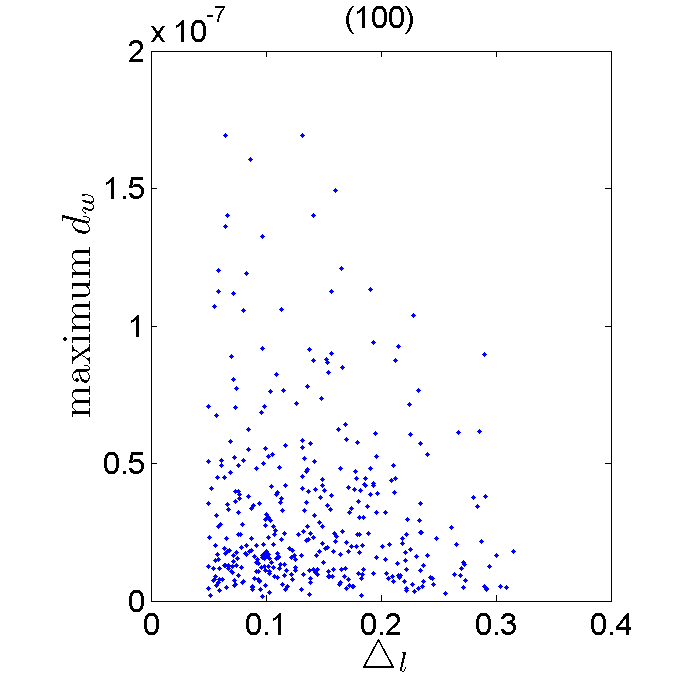} }
\subcaptionbox[]{$l = 2, K = 3$}[ 0.24\textwidth ]
{\includegraphics[width=0.24\textwidth]{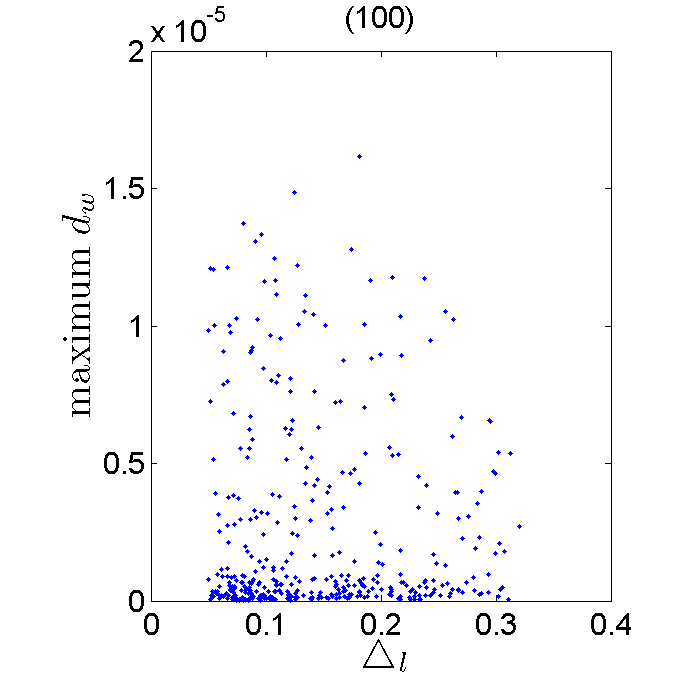} }
%
\subcaptionbox[]{$l = 3, K = 3$}[ 0.24\textwidth ]
{\includegraphics[width=0.24\textwidth]{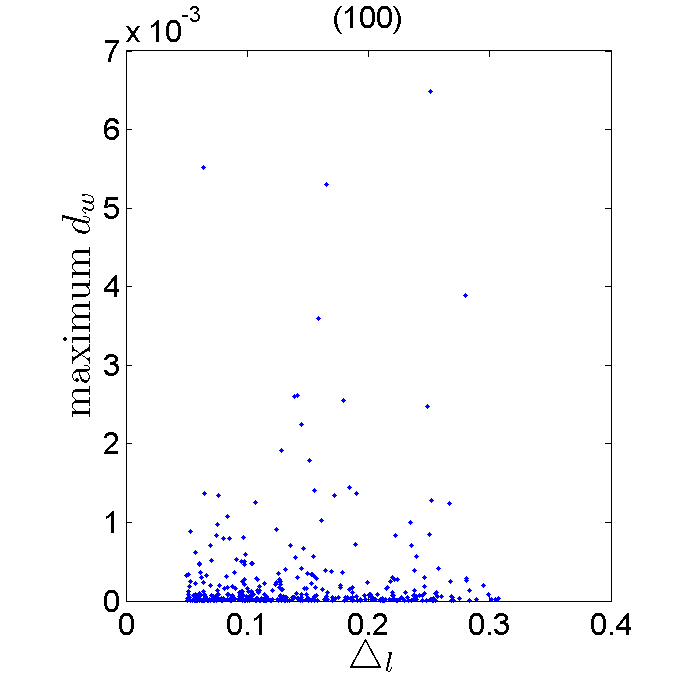} }
%
\subcaptionbox[]{$l = 4, K = 3$}[ 0.24\textwidth ]
{\includegraphics[width=0.24\textwidth]{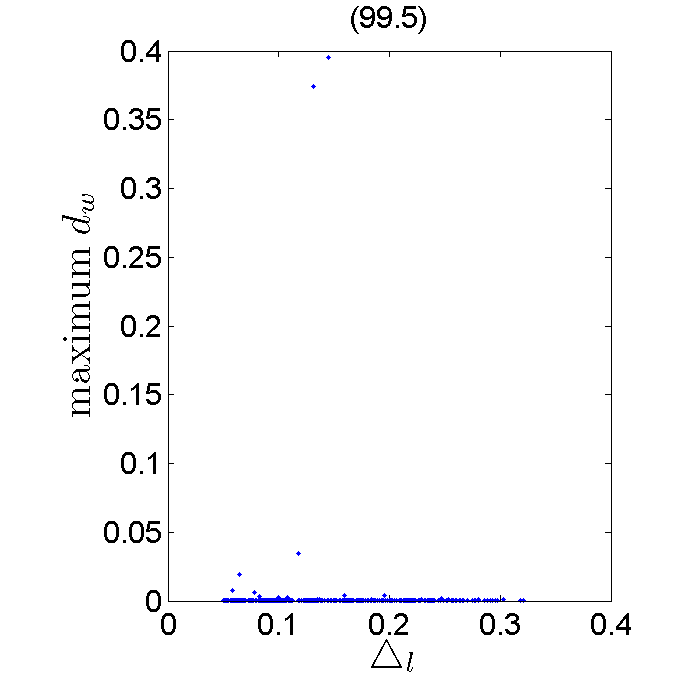} }
%

\subcaptionbox[]{$l = 1, K = 4$}[ 0.24\textwidth ]
{\includegraphics[width=0.24\textwidth]{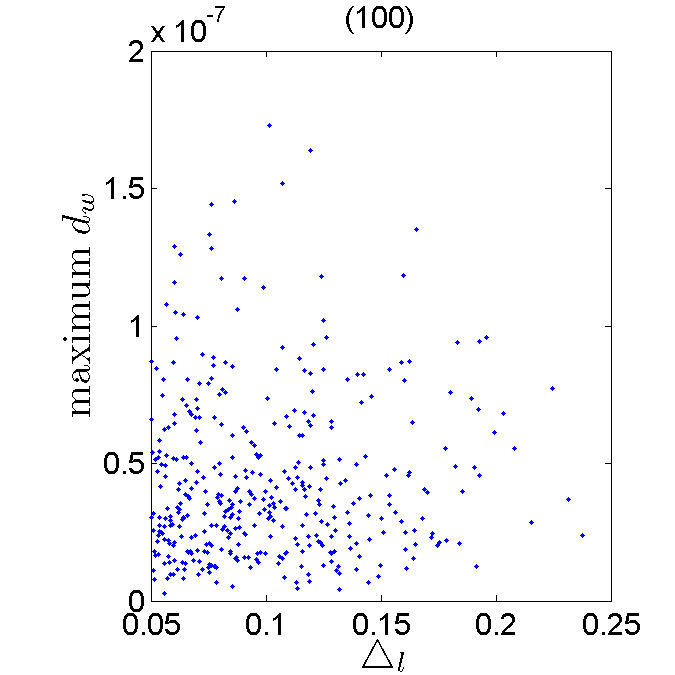} }
\subcaptionbox[]{$l = 2, K = 4$}[ 0.24\textwidth ]
{\includegraphics[width=0.24\textwidth]{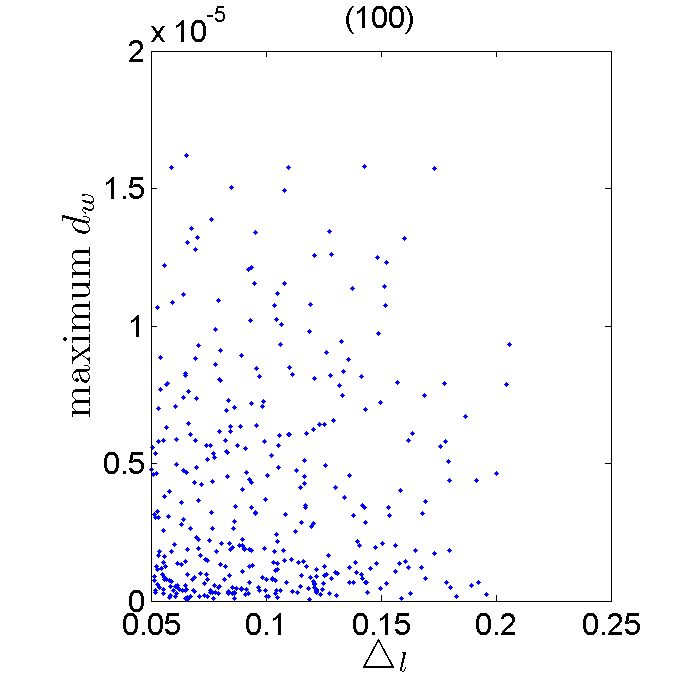} }
%
\subcaptionbox[]{$l = 3, K = 4$}[ 0.24\textwidth ]
{\includegraphics[width=0.24\textwidth]{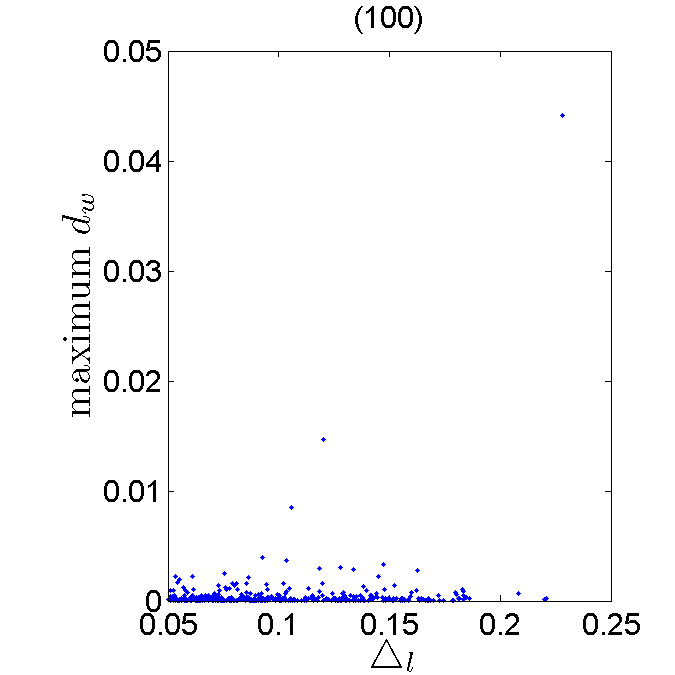} }
%
\subcaptionbox[]{$l = 4, K = 4$}[ 0.24\textwidth ]
{\includegraphics[width=0.24\textwidth]{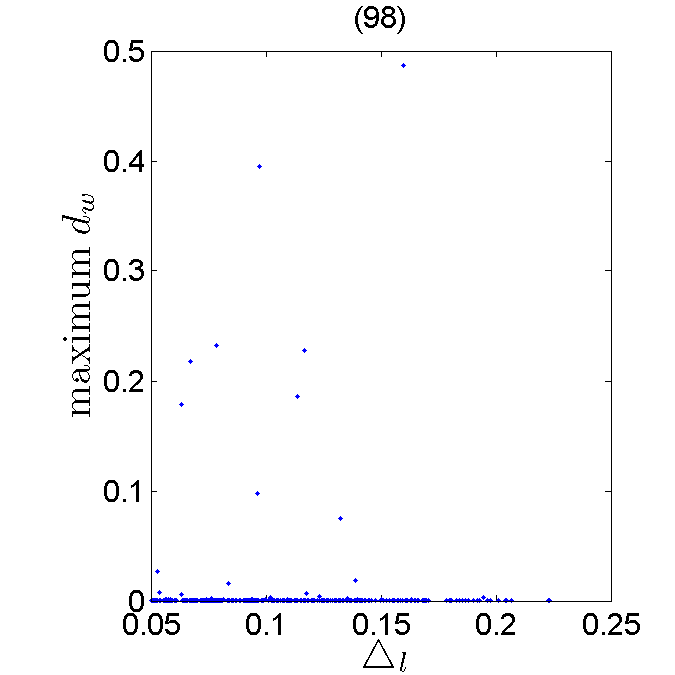} }
%

\subcaptionbox[]{$l = 1, K = 5$}[ 0.24\textwidth ]
{\includegraphics[width=0.24\textwidth]{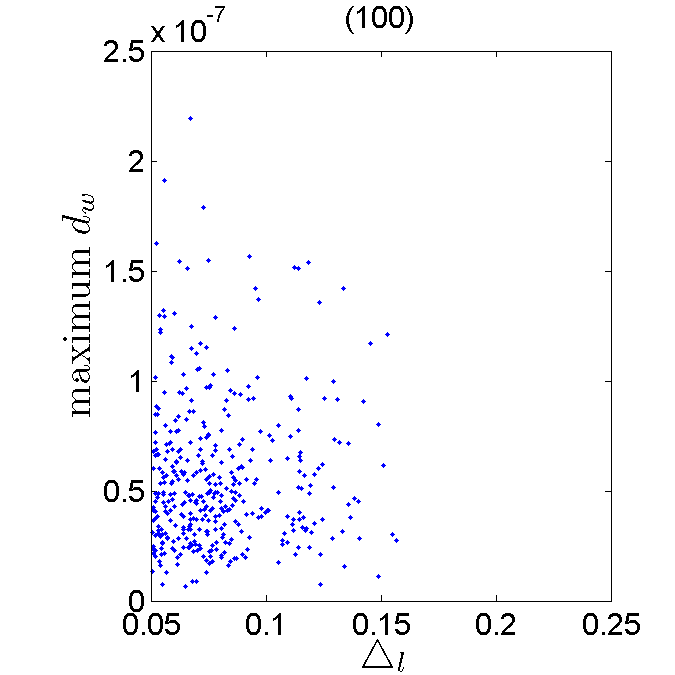} }
\subcaptionbox[]{$l = 2, K = 5$}[ 0.24\textwidth ]
{\includegraphics[width=0.24\textwidth]{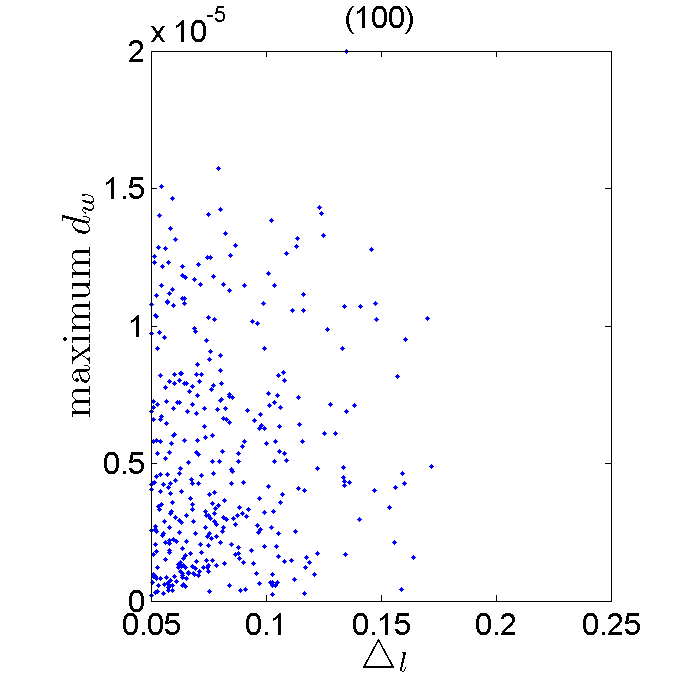} }
%
\subcaptionbox[]{$l = 3, K = 5$}[ 0.24\textwidth ]
{\includegraphics[width=0.24\textwidth]{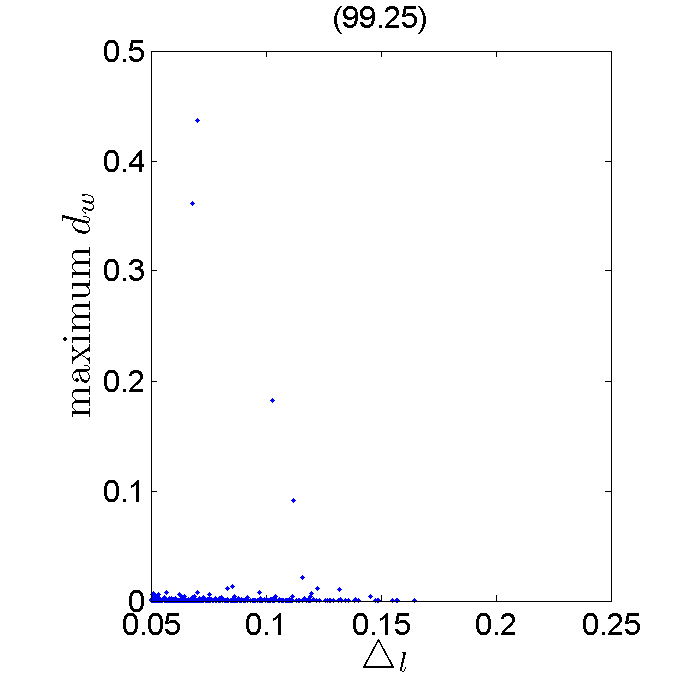} }
%
\subcaptionbox[]{$l = 4, K = 5$}[ 0.24\textwidth ]
{\includegraphics[width=0.24\textwidth]{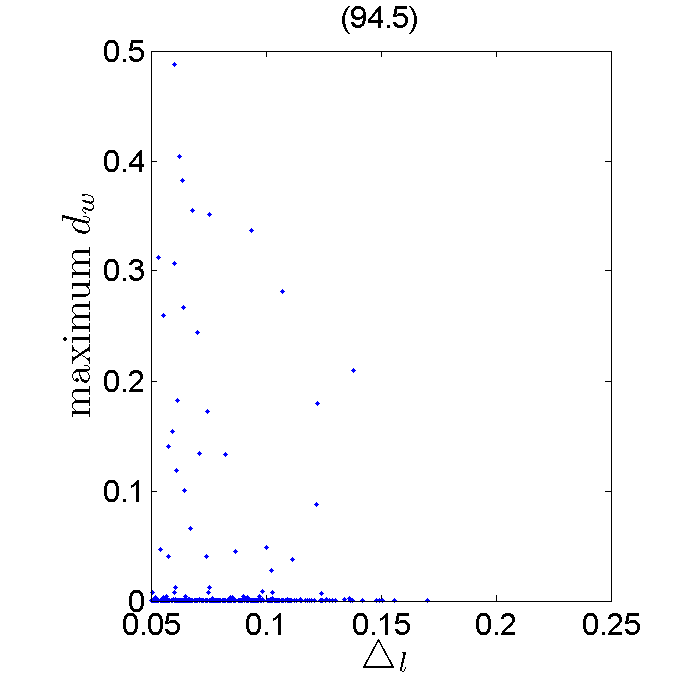} }
\captionsetup{width=0.98\linewidth}
\caption[Short Caption]{Scatter plots for maximum wrap around error ($d_{w,l,\max}$) v/s minimum separation ($\sep_l$) 
for $400$ Monte Carlo trials, with no external noise. This is shown for $K \in \set{2,3,4,5}$ with $L = 4$ 
and $C = 1$. For each sub-plot, we mention the percentage of trials with $d_{w,l,\max} \leq 0.05$ in parenthesis.}
\label{fig:max_loc_err_noiseless_C1}
\end{figure}

%
\begin{figure}[!ht]
\centering
\subcaptionbox[]{$l = 1, K = 2$}[ 0.24\textwidth ]
{\includegraphics[width=0.24\textwidth]{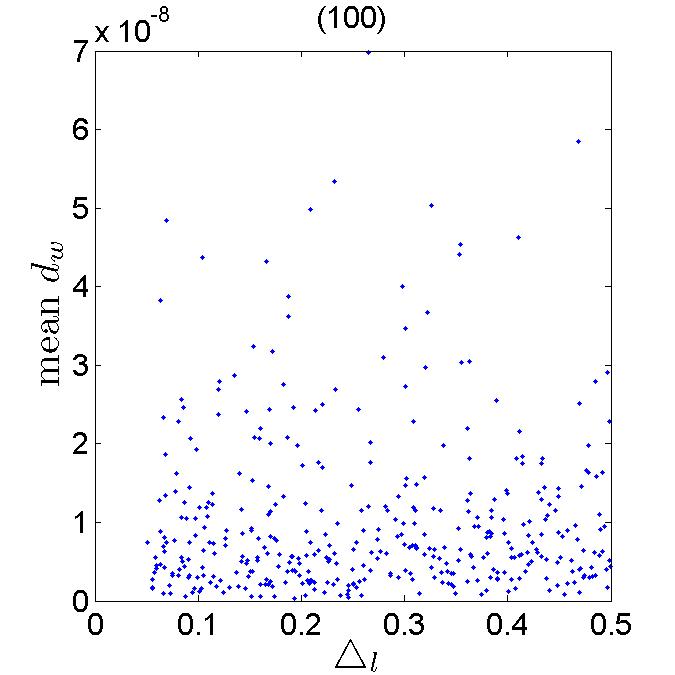} }
\subcaptionbox[]{$l = 2, K = 2$}[ 0.24\textwidth ]
{\includegraphics[width=0.24\textwidth]{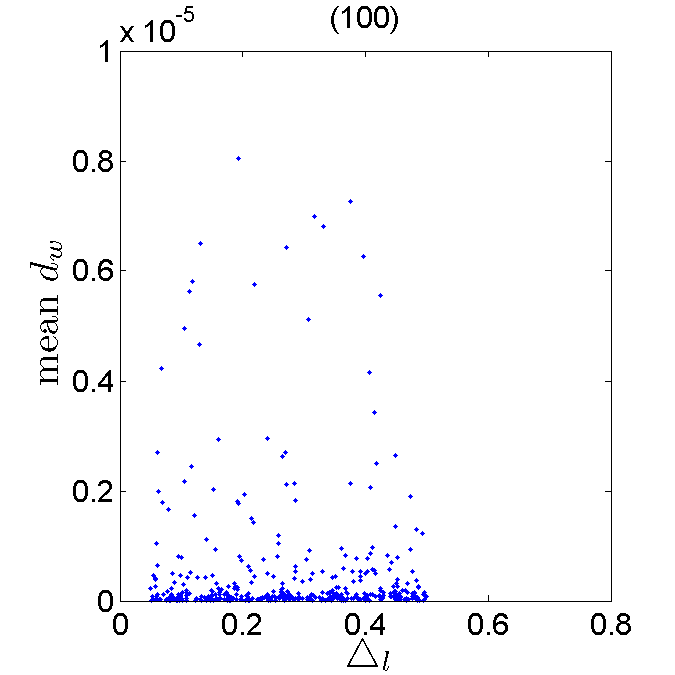} }
%
\subcaptionbox[]{$l = 3, K = 2$}[ 0.24\textwidth ]
{\includegraphics[width=0.24\textwidth]{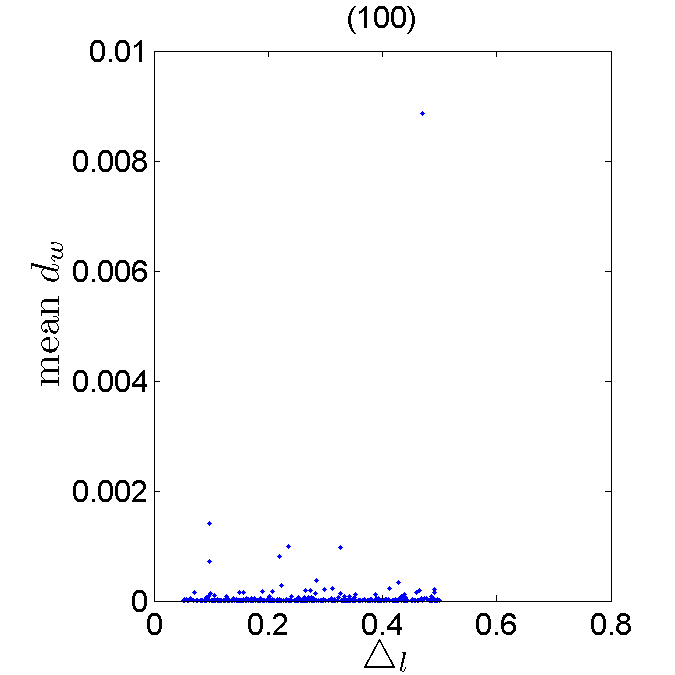} }
%
\subcaptionbox[]{$l = 4, K = 2$}[ 0.24\textwidth ]
{\includegraphics[width=0.24\textwidth]{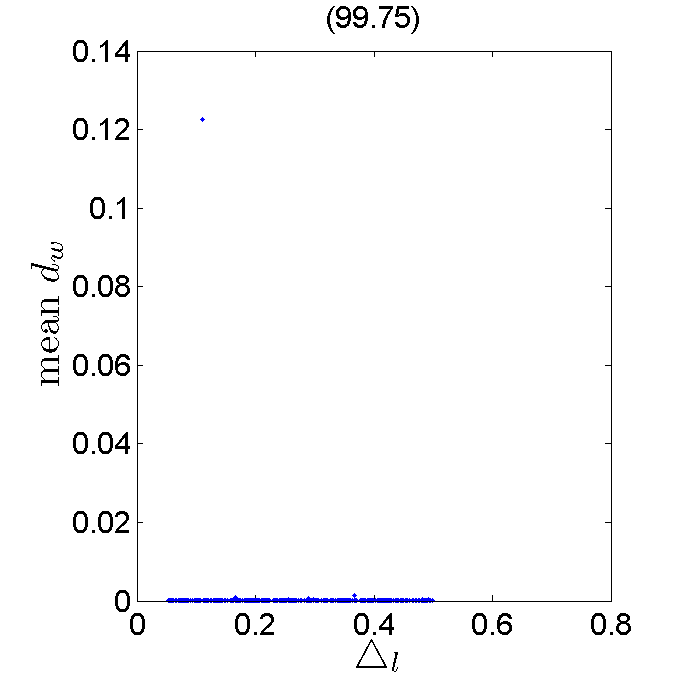} }
%

\subcaptionbox[]{$l = 1, K = 3$}[ 0.24\textwidth ]
{\includegraphics[width=0.24\textwidth]{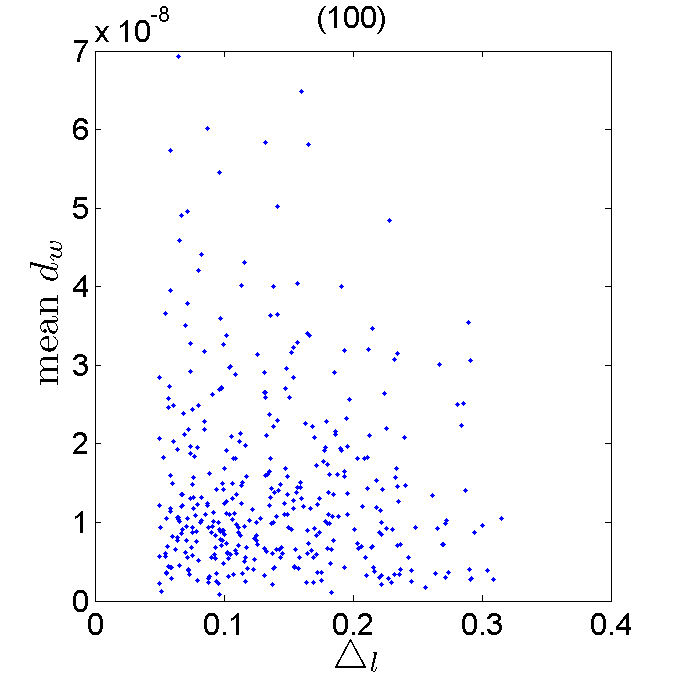} }
\subcaptionbox[]{$l = 2, K = 3$}[ 0.24\textwidth ]
{\includegraphics[width=0.24\textwidth]{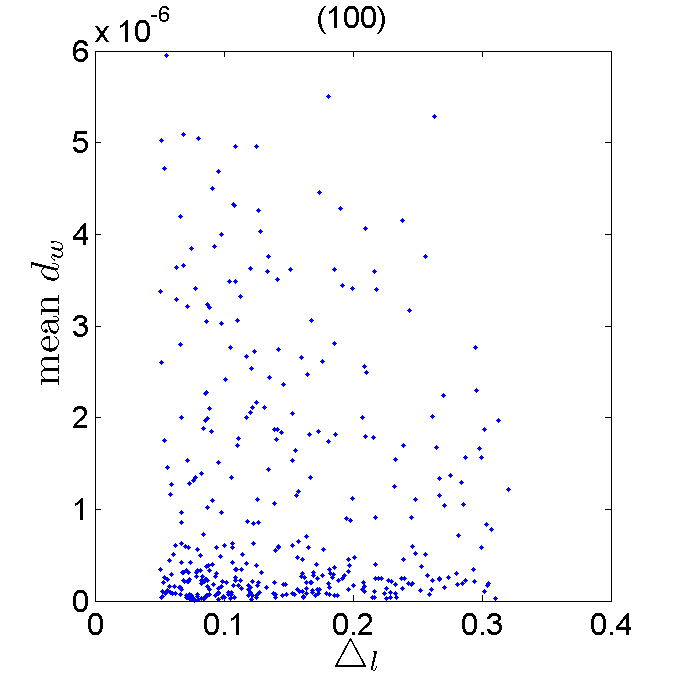} }
%
\subcaptionbox[]{$l = 3, K = 3$}[ 0.24\textwidth ]
{\includegraphics[width=0.24\textwidth]{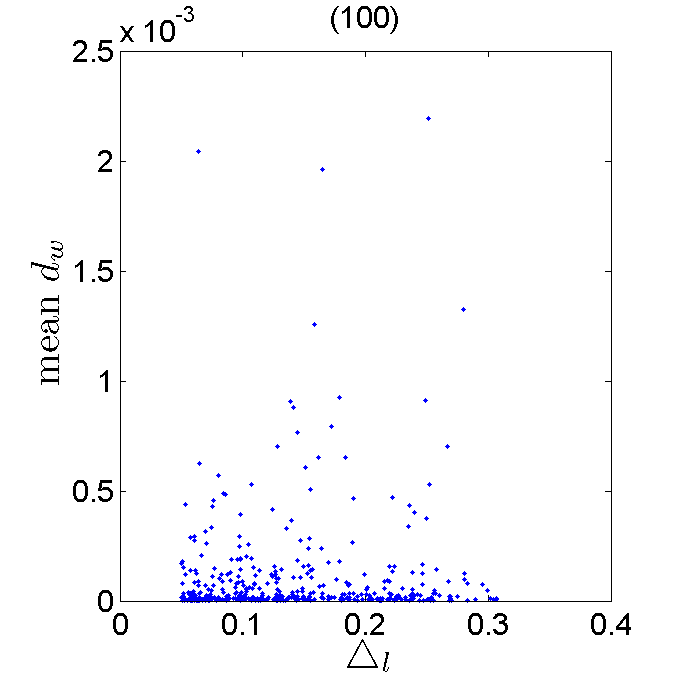} }
%
\subcaptionbox[]{$l = 4, K = 3$}[ 0.24\textwidth ]
{\includegraphics[width=0.24\textwidth]{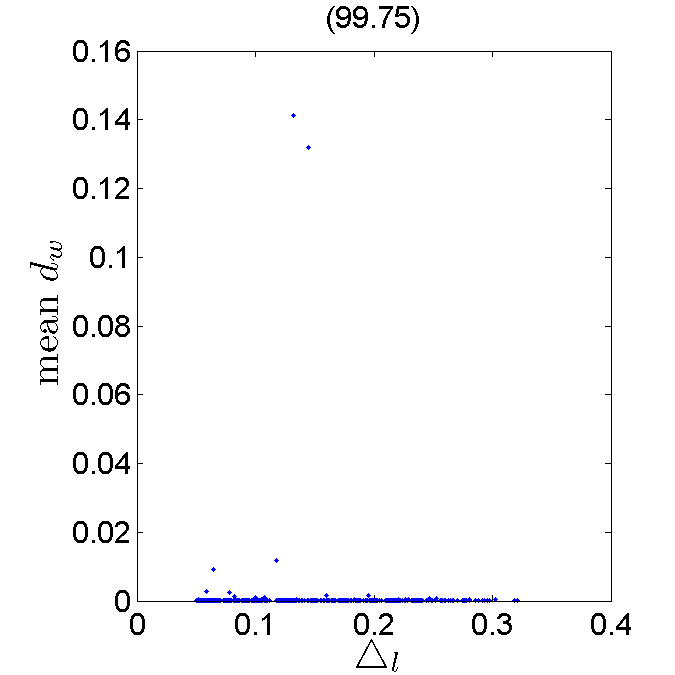} }
%

\subcaptionbox[]{$l = 1, K = 4$}[ 0.24\textwidth ]
{\includegraphics[width=0.24\textwidth]{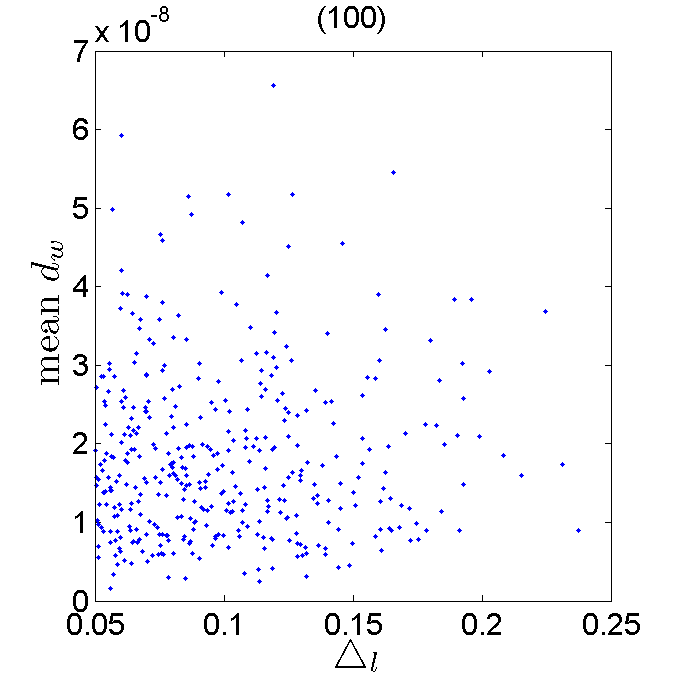} }
\subcaptionbox[]{$l = 2, K = 4$}[ 0.24\textwidth ]
{\includegraphics[width=0.24\textwidth]{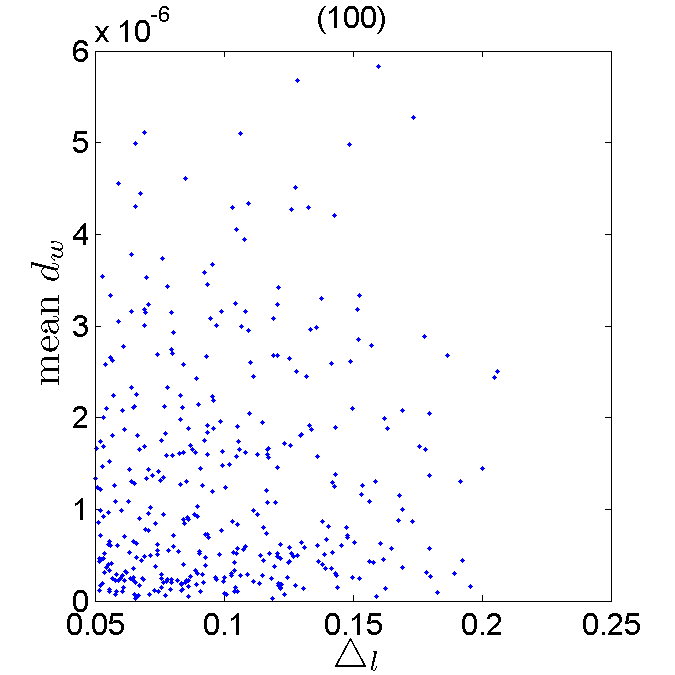} }
%
\subcaptionbox[]{$l = 3, K = 4$}[ 0.24\textwidth ]
{\includegraphics[width=0.24\textwidth]{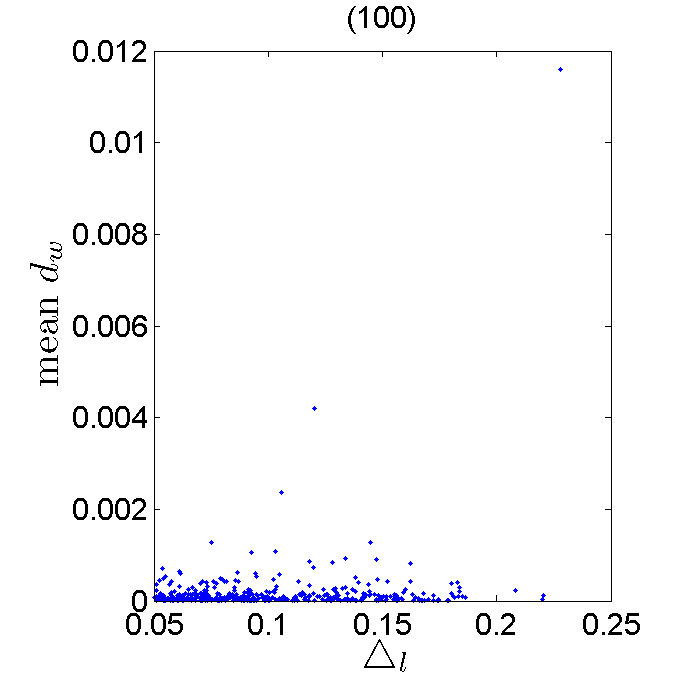} }
%
\subcaptionbox[]{$l = 4, K = 4$}[ 0.24\textwidth ]
{\includegraphics[width=0.24\textwidth]{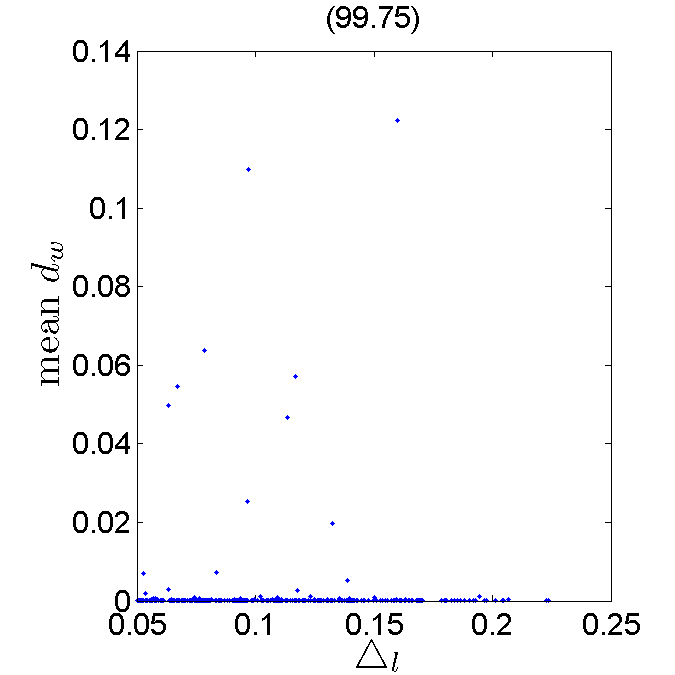} }
%

\subcaptionbox[]{$l = 1, K = 5$}[ 0.24\textwidth ]
{\includegraphics[width=0.24\textwidth]{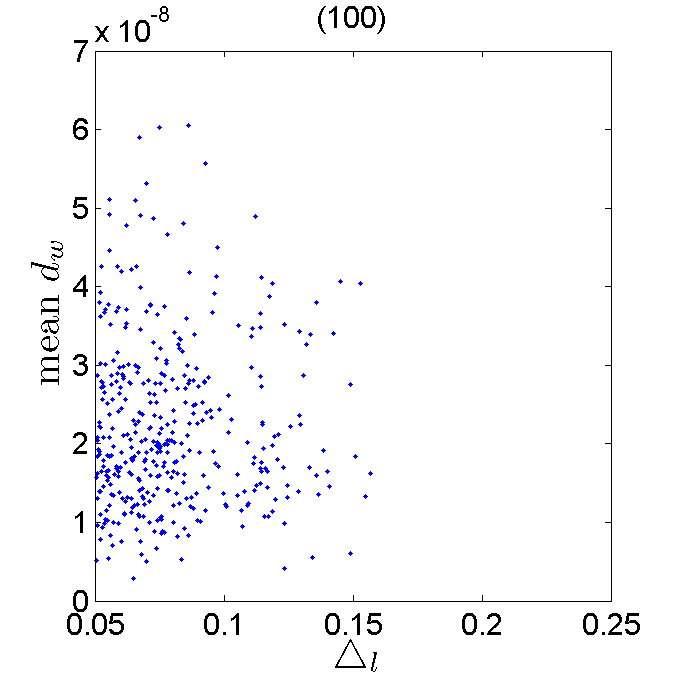} }
\subcaptionbox[]{$l = 2, K = 5$}[ 0.24\textwidth ]
{\includegraphics[width=0.24\textwidth]{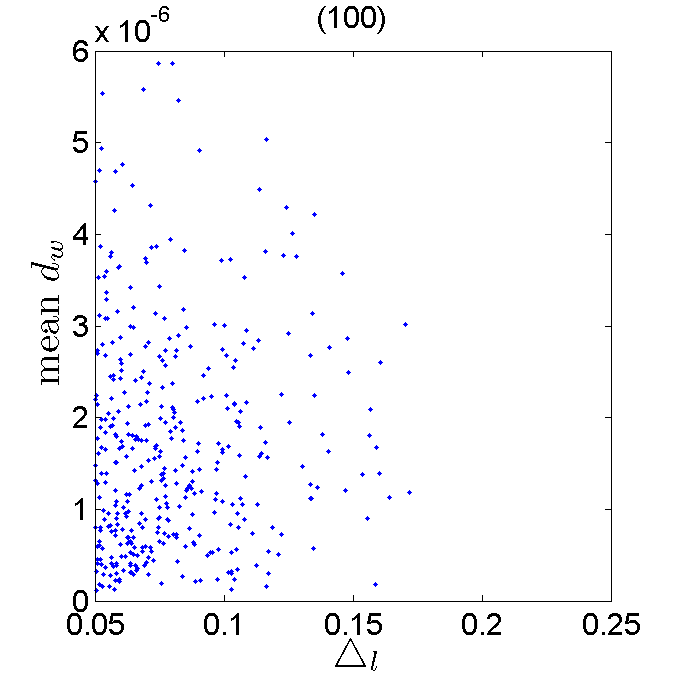} }
%
\subcaptionbox[]{$l = 3, K = 5$}[ 0.24\textwidth ]
{\includegraphics[width=0.24\textwidth]{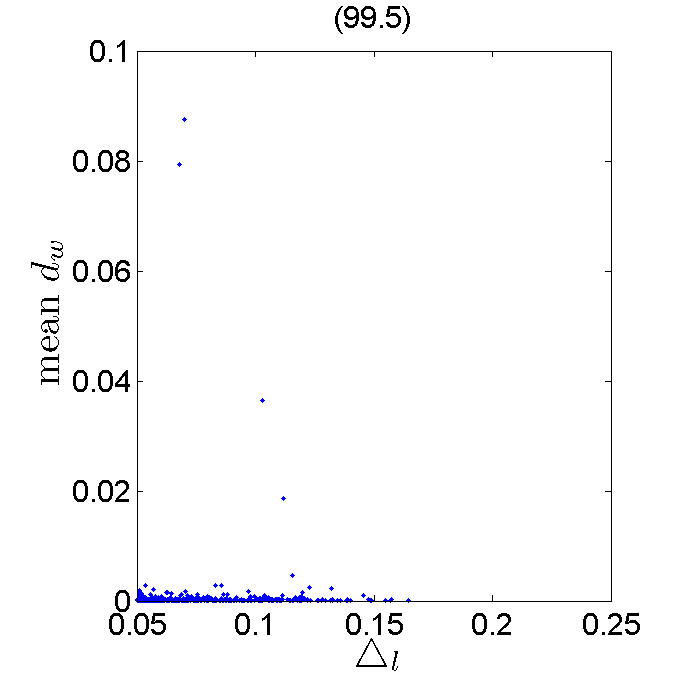} }
%
\subcaptionbox[]{$l = 4, K = 5$}[ 0.24\textwidth ]
{\includegraphics[width=0.24\textwidth]{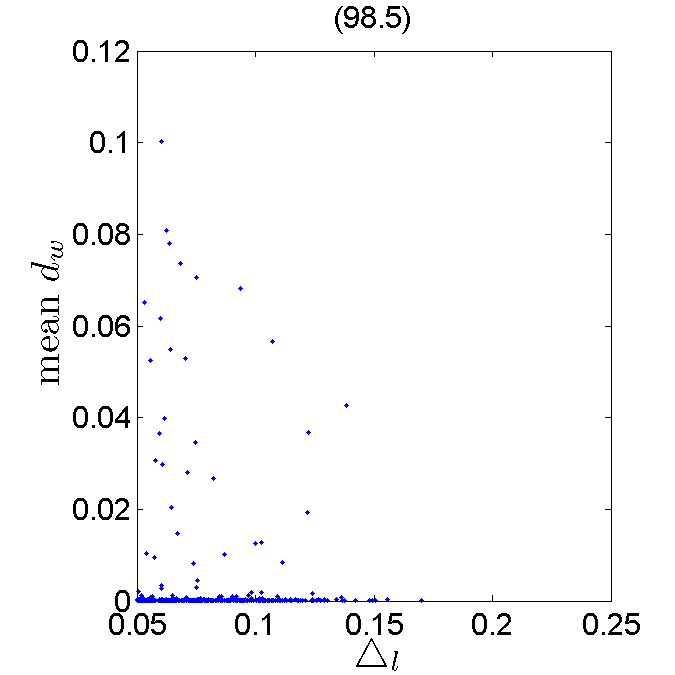} }
\captionsetup{width=0.98\linewidth}
\caption[Short Caption]{Scatter plots for the mean wrap around error ($d_{w,l,\text{avg}}$) v/s minimum separation ($\sep_l$) 
for $400$ Monte Carlo trials, with no external noise. This is shown for $K \in \set{2,3,4,5}$ with $L = 4$ and $C = 1$.
For each sub-plot, we mention the percentage of trials with $d_{w,l,\text{avg}} \leq 0.05$ in parenthesis.}
\label{fig:mean_loc_err_noiseless_C1}
\end{figure}

%
\begin{figure}[!ht]
\centering
\subcaptionbox[]{$l = 1, K = 2$}[ 0.24\textwidth ]
{\includegraphics[width=0.24\textwidth]{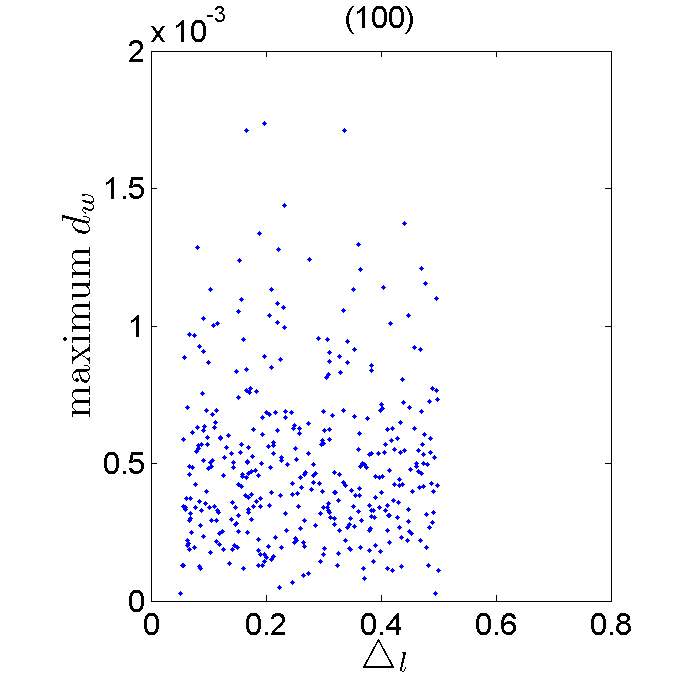} }
\subcaptionbox[]{$l = 2, K = 2$}[ 0.24\textwidth ]
{\includegraphics[width=0.24\textwidth]{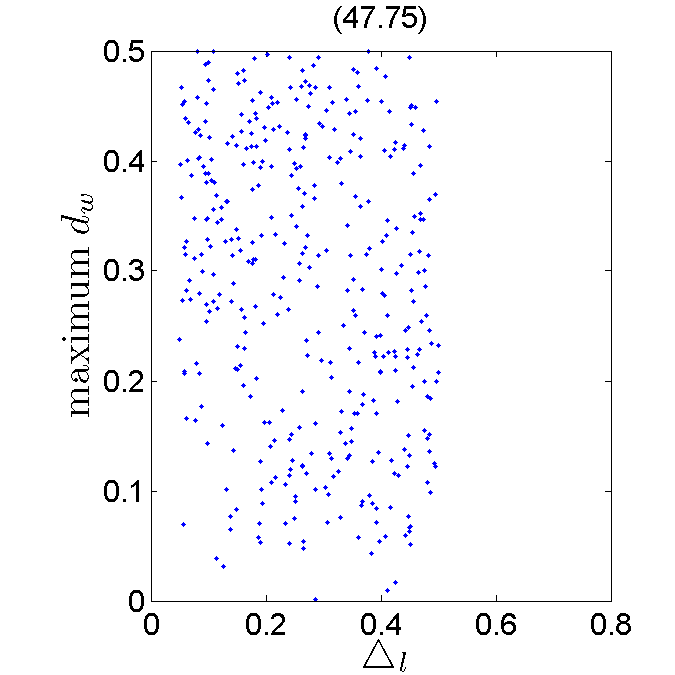} }
%
\subcaptionbox[]{$l = 3, K = 2$}[ 0.24\textwidth ]
{\includegraphics[width=0.24\textwidth]{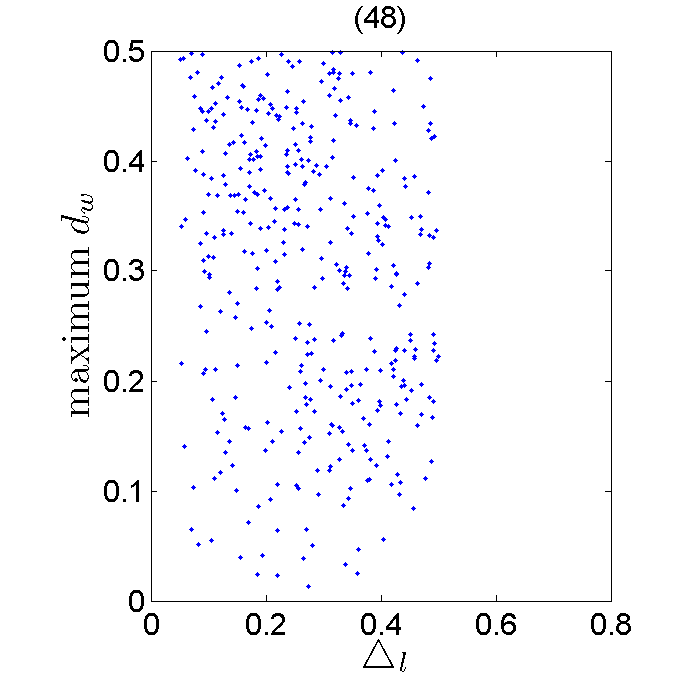} }
%
\subcaptionbox[]{$l = 4, K = 2$}[ 0.24\textwidth ]
{\includegraphics[width=0.24\textwidth]{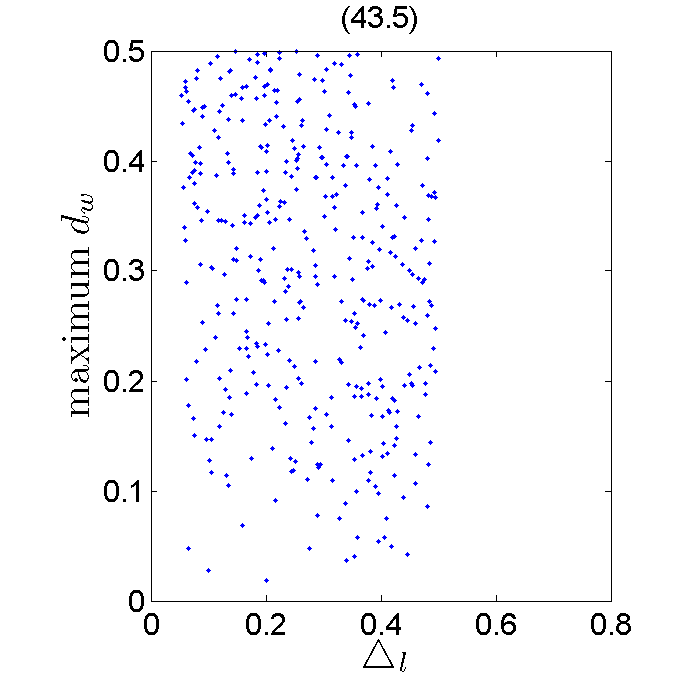} }
%

\subcaptionbox[]{$l = 1, K = 3$}[ 0.24\textwidth ]
{\includegraphics[width=0.24\textwidth]{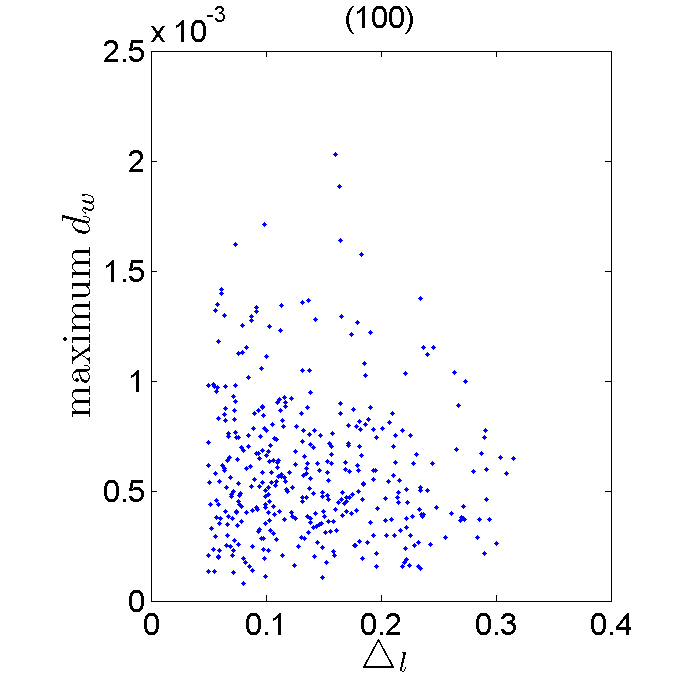} }
\subcaptionbox[]{$l = 2, K = 3$}[ 0.24\textwidth ]
{\includegraphics[width=0.24\textwidth]{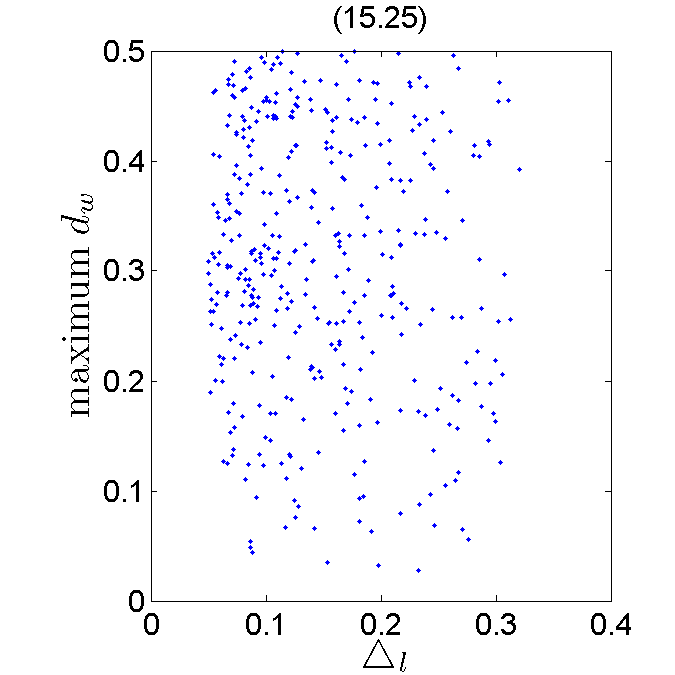} }
%
\subcaptionbox[]{$l = 3, K = 3$}[ 0.24\textwidth ]
{\includegraphics[width=0.24\textwidth]{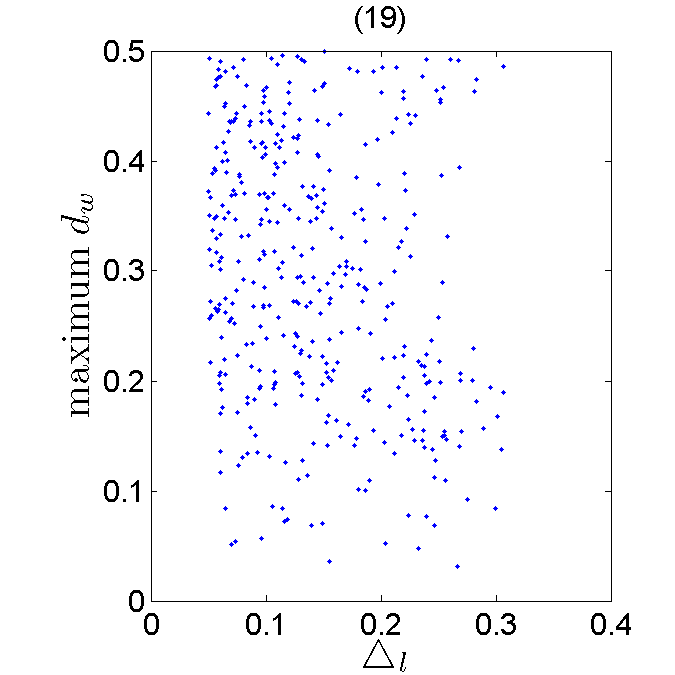} }
%
\subcaptionbox[]{$l = 4, K = 3$}[ 0.24\textwidth ]
{\includegraphics[width=0.24\textwidth]{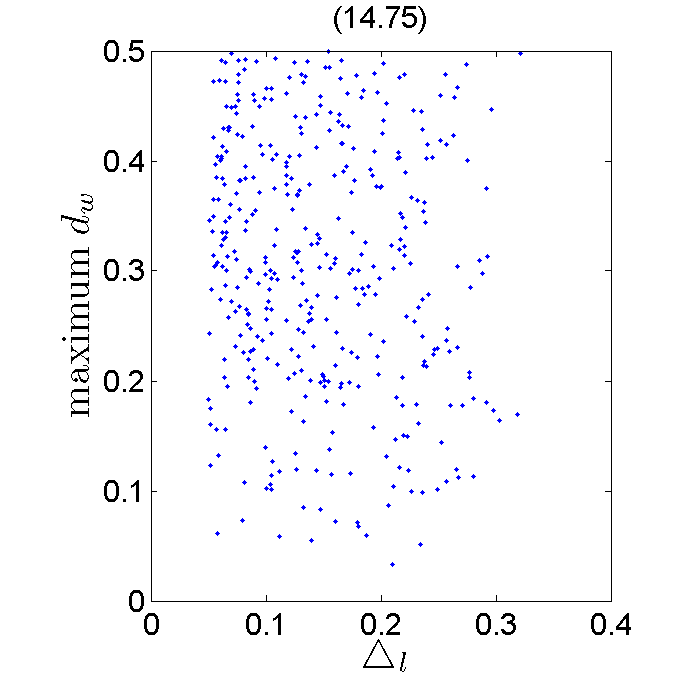} }
%

\subcaptionbox[]{$l = 1, K = 4$}[ 0.24\textwidth ]
{\includegraphics[width=0.24\textwidth]{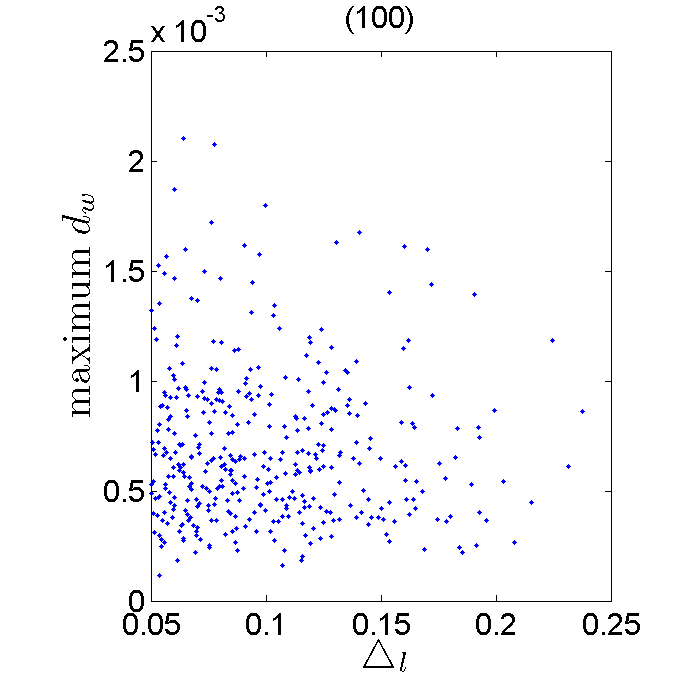} }
\subcaptionbox[]{$l = 2, K = 4$}[ 0.24\textwidth ]
{\includegraphics[width=0.24\textwidth]{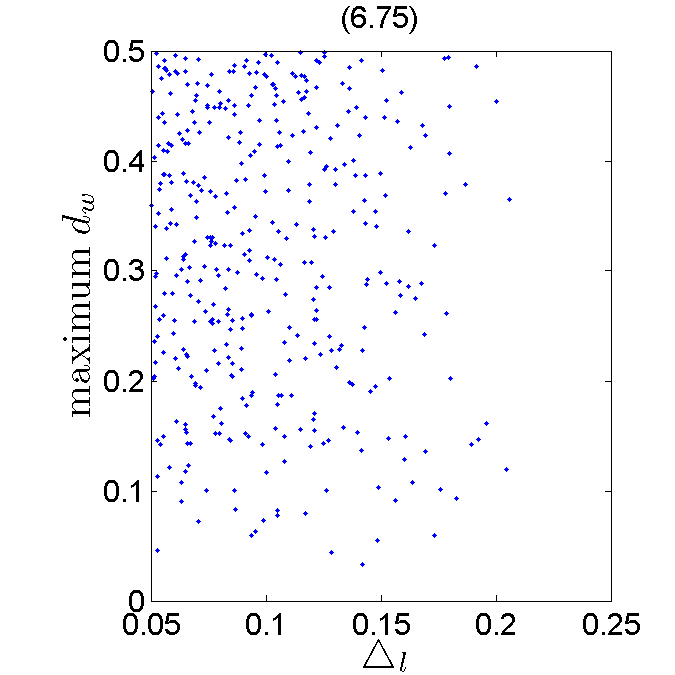} }
%
\subcaptionbox[]{$l = 3, K = 4$}[ 0.24\textwidth ]
{\includegraphics[width=0.24\textwidth]{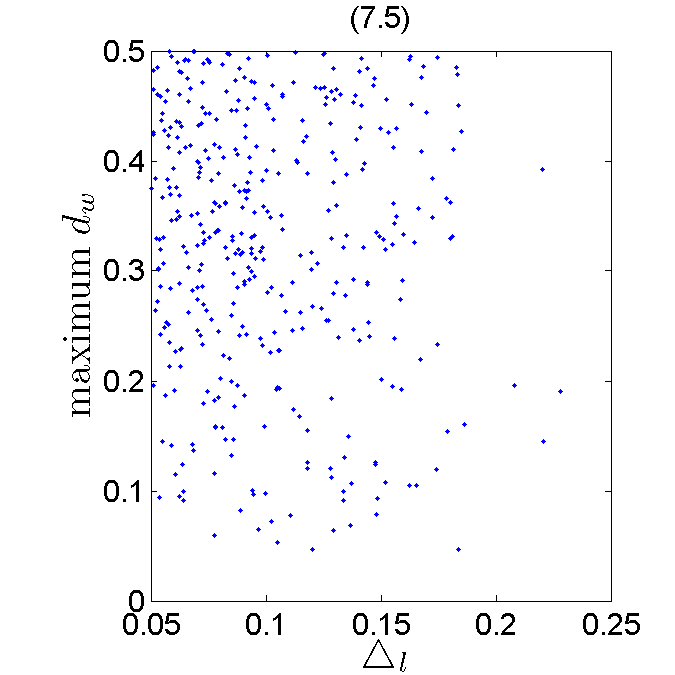} }
%
\subcaptionbox[]{$l = 4, K = 4$}[ 0.24\textwidth ]
{\includegraphics[width=0.24\textwidth]{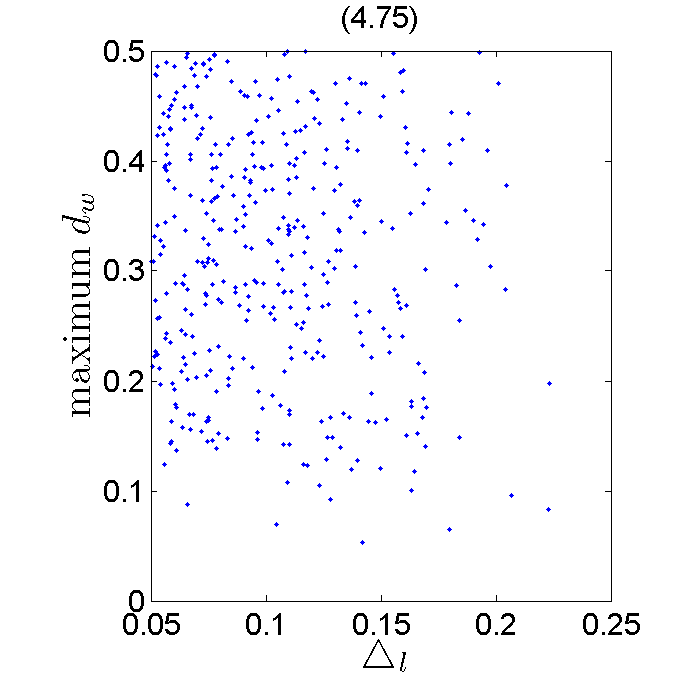} }
%

\subcaptionbox[]{$l = 1, K = 5$}[ 0.24\textwidth ]
{\includegraphics[width=0.24\textwidth]{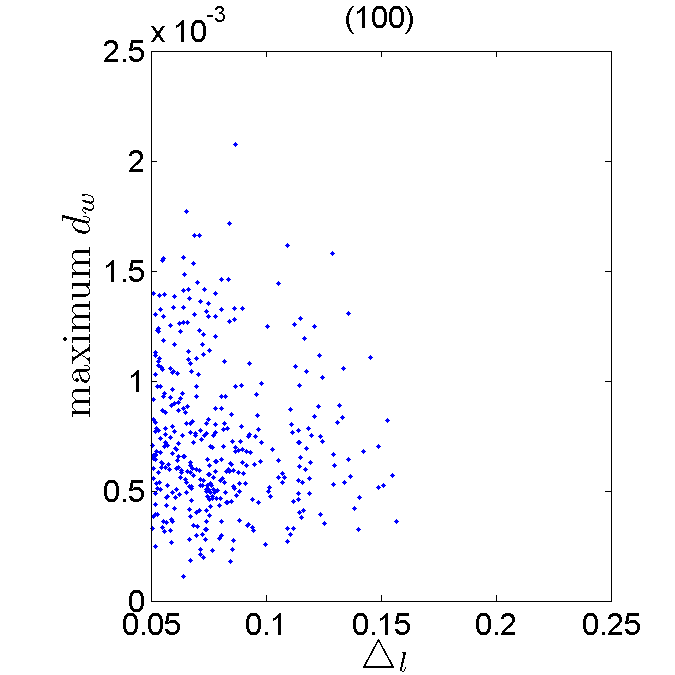} }
\subcaptionbox[]{$l = 2, K = 5$}[ 0.24\textwidth ]
{\includegraphics[width=0.24\textwidth]{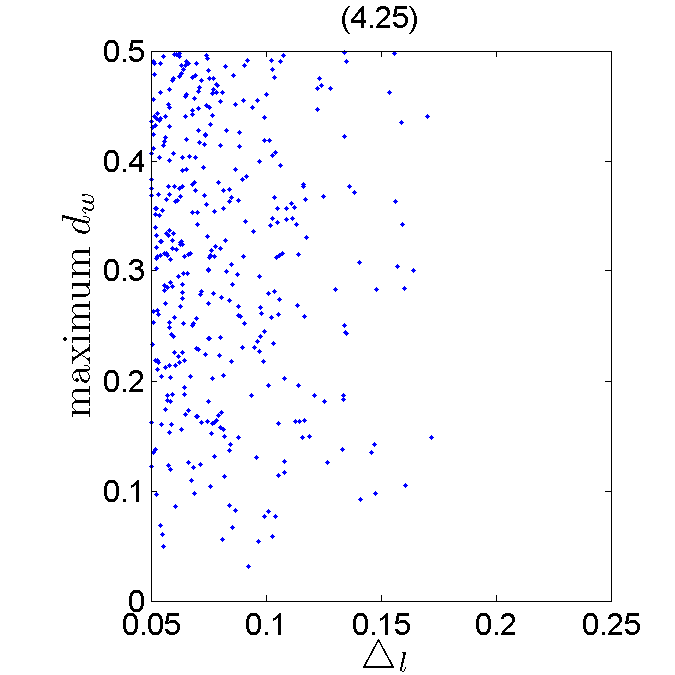} }
%
\subcaptionbox[]{$l = 3, K = 5$}[ 0.24\textwidth ]
{\includegraphics[width=0.24\textwidth]{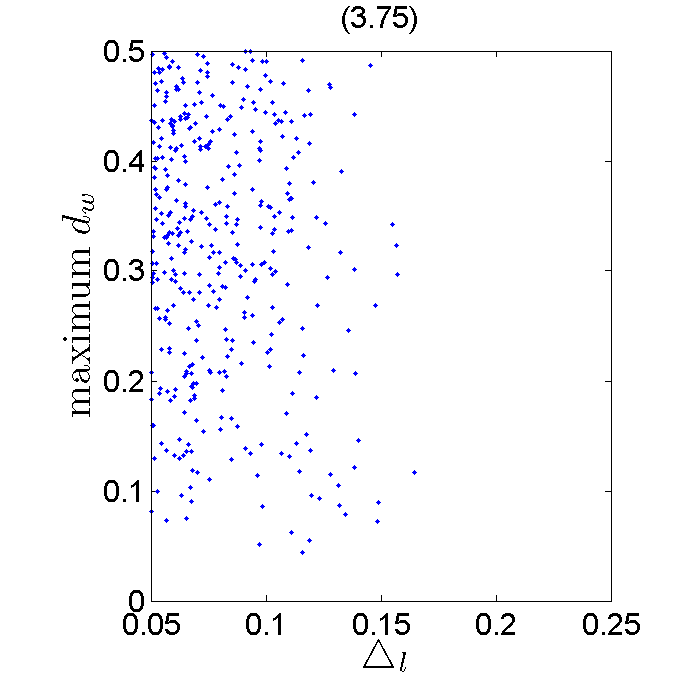} }
%
\subcaptionbox[]{$l = 4, K = 5$}[ 0.24\textwidth ]
{\includegraphics[width=0.24\textwidth]{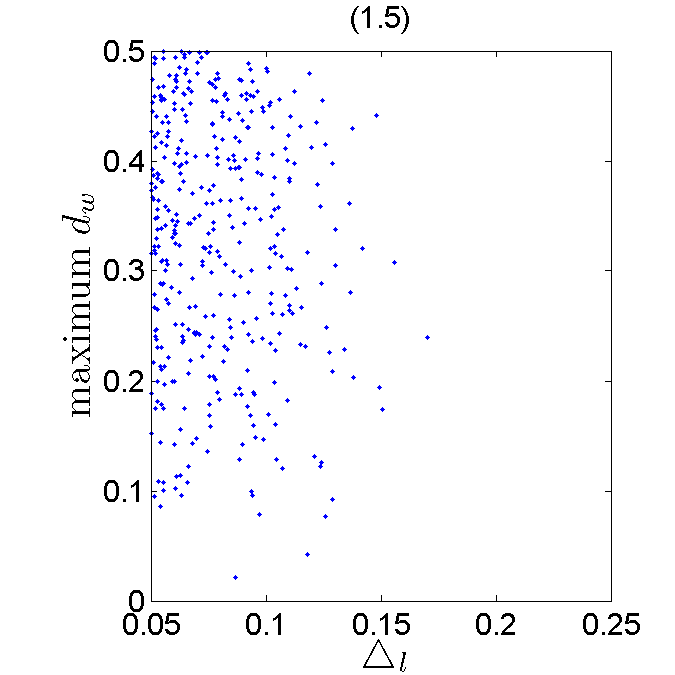} }
\captionsetup{width=0.98\linewidth}
\caption[Short Caption]{Scatter plots for maximum wrap around error ($d_{w,l,\max}$) v/s 
minimum separation ($\sep_l$) for $400$ Monte Carlo trials, 
with external Gaussian noise (standard deviation $5 \times 10^{-5}$). This is shown for $K \in \set{2,3,4,5}$ 
with $L = 4$ and $C = 1$. For each sub-plot, we mention the percentage of trials with 
$d_{w,l,\max} \leq 0.05$ in parenthesis.}
\label{fig:max_loc_err_noise_C1}
\end{figure}

%
\begin{figure}[!ht]
\centering
\subcaptionbox[]{$l = 1, K = 2$}[ 0.24\textwidth ]
{\includegraphics[width=0.24\textwidth]{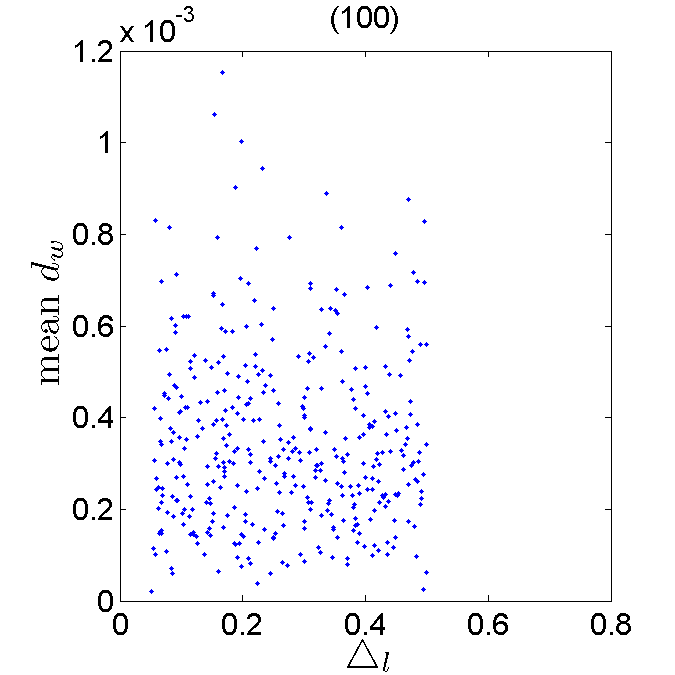} }
\subcaptionbox[]{$l = 2, K = 2$}[ 0.24\textwidth ]
{\includegraphics[width=0.24\textwidth]{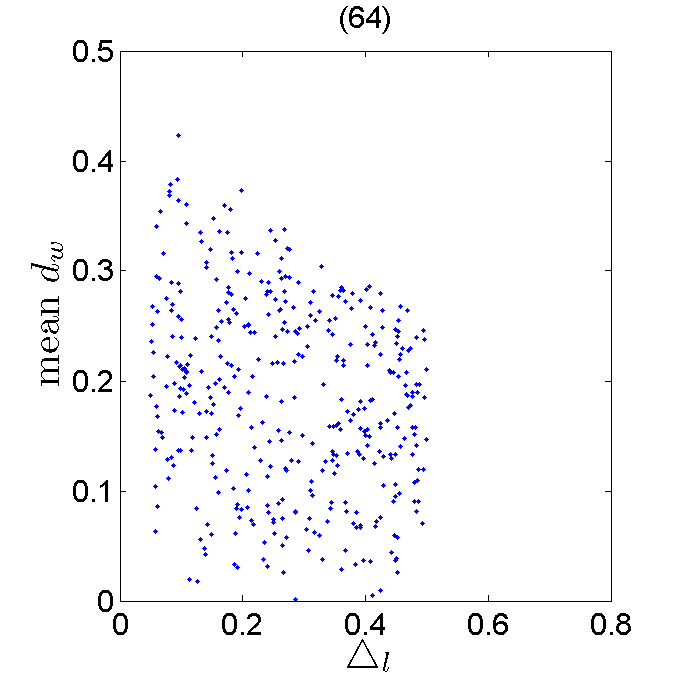} }
%
\subcaptionbox[]{$l = 3, K = 2$}[ 0.24\textwidth ]
{\includegraphics[width=0.24\textwidth]{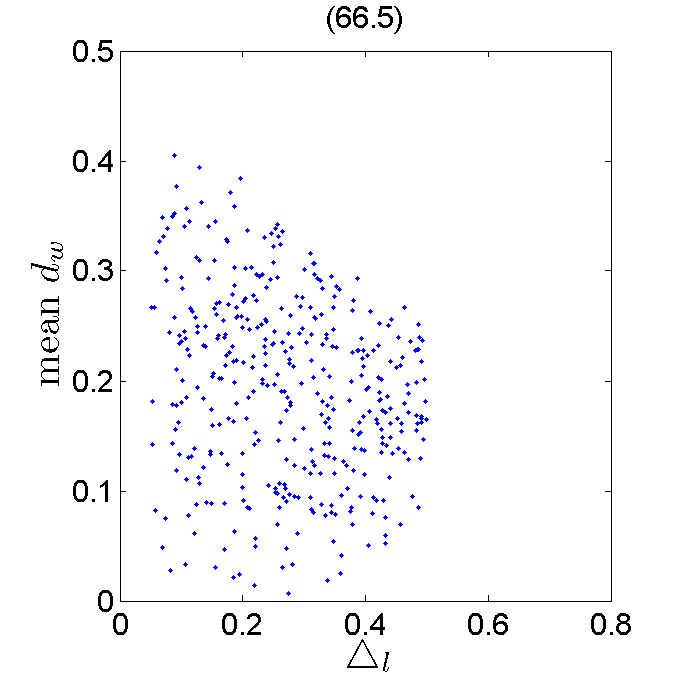} }
%
\subcaptionbox[]{$l = 4, K = 2$}[ 0.24\textwidth ]
{\includegraphics[width=0.24\textwidth]{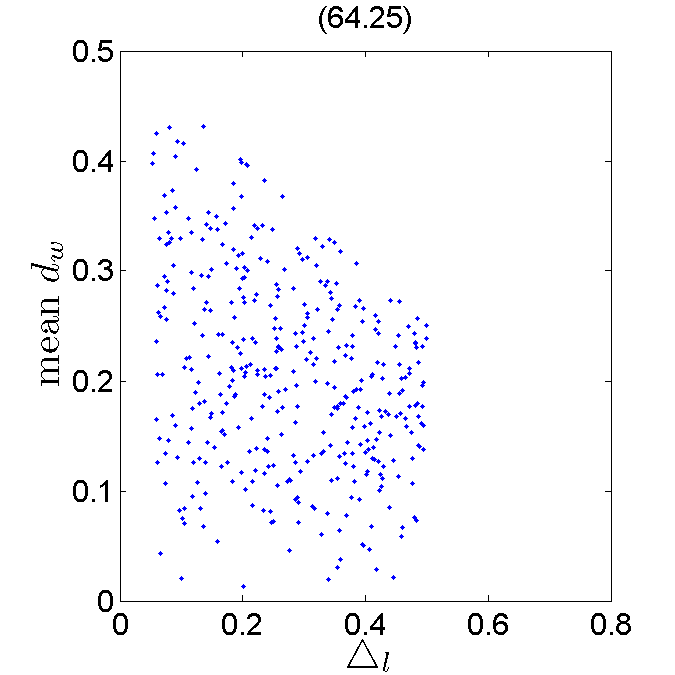} }
%

\subcaptionbox[]{$l = 1, K = 3$}[ 0.24\textwidth ]
{\includegraphics[width=0.24\textwidth]{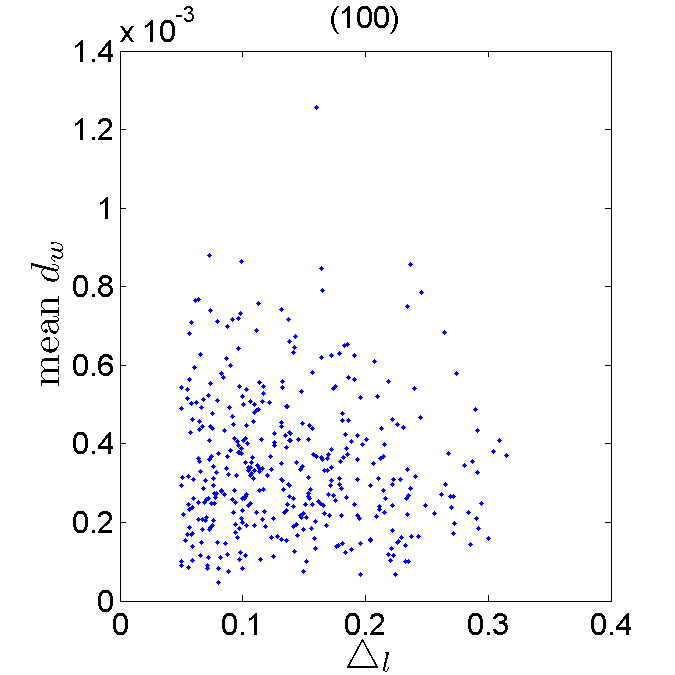} }
\subcaptionbox[]{$l = 2, K = 3$}[ 0.24\textwidth ]
{\includegraphics[width=0.24\textwidth]{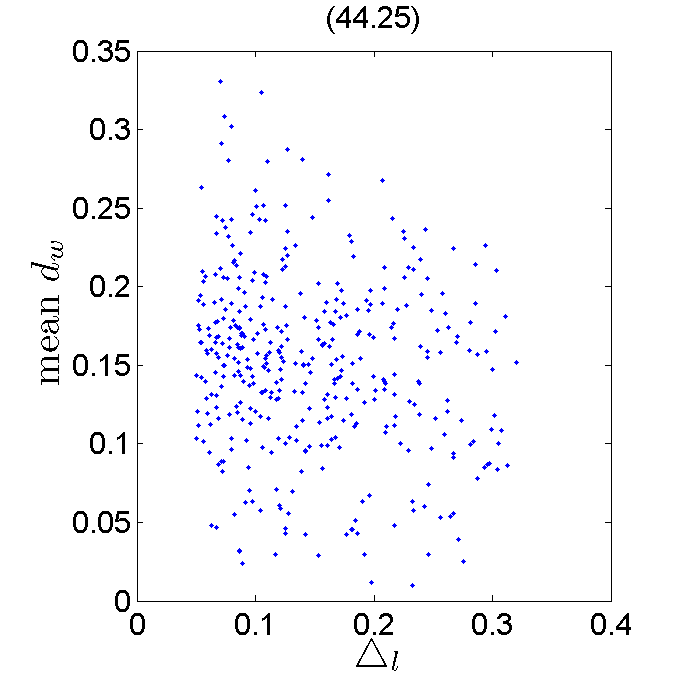} }
%
\subcaptionbox[]{$l = 3, K = 3$}[ 0.24\textwidth ]
{\includegraphics[width=0.24\textwidth]{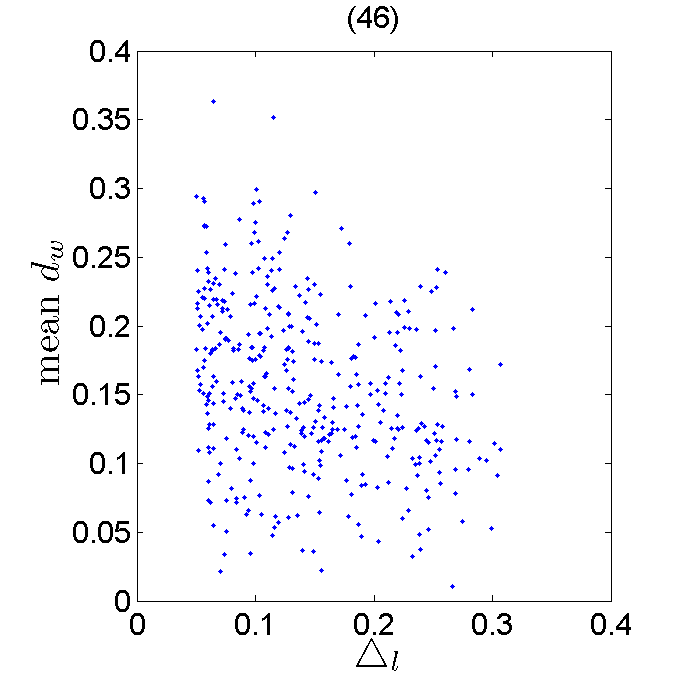} }
%
\subcaptionbox[]{$l = 4, K = 3$}[ 0.24\textwidth ]
{\includegraphics[width=0.24\textwidth]{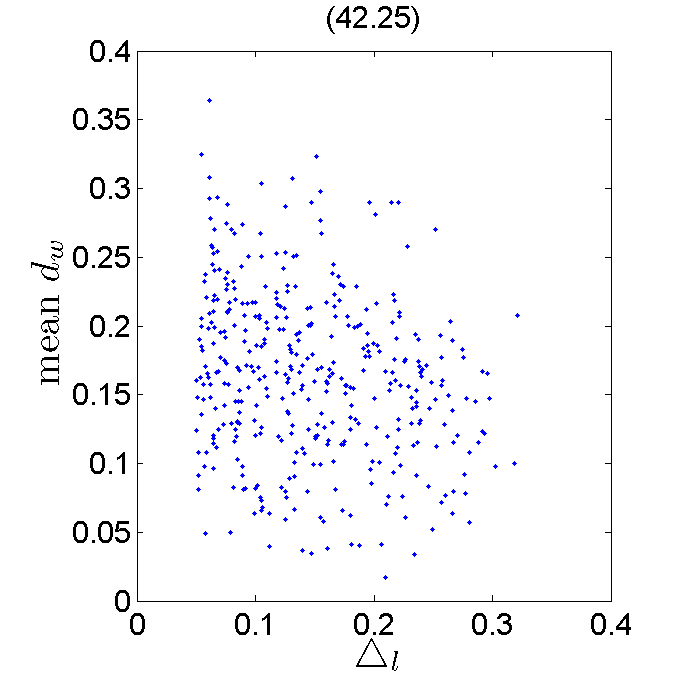} }
%

\subcaptionbox[]{$l = 1, K = 4$}[ 0.24\textwidth ]
{\includegraphics[width=0.24\textwidth]{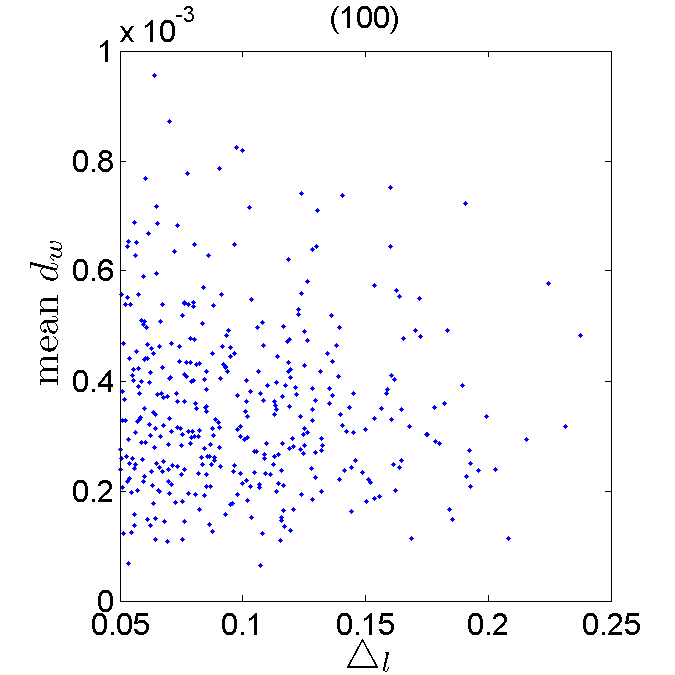} }
\subcaptionbox[]{$l = 2, K = 4$}[ 0.24\textwidth ]
{\includegraphics[width=0.24\textwidth]{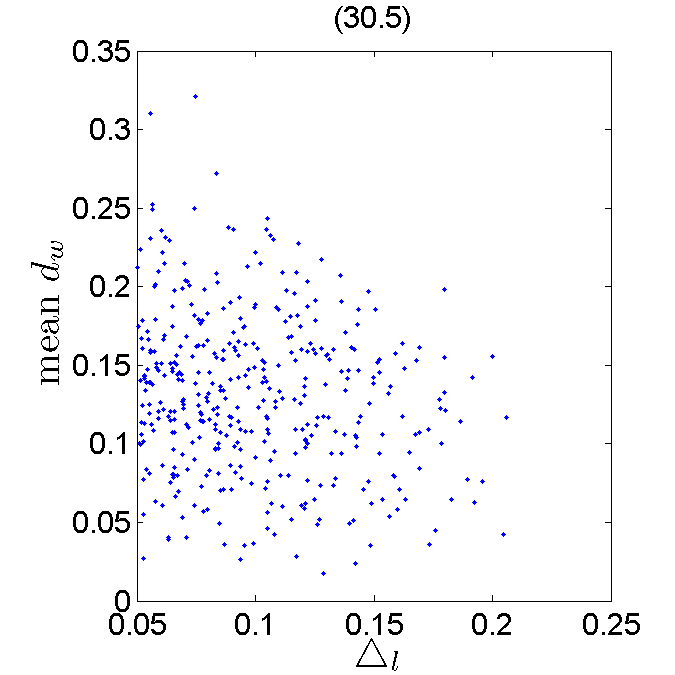} }
%
\subcaptionbox[]{$l = 3, K = 4$}[ 0.24\textwidth ]
{\includegraphics[width=0.24\textwidth]{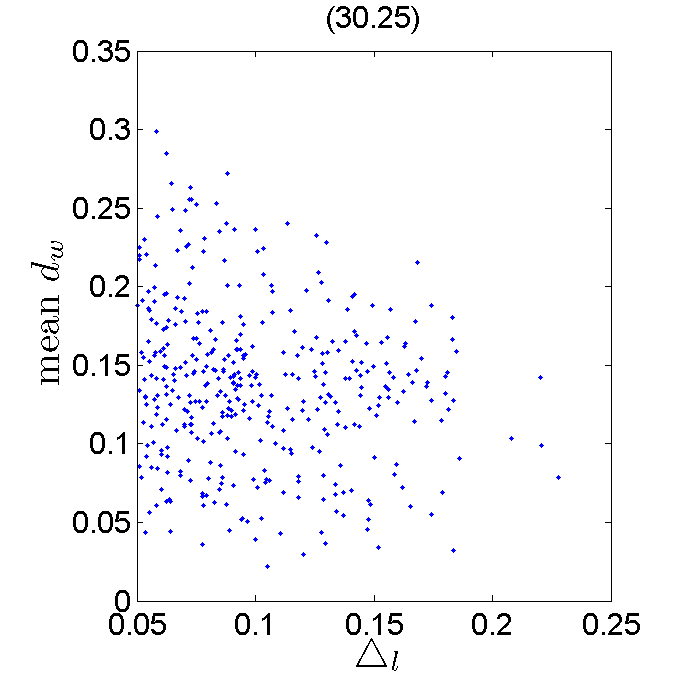} }
%
\subcaptionbox[]{$l = 4, K = 4$}[ 0.24\textwidth ]
{\includegraphics[width=0.24\textwidth]{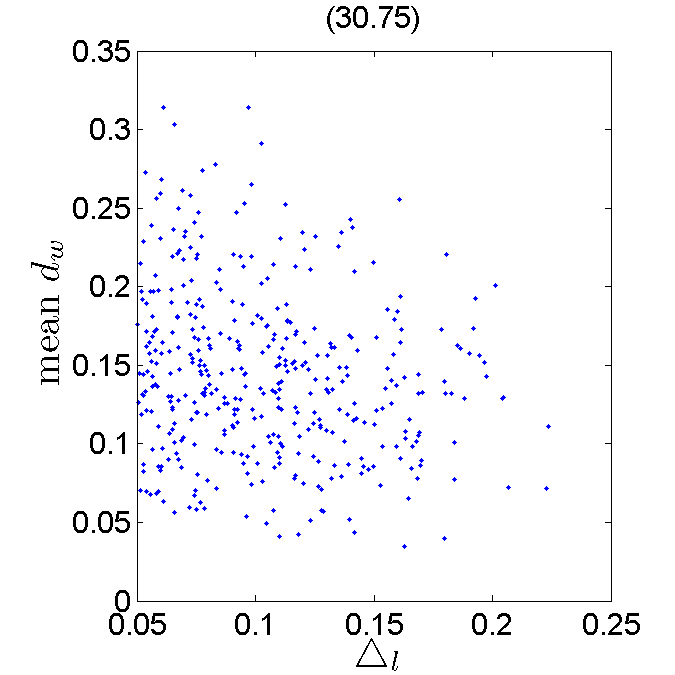} }
%

\subcaptionbox[]{$l = 1, K = 5$}[ 0.24\textwidth ]
{\includegraphics[width=0.24\textwidth]{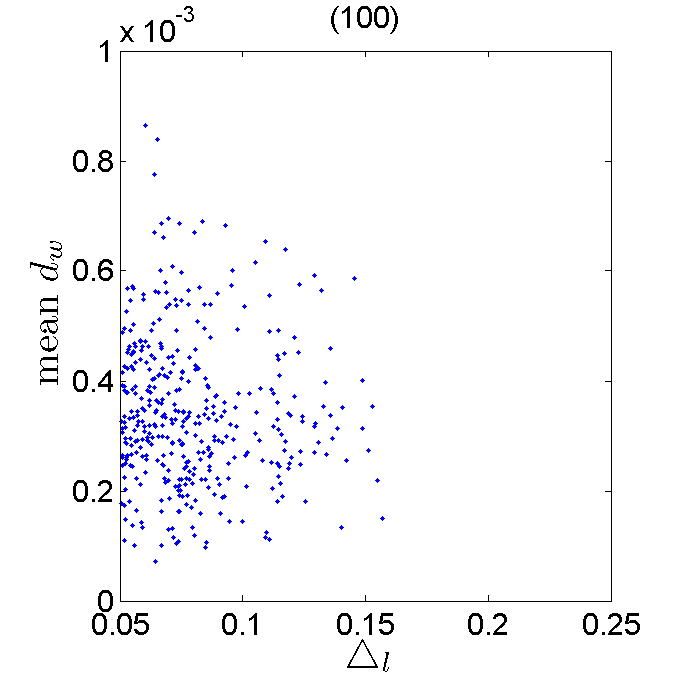} }
\subcaptionbox[]{$l = 2, K = 5$}[ 0.24\textwidth ]
{\includegraphics[width=0.24\textwidth]{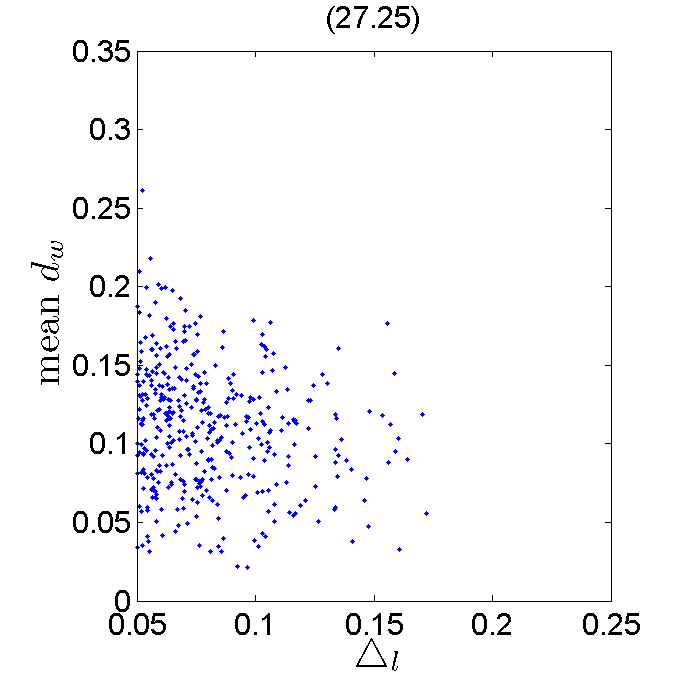} }
%
\subcaptionbox[]{$l = 3, K = 5$}[ 0.24\textwidth ]
{\includegraphics[width=0.24\textwidth]{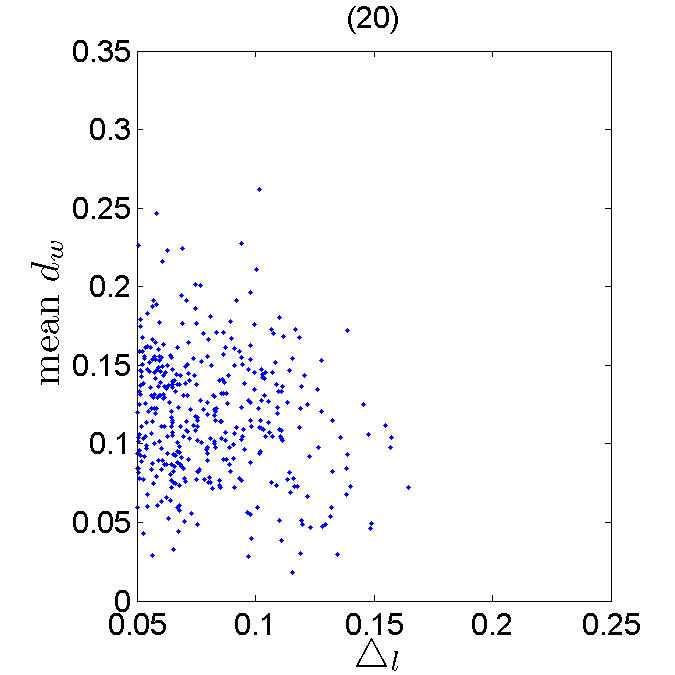} }
%
\subcaptionbox[]{$l = 4, K = 5$}[ 0.24\textwidth ]
{\includegraphics[width=0.24\textwidth]{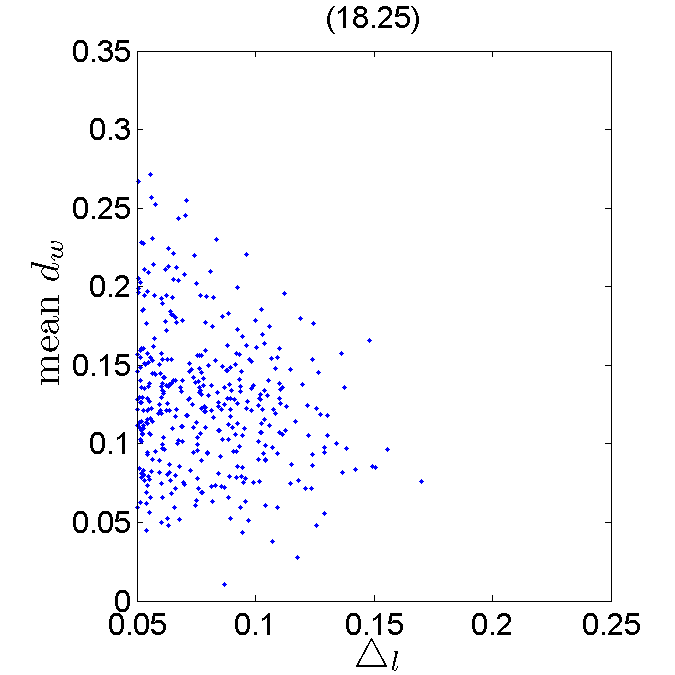} }
\captionsetup{width=0.98\linewidth}
\caption[Short Caption]{Scatter plots for mean wrap around error ($d_{w,l,\text{avg}}$) v/s minimum separation ($\sep_l$) 
for $400$ Monte Carlo trials, with external Gaussian noise (standard deviation $5 \times 10^{-5}$). 
This is shown for $K \in \set{2,3,4,5}$ with $L = 4$ and $C = 1$. For each sub-plot, we mention the percentage of trials with 
$d_{w,l,\text{avg}} \leq 0.05$ in parenthesis.}
\label{fig:mean_loc_err_noise_C1}
\end{figure}

\end{document}